\documentclass[letterpaper,journal]{IEEEtran}
\usepackage{mathptmx} 
\usepackage{times} 
\usepackage{amsmath} 
\usepackage{amsbsy} 
\usepackage{amssymb} 
\usepackage{mathrsfs}
\usepackage{comment}
\usepackage{xfrac}
\usepackage[export]{adjustbox}
\usepackage{tikz}
\usetikzlibrary{external,positioning,decorations.pathreplacing,shapes,arrows,patterns}
\usepackage[outline]{contour}

\usepackage{pgfplots}
\usepackage{graphicx}
\usepackage{pstool}
\usepackage[latin1]{inputenc}
\usepackage{xifthen}
\usepackage{epic}
\usepackage{subcaption}
\usepackage{caption}
\usepackage[outdir=./]{epstopdf}

\newtheorem{thm}{\bf{Theorem}}
\newtheorem{prop}[thm]{\bf {Proposition}}
\newtheorem{cor}[thm]{\bf {Corollary}}

\newtheorem{example}{\bf {Example}}
\newtheorem{definition}{\bf {Definition}}

\newcommand{\mmse}{\mathsf{mmse}}
\newcommand{\Nyq}{\mathsf{Nyq}}
\newcommand{\Lnd}{\mathsf{Lnd}}
\newcommand{\PCM}{\mathsf{PCM}}

\newcommand{\SNR}{\mathsf{snr}}
\newcommand{\expphi}{\left(e^{2\pi i \phi} \right)}
\newcommand{\supp}{\mathrm{supp} }
\newcommand{\card}{\mathrm{card} }
\newcommand{\Pw}{\mathrm{Pw} }
\newcommand{\tr}{\mathrm{trace} }
\newcommand{\inthalftohalf}{\int_{-\frac{1}{2}}^\frac{1}{2} }
\newcommand{\Yv}{\mathbf Y}

\newcommand{\Ell}{\mathrm{L}_2 }
\newcommand{\enc}{\mathrm{Enc}}
\newcommand{\dec}{\mathrm{Dec}}

\newenvironment{proof}{\paragraph*{Proof}}{\hfill$\square$ \newline}

\tikzstyle{int1}=[draw, fill=blue!10, minimum height = 0.5cm, minimum width=1cm,thick ]
\tikzstyle{int}=[draw, fill=blue!10, minimum height = 1cm, minimum width=1.5cm,thick ]
\tikzstyle{sint}=[draw, fill=blue!10, minimum height = 0.5cm, minimum width=0.8cm,thick ]
\tikzstyle{sum}=[circle, fill=blue!10, draw=black,line width=1pt,minimum size = 0.5cm, thick ]
\tikzstyle{ssum}=[circle, fill=blue!10,draw=black,line width=1pt,minimum size = 0.1cm]

\tikzstyle{enc}=[draw, fill=blue!10, minimum height = 2.7cm, minimum width=1cm,thick ]

\title{\LARGE \bf Fundamental Distortion Limits of Analog-to-Digital Compression
}

\author{ 
\IEEEauthorblockN{
Alon Kipnis, Yonina C. Eldar and  Andrea J. Goldsmith}

\thanks{  A. Kipnis was with the Department of Electrical Engineering, Stanford University, Stanford, CA 94305 USA. He is now with the Department of Statistics at the same institution. 

A. J. Goldsmith are with the Department of Electrical Engineering, Stanford University, Stanford, CA 94305 USA. 

Y. C. Eldar is with the Department of Electrical Engineering, Technion - Israel Institute of Technology Haifa 32000, Israel.}

\thanks{This paper was presented in part at the 52nd and 53rd Annual Allerton Conference on Communication, Control, and Computing (Allerton), October 2014 and 2015 \cite{KipnisAllerton2014} \cite{KipnisAllerton2015}, and at the Information Theory Workshop (ITW), April 2015, Jerusalem \cite{KipnisITW2015}. }
}

\newcommand*{\QEDA}{\hfill\ensuremath{\square}}

\pgfplotsset{compat=1.14}
\begin{document}
\graphicspath{{./Figs/}}
\maketitle

\thispagestyle{plain}
\pagestyle{plain}

%%%%%%%%%%%%%%%%%%%%%%%%%%%%%%%%%%%%%%%%%%%%%%%%%%%%%%%%%%%%%%%%%%%%%%%%%%%%%%%%
\begin{abstract}
Representing a continuous-time signal by a set of samples is a classical problem in signal processing. We study this problem under the additional constraint that the samples are quantized or compressed in a lossy manner under a limited bitrate budget. To this end, we consider a combined sampling and source coding problem in which an analog stationary Gaussian signal is reconstructed from its encoded samples. These samples are obtained by a set of bounded linear functionals of the continuous-time path, with a limitation on the average number of samples obtained per unit time available in this setting. 
We provide a full characterization of the minimal distortion in terms of the sampling frequency, the bitrate, and the signal's spectrum. 
Assuming that the signal's energy is not uniformly distributed over its spectral support, we show that for each compression bitrate there exists a critical sampling frequency smaller than the Nyquist rate, such that the distortion in signal reconstruction when sampling at this frequency is minimal. Our results can be seen as an extension of the classical sampling theorem for bandlimited random processes in the sense that it describes the minimal amount of excess distortion in the reconstruction due to lossy compression of the samples, and provides the minimal sampling frequency required in order to achieve this distortion. Finally, we compare the fundamental limits in the combined source coding and sampling problem to the performance of pulse code modulation (PCM), where each sample is quantized by a scalar quantizer using a fixed number of bits. 
\end{abstract}

%%%%%%%%%%%%%%%%%%%%%%%%%%%%%%%%%%%%%%%%%%%%%%%%%%%%%%%%%%%%%%%%%%%%%%%%%%%%%%%%
\section{INTRODUCTION}
\label{sec:Intro}
The minimal sampling rate required for perfect reconstruction of a bandlimited continuous-time process from its samples is given by the celebrated works of Whittaker, Kotelnikov, Shannon and Landau \cite{eldar2015sampling}. However, these results focus only on performance associated with sampling rates; they do not incorporate other sampling parameters, in particular the quantization precision of the samples. This work aims to develop a theory of sampling and associated fundamental performance bounds that incorporates both sampling rate as well as quantization precision. \par
%When the source is considered under a statistical model, the fundamental trade-off between bit representation and distortion in reconstruction from such representation is given by the distortion-rate function (DRF) of the analog source. A key idea in proving the relation between the DRF and finite bit-representation problem relays on mapping the analog waveform to a discrete-time process. Such mapping ignores possible limitations on resources available for data processing, such as limited sampling frequency and/or limited memory of the vector quantizer. 
The Shannon-Kotelnikov-Whittaker sampling theorem states that sampling a signal at its Nyquist rate is a sufficient condition for exact recreation of the signal from its samples. However, quoting Shannon \cite{Shannon1948}:
\begin{quote}
...``we are not interested in exact transmission when we have a continuous [amplitude] source, but only in transmission to within a certain [distortion] tolerance...". 
\end{quote}
It is in fact impossible to obtain an exact digital representation of any continuous amplitude signal due to the finite precision of the samples. Hence, any digital representation of an analog signal is prone to some error, regardless of the sampling rate. This raises the question as to whether the condition of Nyquist rate sampling can be relaxed when we are interested in converting an analog signal to bits at a given bitrate (bits per unit time), such that the associated point on the distortion-rate function (DRF) of the signal is achieved.
\par
\begin{figure}
\begin{center}
\input{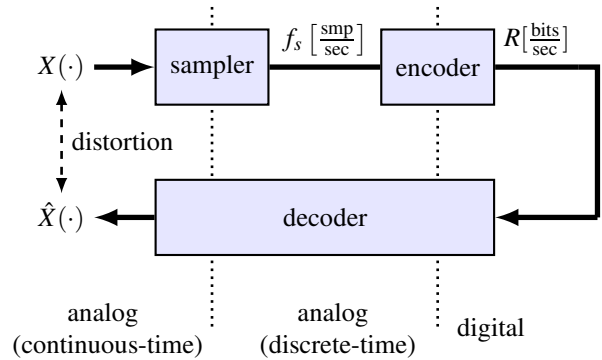}
\end{center}
\caption{\label{fig:ADX_system} Analog-to-digital compression (ADX) and reconstruction setting. Our goal is to derive the minimal distortion between the signal and its reconstruction from any bitrate $R$ representation of the samples taken at sampling rate $f_s$.}
\end{figure}
The DRF describes the minimal distortion for any digital representation of a given signal under a fixed number of bits per unit time. While this implies that the DRF provides a theoretical limit on the distortion as a result of analog to digital (A/D) conversion, in fact, A/D conversion involves both sampling a signal as well as converting those samples to bits, which entails some form of source coding, typically quantization. In some situations, it is possible to achieve the DRF of a continuous-time signal by mapping it into an equivalent discrete-time representation based on sampling at or above its Nyquist rate \cite{neuhoff2013information}. However, A/D technology limitations can preclude sampling signals at their Nyquist rate, particularly for wideband signals or under energy constraints \cite{761034, EldarMichaeliBeyond}. In such scenarios, the data available for source encoding is a sub-Nyquist sampled discrete-time representation of the signal \cite{EldarMichaeliBeyond}.
Our goal in this work is to consider the minimal distortion in recovering an analog signal from its samples with lossy compression of the samples at a prescribed bitrate, a setting which we call analog-to-digital compression (ADX) and is illustrated in Fig.~\ref{fig:ADX_system}. We are interested in particular in the optimal sampling rate to achieve this minimal distortion for a given lossy compression rate of the samples. %Understanding the sampling rate and sampling mechanism that attains the optimal lossy compression performances is useful not only in providing the optimal representation for reconstructing analog signals, but also in understanding the fundamental limits of their classification and errors in decisions based on finite-state functions of their measurements. 
\\

\begin{figure}
\begin{center}
\begin{tikzpicture}[scale=1]

\draw[->,line width=1pt]  (0,0) node[left] {$0$}--(6.2,0) node[right] {$f_s$};
\draw[->,line width=1pt]  (0,0)--(0,2.2) node[above] {$MSE$};

 \draw[blue] plot[domain=0:2, samples=100] (\x*\x/2,2-\x/2) -- plot[domain=2:6, samples=100] (\x, 1) node[above, xshift=-3cm, blue] {$D^\star(f_s,R)$};
 
\draw[red] plot[domain=0:4, samples=100] (\x*\x/4,2-\x/2) --  plot[domain=4:6, samples=100] (\x, 0);	

\draw[dashed] (2,0) node [below] {$f_{R}$}-- (2,1); 

\draw[dashed] (4,0) node [below] {$f_{\Nyq}$}-- (4,1); 

\draw[dashed] (-0.05,1) node[left] {$D_X(R)$} -- (2,1);
\draw (-0.05,1.9) node[left] {$\sigma_X^2$} -- (0.1,1.9);
\node [rotate=-18.5, red] at  (1.2,0.6) (Sx_or) {$\mmse^\star(f_s)$};
\end{tikzpicture}
\end{center}
\caption{\label{fig:contribution} Minimal distortion versus sampling rate. $D_X(R)$ is the information DRF describing the minimal distortion using lossy compression at bitrate $R$. $D^\star(R,f_s)$ is the minimal distortion using sampling at frequency $f_s$ followed by coding at bitrate $R$, and $\mmse(f_s)$ is the minimal distortion under sub-Nyquist sampling with infinite bit precision.} 
\end{figure}
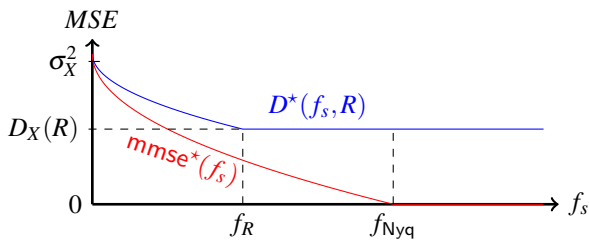

The distortion in ADX can be analyzed by considering the combined sampling and source coding model studied in \cite{Kipnis2014}. In this model, the analog source is a Gaussian stationary process. This process, or a noisy version of it, is sampled at rate $f_s$, after which the samples are encoded using a code rate of $R$ bits per time unit ($R/f_s$ bits per sample on average). However, while \cite{Kipnis2014} focused on uniform sampling with linear time invariant (LTI) pre-processing, our setting incorporates the class of all bounded linear samplers. That is, each sample is obtained by a bounded linear functional applied to the continuous-time analog path. We limit the average number of such samples obtained over a finite time interval to be at most $f_s$. As a result, our setting encapsulates a wide range of sampling models that are used in theory and practice. These include: filter-bank sampling, nonuniform sampling, multi-coset sampling \cite{950786, EldarMichaeliBeyond}, and truncated wavelet transforms \cite{eldar2015sampling} \cite{daubechies1998factoring}. 
%Our setting does not include compressed sensing since the sampled signal is sparse rather than Gaussian. although a the fundamental limits of lossy compression in compressed sensing using a setting equivalent to ADX are considered in \cite{kipnis2017fundamental}.
\par
In the special case of scalar uniform sampling and assuming that $f_s$ is above the Nyquist rate of the source, the  encoder in Fig.~\ref{fig:ADX_system} can estimate the signal with vanishing distortion prior to encoding it. As a result, in this case, the distortion associated with sampling is zero, and the minimal ADX distortion is described by the information DRF of the analog source. In this paper we ask the following question: given a source coding rate constraint $R$ (for example, as a result of quantizing each sample using  $R/f_s$ bits), do we still need to sample at the Nyquist rate in order to achieve the DRF or is a lower sampling rate sufficient? By answering this question, we establish in this work a critical sampling rate $f_R$, which is in general lower than the Nyquist rate, such that sampling at this rate achieves the distortion-rate bound at bitrate $R$. This is illustrated in Fig.~\ref{fig:contribution}, where we see that sampling below the Nyquist rate is possible without additional distortion over that given by the DRF associated with Nyquist rate sampling. 
Our results also imply that a picture similar to Fig.~\ref{fig:contribution} holds if we replace the uniform sampler by a multi-branch sampler or a nonuniform sampler. In this case, the sampling rate $f_s$ allowing optimal sampling with unlimited bitrate is the \emph{spectral occupancy} or \emph{Landau rate} of the signal \cite{1447892, 950786}, i.e., the Lebesgue measure of the support of its PSD. The spectral occupancy is also termed the \emph{Landau rate}, and it coincides with the Nyquist rate whenever the support of the PSD is a connected set (centered around the origin since the the signal is real). \par
When the signal is contaminated by noise before or during the sampling operation, there in no hope to achieve the DRF even with an unlimited sampling budget. Instead, the minimal distortion is described by the indirect DRF of the signal given its noisy version \cite{1057738}\cite[Sec. 4.5.4]{berger1971rate}. In this case, our results imply that the critical sampling rate $f_R$ achieving the indirect DRF at bitrate $R$ depends both on $R$ and the noise, and can be attained in a similar manner as in the non-noisy setting. \\

In order to intuitively understand why optimal lossy compression performance can be attained by sampling below the Nyquist rate, one may consider the lossy compression of a signal represented by a sequence of independent Gaussian random variables. This representation is quite general since most signals of interest can be represented using their independent coefficients under some orthogonal basis transformation \cite{donoho1998data}. In order to compress such a sequence in an optimal manner subject to a minimum mean squared (MSE) criterion in reconstruction, a source code is constructed by random samples from the distribution resulting from the water-filling formula of Kolmogorov \cite{1056823}. This distribution is a Gaussian product distribution where the variance of each of its components is obtained by subtracting the water-level parameter from the variance of the corresponding component in the original signal, so that signal components with variances smaller than the water-level are set to zero. As we explain in detail in Section~\ref{sec:finite_dims}, the ratio between the number of non-zeros in Kolmogorov's formula and the original support of the distribution of the sequence can be seen as the optimal sampling rate required to attain the minimal distortion subject to the bit constraint. 
\par
For an analog stationary signal, its Fourier basis decomposition provides a canonical orthogonal representation. Hence, the main challenge in attaining the optimal lossy compression at bitrate $R$ by sampling at rate $f_R$ is in ``aligning'' the distribution of the sampled signal in the Fourier domain with the optimal lossy compression attaining distribution. When $f_s$ is below $f_R$, the optimal alignment is described by a function $D^\star(f_s,R)$ defined by water-filling over $f_s$ spectral bands of maximal energy (or maximal SNR in the noisy version). As we show, this ``alignmnet'' is attainable by uniform multi-branch sampling using appropriate LTI pre-sampling operations. Together with a matching converse theorem with respect to $D^\star(f_s,R)$ under any bounded linear sampler, we conclude that $D^\star(f_s,R)$ fully characterizes the distortion in ADX. In particular, our results imply that the class of multi-branch LTI uniform sampling is optimal, in the sense that the distortion attained by any bounded linear sampler can be attained by a multi-branch uniform sampler with a sufficient number of sampling branches. \par
We also examine the distortion-rate performance of a very simple and sub-optimal A/D scheme: a scalar quantizer with a fixed number of bits per sample as an encoder and a linear non-causal decoder. We analyze this A/D scheme under a fixed bitrate budget, and show that there exists a distortion minimizing sampling rate that optimally trades off distortion due to sampling and due to quantization precision. This optimal sampling rate is at or below the Nyquist rate, and experiences a similar dependency in the bitrate as the critical ADX rate $f_R$. Our results also imply that, as opposed to the behavior of the optimal ADX distortion $D^\star(f_s,R)$, oversampling a bandlimited signal in PCM has a detrimental effect on the distortion.
%This phenomena was observed under various other source coding settings combining sampling and quantization \cite{370112,wu2012optimum,marco2010entropy}. 

%The main contribution of this paper is a full characterization of the minimal distortion in ADX. This distortion is given in terms of a function $D^\star(f_s,R)$ defined by water-filling over $f_s$ spectral bands of maximal energy (or maximal SNR in the noisy version), in a way analogous to Pinsker's water-filling expression for the DRF of a Gaussian analog source \cite{1056823}. The relation between the minimal distortion in ADX and $D^\star(f_s,R)$, is given in Section~\ref{sec:main_result} by standard achievability and converse theorems. 
% is informally given by:
% \begin{itemize}
%     \item [(i)] For any bounded linear sampling with average sampling rate at most $f_s$, the distortion in ADX at bitrate $R$ is bounded from below by $D^\star(f_s,R)$. 
%     \item [(ii)] For any sampling rate $f_s$, there exists a bounded linear sampler of average sampling rate $f_s$ and a representation of the samples obtained over time lag $T$ using $TR$ bits, such that the distortion in ADX converges to $D^\star(f_s,R)$. 
% \end{itemize}
%These results, derived in Section~\ref{sec:ADX}, establish $D^\star(f_s,R)$ as the minimal distortion in ADX. The critical sampling rate $f_R$ above which there is not additional distortion due to sampling is derived from $D^\star(f_s,R)$ in Section~\ref{sec:main_result}. We also illustrate in this section the dependency of $f_R$ in $R$ through various examples. \par

To put our work into context, we now briefly review some of the well-known sampling theories and their relation to our work. The celebrated Shannon-Kotelnikov-Wittaker sampling theorem asserts that a bandlimited deterministic signal $x(\cdot)$ with finite $\mathrm L_2$ norm can be perfectly reconstructed from its uniform samples at frequency $f_s > f_{\Nyq}$, where $f_{\Nyq}$ is the bandwidth of the signal. 
This statement can be refined when the exact support $\supp\,S_x$ of the Fourier transform of $x(\cdot)$ is known: $x(\cdot)$ can be obtained as the limit in $\Ell$ of linear combinations of the samples $x\left(\mathbb Z/f_s \right)$ iff for all $k\neq n \in \mathbb Z$, $\left(\supp\, S_x + f_s k \right) \cap \left(\supp\,S_x+f_s n \right)= \emptyset$, where a reconstruction formula is also available \cite{dodson1985fourier}. 

Lloyd \cite{Lloyd1959} provided an equivalent result for stationary stochastic processes, where the Fourier transform is replaced by the power spectral density (PSD). 
When sampling at the Nyquist rate is not possible, the minimal MSE (MMSE) in estimating a Gaussian stationary process from its uniform samples can be expressed in terms of its PSD \cite{1057404,1090615, 815501}. This MMSE in the case of multi-branch sampling was derived in \cite[Sec. IV]{Kipnis2014}. \par
In general, the estimation of any regular Gaussian stationary process from its partial observations can be translated into the problem of projections into Hilbert spaces generated by complex exponentials \cite{beutler1961sampling, dym1978gaussian}. 
In particular, when the PSD is supported over a compact set $S \subset \mathbb R$, then the closed linear space (CLS) of the exponentials with support over $S$ is isomorphic, by the Fourier transform operator, to the \emph{Paley-Wiener space} $\Pw(S)$ of functions with Fourier transform supported in $S$. 
% see if can consolidate sampling theorey here
In this space, optimal reconstruction of signals from their samples is possible when the samples define a \emph{frame} \cite{FEICHTINGER1992530, 1447892}. Beurling and Landau \cite{beurling1989collected, Landau1967} showed that a sufficient and necessary condition for a discrete set of time samples to define a frame in $\Pw(S)$ is that its Beurling density (also called \emph{uniform density}) exceeds the Lebesgue measure $\mu(S)$ of $S$. In our setting $\mu(S)$ is the spectral occupancy of the signal, which we also refer to as its \emph{Landau} rate. %Since the spectrum of real signals is symmetric, their Landau and Nyquist rates coincides whenever their spectrum is supported over a single interval centered at the origin\footnote{For this reason Landau referred to the spectral occupancy as the Nyquist rate in \cite{1447892}, although the former is clearly a generalization of the latter.}. 
We refer to \cite{john1996sampling, marvasti2012nonuniform, unser1994general, eldar2015sampling} for additional background on sampling theory and generalized sampling techniques.
\par
On the other side of the ADX setting is the distortion in lossy compression at a limited bitrate $R$. The optimal trade-off between the average quadratic distortion and bitrate in the description of a Gaussian stationary process $X(\cdot)$ is given by its quadratic DRF, denoted here by $D_X(R)$. This DRF was initially derived by Pinsker in \cite{Pinsker1954} and was reported in \cite{1056823}, and then extended by Dubroshin and Tsybakov \cite{1057738} to the case where the process is contaminated by Gaussian noise. Both the noisy case explored by Dubroshin and Tsybakov and the ADX characterized in this work fall within the \emph{indirect} or \emph{remote} source coding setting \cite[Sec. 4.5.4]{berger1971rate}, in which the encoder has no direct access to the signal it tries to describe. Indirect source coding problems were also considered in \cite{1056251,1054469, KipnisRini2015}. 
%Within the context of Gaussian signals, a review of the interplay and between source coding and harmonic analysis is given in \cite{donoho1998data}. 
\par
%
%The interplay between source coding and sampling arise in numerous settings, many of which can be seen as special cases of our ADX setting by restricting to a special signal distribution, specific bounded linear sampler, or sub-optimal encoder and/or decoder. Specifically, the tradeoff between sampling rates and bitrate in the description of a function of correlated sources was investigated in \cite{6573236}. The conditions on the encoder such that the distortion-rate is achievable as sampling density or rate goes to infinity were addressed in \cite{neuhoff2013information}. For sampling rates above the Nyquist rate and under non-ideal encoding, the tradeoff between lossy compression rates and distortion was considered in \cite{370112}, and the trade-off between sampling density (number of sensors) and bit-resolution of each sensor leading to exponential distortion decay was studied in \cite{kumar2011high}. \par
%
The interplay between bit resolution in source coding and sampling rates arise in numerous settings. For sampling rates above the Nyquist rate, non trivial trade-offs between the oversampling rate and bitrate, under different encoding scenarios, can be found in \cite{neuhoff2013information}, 
\cite{6573236}, \cite{370112}, \cite{marco2010entropy} and \cite{kumar2011high}. 
%many of which can be seen as special cases of our ADX setting by restricting to a special signal distribution, specific bounded linear sampler, or sub-optimal encoder and/or decoder. Specifically, the tradeoff between sampling rates and bitrate in the description of a function of correlated sources was investigated in \cite{6573236}. The conditions on the encoder such that the distortion-rate is achievable as sampling density or rate goes to infinity were addressed in \cite{neuhoff2013information}. For sampling rates above the Nyquist rate and under non-ideal encoding, the tradeoff between lossy compression rates and distortion was considered in \cite{370112}, and the trade-off between sampling density (number of sensors) and bit-resolution of each sensor leading to exponential distortion decay was studied in \cite{kumar2011high}. \par
In order to explore the trade-off between lossy compression and sub-Nyquist sampling rates, a combined sampling and source coding problem was recently introduced in \cite{Kipnis2014} assuming uniform sampling. The ADX can be seen as an extension of the setting in \cite{Kipnis2014} to any bounded linear sampling technique, and the determination of the minimal sampling rate $f_R$ attaining the optimal source coding performance. 
Finally, in the context of compressed sensing (CS) \cite{eldar2012compressed}, the optimal trade-off between the sampling rate and bitrate is explored in \cite{kipnis2017fundamental} for an i.i.d. Bernoulli-Gauss distribution, and in \cite{wu2011optimal} for an arbitrary i.i.d. distribution as the number of bits goes to infinity. We note that our results are not directly relevant to CS since we focus on sampling continuous-time Gaussian signals that are not sparse. Nevertheless, the discrete-time counterpart of our results may be applied to CS to obtain a lower bound on the distortion when the signal's support is given as side information, or an upper bound on the distortion when the samples of the signal are encoded using a Gaussian codebook \cite{KipnisCS}.
\\

%\subsection{Paper Organization}
The rest of the paper is organized as follows:
in Section~\ref{sec:finite_dims} we provide intuition for the dependency between sampling and lossy compression in representing finite dimensional random vectors. In Section~\ref{sec:problem_formulation} we define the ADX problem and the class of bounded linear samplers. Our main results are given in Section~\ref{sec:ADX}. In Section~\ref{sec:PCM} we consider scalar quantization encoding and compare its performance to the minimal ADX distortion. Concluding remarks are provided in Section \ref{sec:conclusion}.

\section{Lossy Compression of Finite Dimensional Signals \label{sec:finite_dims}}
%To understand intuitively why we may have equality $D^\star(f_s,R)=D_X(R)$ for sampling frequencies below the Landau rate, i.e., the sampling rate determined by the support of the spectrum, we consider the finite dimensional version of the combined sampling and source coding problem of Fig.~\ref{fig:system_model}. This intuition is presented in the following subsection. 
As an introduction to the ADX setup, it is instructive to consider a simpler setting involving the sampling and lossy compression of signals represented as finite dimensional random real vectors.  \\

Let $X^n$ be an $n$-dimensional Gaussian random vector with covariance matrix $\Sigma_{X^n}$, and let $Y^m$ be a projected version of $X^n$ defined by
\begin{equation} \label{eq:samples_discrete}
Y^m = H X^n,
\end{equation}
where $H \in \mathbb R^{m \times n}$ is a deterministic matrix and $m<n$. This projection of $X^n$ into a lower dimensional space is the counterpart for the sampling operation in the ADX setting of Fig.~\ref{fig:ADX_system}. We consider the normalized MMSE estimate of $X^n$ from a representation of $Y^m$ using a limited number of bits. \par
Without constraining the number of bits, the distortion in this estimation is given by
\begin{equation}
\label{eq:mmse_finite_dim}
\mmse(X^n|Y^m) \triangleq  \frac{1}{n} \tr\left(\Sigma_{X^n} - \Sigma_{X^n|Y^m} \right),
\end{equation}
where $\Sigma_{X^n|Y^m}$ is the conditional covariance matrix of $X^n$ given $Y^m$.
%Since $X^n$ is Gaussian, 
%\[
%\Sigma_{X^n|Y^m} \triangleq \Sigma_{X^n} H^* \left( {H} %\Sigma_{X^n} H^* \right)^{-1} H \Sigma_{X^n},
%\]
When $Y^m$ is encoded using a code of no more than $nR$ bits, the minimal distortion cannot be smaller than the indirect DRF of $X^n$ given $Y^m$, denoted by $D_{X^n|Y^m}(R)$. This function is given by the following parametric expression \cite{1057738}
\begin{equation}
\begin{split}
D(R_{\theta}) & = \tr \left(\Sigma_{X^n}\right) - \sum_{i=1}^m \left[  \lambda_i \left( \Sigma_{X^n|Y^m} \right) -\theta \right]^+, \label{eq:D_finite} \\
R_{\theta}  & = \frac{1}{2} \sum_{i=1}^m \log^+\left[ \lambda_i \left(\Sigma_{X^n|Y^m} \right)/\theta \right] 
\end{split}
\end{equation}
where $x^+ = \max\{x,0\}$ and $\lambda_i \left( \Sigma_{X^n|Y^m} \right)$ is the $i$th eigenvalue of $\Sigma_{X^n|Y^m}$. \\

\begin{figure}
\begin{center}
\begin{tikzpicture}
\draw[->,line width = 1] (0,0) -- (0,4);
\draw[line width = 1] (0,0) -- (7.2,0);

\draw[line width = 1] (0.1,4.2)--(0.1,4.4) -- node[above] {eigenvalues of $\Sigma_{X^n}$} (3.8,4.4) -- (3.8,4.2);

\draw  [fill = blue!30, line width = 1] (0.1,0) rectangle (0.6,3.7) node[above, xshift=-0.1cm] {$\lambda_n$};
\draw  [fill = blue!30, line width = 1] (0.9,0) rectangle (1.4,3) node[above, xshift=-0.1cm] {$\lambda_{n-1}$};
\draw  [fill = blue!30, line width = 1] (1.7,0) rectangle (2.2,2.5) node[above, xshift=-0.1cm] {\small $\lambda_{n-2}$};
\draw  [fill = red!30, line width = 1] (0.1,0) rectangle (0.6,1.7);
\draw  [fill = red!30, line width = 1] (0.9,0) rectangle (1.4,1.7);
\draw  [fill = red!30, line width = 1] (1.7,0) rectangle (2.2,1.7);
\draw  [fill = red!30, line width = 1] (2.5,0) rectangle (3,1.5);
\draw  [fill = red!30, line width = 1] (3.3,0) rectangle (3.8,1);
\draw [dashed, line width = 1] (-0.1,1.7) -- (7,1.7) node[right, xshift = 0cm] {$\theta$};

\draw[line width = 1] (4.1,4.2)--(4.1,4.4) -- node[above] {eigenvalues of $\Sigma_{X^n|Y^m}$} (7,4.4) -- (7,4.2);
\draw[->,line width = 1] (4,0) -- (4,4);
\draw  [fill = blue!30, line width = 1] (4.1,0) rectangle (4.6,3.7) node[above, xshift=0.1cm] {$\lambda_m$};
\draw  [fill = blue!30, line width = 1] (4.9,0)  rectangle (5.4,3) node[above, xshift=0.1cm] {$\lambda_{m-1}$};
\draw  [fill = blue!30, line width = 1] (5.7,0) rectangle (6.2,2.5) node[above, xshift=-0.1cm] {\small $\lambda_{m-2}$};

\draw  [fill = red!30, line width = 1] (4.1,0) rectangle (4.6,1.7);
\draw  [fill = red!30, line width = 1] (4.9,0) rectangle (5.4,1.7);
\draw  [fill = red!30, line width = 1] (5.7,0) rectangle (6.2,1.7);
\draw  [fill = red!30, line width = 1] (6.5,0) rectangle (7,0.5);

\end{tikzpicture}
\caption{\label{fig:eigenvalues} 
Optimal sampling occurs whenever $D_{X^n}(R)=D_{X^n|Y^m}(R)$. This condition is satisfied even when $m<n$, as long as there is equality among the eigenvalues of $\Sigma_{X^n}$ and $\Sigma_{X^n|Y^m}$ which are larger than the water-level parameter $\theta$.}
\end{center}
\end{figure}
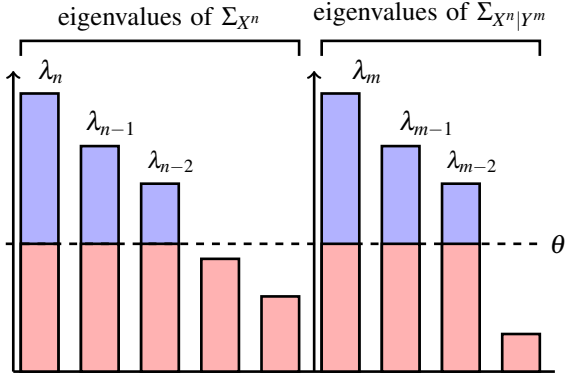

It follows from \eqref{eq:mmse_finite_dim} that $X^n$ can be recovered from $Y^m$ with zero MMSE if and only if 
\begin{equation}\label{eq:eigen_equality}
 \lambda_i\left(\Sigma_{X^n} \right) =  \lambda_i \left( \Sigma_{X^n| Y^m} \right),
\end{equation}
for all $i=1,\ldots,n$. When this condition is satisfied, \eqref{eq:D_finite} takes on the form 
\begin{equation}
\begin{split}
D(R_{\theta}) & =  \sum_{i=1}^n \min \left\{ \lambda_i \left( \Sigma_{X^n} \right), \theta \right\},  \label{eq:D_finite_regular} \\
R_{\theta} & = \frac{1}{2} \sum_{i=1}^n \log^+\left[ \lambda_i \left(\Sigma_{X^n} \right)/\theta \right] 
\end{split}
\end{equation}
which is Kolmogorov's reverse waterfilling expression for the DRF of the vector Gaussian source $X^n$ \cite{1056823}, i.e., the minimal distortion in encoding $X^n$ using codes of rate $R$ bits per source realization. The key insight is that the requirements for equality between \eqref{eq:D_finite} and \eqref{eq:D_finite_regular} are not as strict as \eqref{eq:eigen_equality}: all that is needed is equality among those eigenvalues that affect the value of \eqref{eq:D_finite_regular}. In particular, assume that for a point $(R,D)$ on $D_{X^n}(R)$, only $\lambda_n (\Sigma_{X^n}),\ldots \lambda_{n-m+1} (\Sigma_{X^n})$ are larger than $\theta$, where the eigenvalues are organized in ascending order. Then we can choose the rows of the matrix $H$ to be the $m$ left eigenvectors corresponding to $\lambda_n (\Sigma_{X^n}),\ldots \lambda_{n-m+1} (\Sigma_{X^n})$. With this choice of $H$, the $m$ largest eigenvalues of $\Sigma_{X^n|Y^m}$ are identical to the $m$ largest eigenvalues of $\Sigma_{X^n}$, and \eqref{eq:D_finite_regular} is equal to \eqref{eq:D_finite}. \par
Since the rank of the sampling matrix is now $m<n$, we effectively performed sampling below the ``Nyquist rate'' of $X^n$ without degrading the performance dictated by its DRF. One way to understand this phenomena is an alignment between the range of the sampling matrix $H$ and the subspace over which $X^n$ is represented, according to Kolmogorov's expression \eqref{eq:D_finite_regular}. When this expression implies that not all degrees of freedom are utilized by the optimal distortion-rate code, sub-sampling does not incur further performance loss provided the sampling matrix is aligned with the optimal code. This situation is illustrated in Fig.~\ref{fig:eigenvalues}. Taking less rows than the actual rank of $\Sigma_{X^n}$ is the finite-dimensional analog of sub-Nyquist sampling in the infinite-dimensional setting of continuous-time signals. \\

In the rest of this paper we explore the counterpart of the phenomena described above in the richer setting of continuous-time stationary processes that may or may not be bandlimited, and whose samples may be corrupted by additive noise. The precise problem description is given in the following section.

\tikzstyle{int1}=[draw, fill=blue!10, minimum height = 0.5cm, minimum width=1cm,thick ]
\tikzstyle{enc}=[draw, fill=blue!10, minimum height = 2.7cm, minimum width=1cm,thick ]
\tikzstyle{int}=[draw, fill=blue!10, minimum height = 1cm, minimum width=1.5cm,thick ]

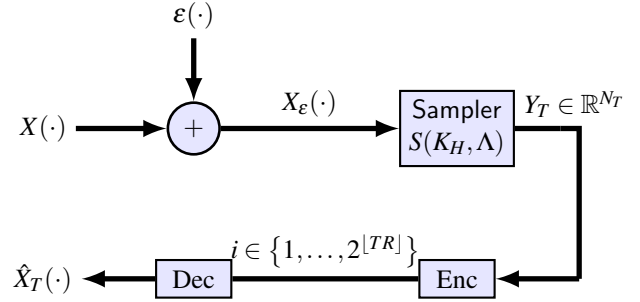
\begin{figure}
\begin{center}
\begin{tikzpicture}[node distance=2cm,auto,>=latex]
 \node at (0,0) (source) {$X(\cdot)$} ;
 \node [right] (dest) [below of=source, node distance = 2cm]{$\hat{X}_T(\cdot)$};
 %\node [coordinate,right of = source, node distance =2.5cm] at (0,0) (smp_in) {};
 \node[ssum, right of = source, node distance = 2cm] (plus) {$+$};
 
\node[above of = plus, node distance = 1.5cm] (noise) {$\epsilon(\cdot)$};
\node [int,right of = plus, node distance =3.5cm, align = center]  (smp) {$\mathsf{Sampler}$ \\ $S(K_H,\Lambda)$};
 \draw[->, line width = 2pt] (noise) -- (plus);

\node[int1,below of=smp, node distance = 2cm, align = center] (enc) {$\enc$};
\node[int1,below of = plus, node distance = 2cm] (dec) {$\dec$};

\node [right of = smp, node distance = 1.5cm] (right_edge) {};
\node [below of = right_edge, node distance = 2cm] (right_b_edge) {};

\draw[-,line width=2pt] (smp) -- node[above, xshift = 0.5cm] {$Y_T \in \mathbb R^{N_T}$} (right_edge);

\draw[-,line width=2pt] (enc) -- node[above, xshift = 0cm]{$i\in \left\{1,\ldots,2^{\lfloor T R \rfloor} \right\}$}(dec);
\draw[-,line width = 2]  (right_edge.west) -|  (right_b_edge.east);
\draw[->,line width = 2]  (right_b_edge.east) -- (enc);
\draw[->,line width = 2]  (dec) -- (dest);
%\draw[->,line width=2pt] (source) -- (smp);
\draw[->,line width=2pt] (source) -- (plus);
\draw[->,line width=2pt] (plus) -- node[above] {$X_{\epsilon}(\cdot)$} (smp);
%\draw[<->,dashed,line width=1pt] (source) -- node {distortion} (dest);
\end{tikzpicture}

\caption{\label{fig:system_model} ADX via a
combined sampling and source coding setting with an additive noise prior to sampling. We consider the distortion in recovering $X(\cdot)$ over $[-T/2,T/2]$ from a representation of its $N_T$ samples using $\lfloor TR \rfloor $ bits, where $N_T$ is the number of samples in $[-T/2,T/2]$ and $N_T/T$ is bounded asymptotically by $f_s$. 
}
\end{center}
\end{figure}

\section{Problem Formulation and Preliminaries \label{sec:problem_formulation}}

\subsection{ADX Setting}
The ADX system is described in  Fig.~\ref{fig:system_model}. We assume that $X\left(\cdot\right)=\left\{ X\left(t\right),\, t\in\mathbb{R}\right\}$ is a zero-mean real Gaussian stationary process with a \emph{known} PSD $S_X(f)$:
%The PSD is real, symmetric and an absolutely integral function that satisfies
\begin{equation}
    \label{eq:PSD_def}
\mathbb E \left[X(t) X(s) \right] = \int_{-\infty}^\infty S_X(f) e^{2\pi j (t-s) f} df 
 ,\quad t,s \in \mathbb R.
\end{equation}
In particular, $S_X(f)$ is in $\mathrm L_1(\mathbb R)$ and the variance of $X(\cdot)$ is given by
\[
\sigma_X^2 = \int_{-\infty}^\infty S_X(f) df. 
\]
The noise is another zero-mean real Gaussian stationary process $\epsilon \left(\cdot\right)=\left\{ \epsilon\left(t\right),\, t\in\mathbb{R}\right\}$ independent of $X(\cdot)$ with PSD $S_\epsilon(f)$ and finite variance. We assume that the spectral measures of $X(\cdot)$ and $\epsilon(\cdot)$ are absolutely continuous with respect to the Lebesgue measure, so that their distribution is fully characterized by their PSDs. 
\par
The sampler $S$ belongs to the class of bounded linear samplers to be defined in the sequel. This sampler 
receives the process 
 \[
X_\epsilon (\cdot) \triangleq X(\cdot)  + \epsilon(\cdot),
\]
i.e., the noisy version of $X(\cdot)$ as its input. For a finite time horizon $T>0$, the sampler $S$ produces a finite number
$N_T$ of samples 
\[
Y_T \triangleq \left(Y_1,\ldots,Y_{N_T} \right) = S_T \left(X_\epsilon(\cdot)\right). 
\]
The assumption that the variance of the noise is finite excludes, for example, $\epsilon(\cdot)$ from being a white noise signal. This assumption is necessary to define sampling of $X_\epsilon(\cdot)$ in a meaningful way, as we explain below. \par
%where each sample $Y_n$ is restricted to be a continuous linear functional of $X_{\epsilon}(\cdot)$ in the sense defined in the sequel. 
The encoder
\begin{equation}
    \label{eq:encoder}
f : \mathbb R^{N_T} \rightarrow  \left\{1,\ldots,2^{\lfloor TR \rfloor} \right\},
\end{equation}
receives the vector $Y_T$ and outputs an index in $\left\{1,\ldots,2^{\lfloor TR \rfloor} \right\}$. The decoder, 
\begin{equation}
    \label{eq:decoder}
g :  \left\{1,\ldots,2^{\lfloor TR \rfloor} \right\} \rightarrow  \mathbb R^{[-T/2,T/2]},
\end{equation}
upon receiving this index from the encoder, produces a reconstruction waveform $\hat{X}_T(\cdot)$. The goal of the joint operation of the encoder and the decoder is to minimize the average MSE
\begin{equation} \label{eq:MSE_def}
\frac{1}{T} \int_{-T/2}^{T/2} \mathbb E \left(X(t)- \widehat{X}_T(t) \right)^2 dt. 
\end{equation}

Given a particular bounded linear sampler $S$, and a bitrate R, we are interested in characterizing the function 
\begin{equation}
    \label{eq:DRF_finit_R}
    D_T(S,R) \triangleq \inf_{f,g} \frac{1}{T} \int_{-T/2}^{T/2} \mathbb E \left(X(t) - \widehat{X}_T(t) \right)^2,
\end{equation}
 where the infimum is over all encoders and decoders of the form \eqref{eq:encoder} and \eqref{eq:decoder}. We also consider the asymptotic version of \eqref{eq:DRF_finit_R}:
\begin{equation}
    \label{eq:DRF_asymp}
    D(S,R) \triangleq \liminf_{T\rightarrow \infty} D_T(S,R). 
\end{equation}

Before describing the class of bounded linear samplers, we remark on some of the properties of ADX setting:
\begin{itemize}
\item Information loss in ADX is due to noise, sampling, and encoding. We do not consider limitations on the decoder that may exists in practice, such as memory or complexity. %For example, the decoder may first output a finite number of samples that are then interpolated to the continuous time estimate $\widehat{X}(\cdot)$. 
\item The additive noise $\epsilon(\cdot)$ may be seen as an external interference in transmitting $X(\cdot)$ or as noise associated with the sampling operation. With obvious adjustments, our setting can also handle a discrete-time noise vector with a stationary distribution added post sampling.
\item For a finite time horizon $T$, the decoder is only required to recover $X(\cdot)$ over the interval $[-T/2,T/2]$. 
However, as follows from the description of the sampler below, each sample may depend on a realization of $X_\epsilon(\cdot)$ over the entire time-horizon (past and future). It is possible to restrict the sampler to be a function of $X(\cdot)$ only over $[-T/2,T/2]$ provided the conditional distribution of $X(\cdot)$ given its samples converges to an asymptotic distribution as $T$ goes to infinity. Our asymptotic analysis below remains valid under this restriction due to the stationary distribution of $X(\cdot)$. Our setting also prohibits the sampler to depend on $T$. This restriction precludes adaptive sampling schemes such as in \cite{boda2017sampling}.
%Under this condition, our main results on an asymptotic lower bound to $D_S(f_s,R)$ holds for samplers 
\item Since $X(\cdot)$ is a stationary process, we can replace the interval $[-T/2,T/2]$ by any other interval of length $T$ without affecting the main results. 
\item As opposed to common situations in source coding of stationary processes (e.g., \cite[Lem. 10.6.2]{gray2011entropy}), the liminf in \eqref{eq:DRF_asymp} cannot be replaced by a simple infimum or a limit. One explanation for this difference is that, as we explain below, the coding scheme that attains $D(f_s,R)$ essentially describes the estimator of $X(\cdot)$ from the samples $Y^{N_T}$, and the distribution of this estimator is in general not stationary. 
\end{itemize}

\subsection{Bounded Linear Sampling of Random Signals \label{subsec:sampler_structure}}
We now describe the class of bounded linear samplers we use in the ADX setting of Fig.~\ref{fig:system_model}. Assume first that the input to the sampler is a deterministic signal $x(\cdot)$ in a class of signals $\mathcal X$. Each sample $Y_n$ can then be seen as the result of applying a functional $\phi_n$ on $x(\cdot)$. In accordance with physical considerations of realizable systems, we require that the space of signals $\mathcal X$ is embedded in the Hilbert space of real functions of finite energy $\Ell$, and that the functional defining the $n$th sample is linear and bounded. In other words, each sample is defined by an element of the dual space $\mathcal X^\star$ of $\mathcal X$. For this reason, we assume that $\mathcal X$ and $\mathcal X^\star$ are standard spaces of \emph{test functions} and \emph{distributions}, respectively \cite{zemanian1965distribution}, so that every distribution $\phi \in \mathcal X^\star$ has a Fourier transform in the Gelfand-Shilov sense \cite{GS2_english}. 
%and that $\mathcal X^\star$ contains the Dirac distribution, i.e., the pointwise evaluation operation on elements of $\mathcal X$ is continuous. 
Consequently, the bilinear operation $\langle \phi, x \rangle$ between $\phi \in \mathcal X^\star$ and $x \in \mathcal X$ satisfies the Plancherel identity
\begin{equation} \label{eq:Parseval}
 %\int_{-\infty}^\infty x(t) \phi^*(t) dt \triangleq
\langle \phi, x \rangle = \int_{-\infty}^\infty \mathcal Fx (f) \left(\mathcal F \phi(f) \right)^* df,
\end{equation}
where $\mathcal F$ is the Fourier transform and $*$ denotes complex conjugation. 
%Note that a sufficient and necessary condition for $\mathcal X^\star$ to include the Direc distribution $\delta_0$ is that each $x(\cdot) \in \mathcal X$ is in $\mathrm{L}_1(\mathbb R)$. \par
To summarize, for each $T>0$ and assuming an appropriate class of input signals $\mathcal X$, the output of the sampler is defined by a set of $N_T$ elements of $\mathcal X^\star$. We denote the samplers constructed in this manner as the class of \emph{bounded linear samplers}.  \\

Next, we consider bounded linear sampling of the random process $X_\epsilon(\cdot)$.
Denote $w_t(f) = e^{2\pi i f t}$. Since the spectral measure of $X_{\epsilon}(\cdot)$ is absolutely continuous with respect to the Lebesgue measure $\mu$, we have 
\begin{align}
\mathbb E \left[ X_\epsilon(t) X_\epsilon(s) \right] & = \int_{-\infty}^\infty w_t(f) w^*_s(f) S_{X_\epsilon}(f) df 
\label{eq:isometry}
\end{align}
so that the mapping $X_\epsilon(t) \rightarrow w_t$ is an isometry. As in \cite{beutler1961sampling}, we extend this isometry to an isomorphism between the Hilbert space
 %$Ell(\Omega,\sigma(X_\epsilon(\cdot)),\mathbf P)$
generated by the closed linear span (CLS) of the process $X_{\epsilon}(\cdot)$ with norm $\| X_\epsilon(t) \|^2 = \mathbb E [X_\epsilon^2(t)]$, and the Hilbert space $\mathcal W(S_{X_\epsilon}))$ which is the CLS of $\left\{ w_t,\, t\in \mathbb R \right\}$ with $\Ell$ the $\Ell$ norm weighed by $S_{X_\epsilon}(f)$. This isomorphism allows us to define bounded linear sampling of $X_\epsilon(\cdot)$ by describing its operation on $\mathcal W(S_{X_\epsilon}))$. 
Specifically, we identify $\mathcal X$ with the elements of $\mathcal W(S_{X_\epsilon}))$ and set $\mathcal X^\star$ to be a space of distributions such that, for any $\phi \in \mathcal X^\star$, its Fourier transform $\hat{\phi}$ satisfies
\begin{equation} \label{eq:func_condition}
\int_{-\infty}^\infty  |\hat{\phi}(f)|^2 S_X(f) df < \infty. 
\end{equation}
For such $\phi_n \in \mathcal X^\star$, we define the sample
\[
Y_n =\int_{-\infty}^\infty X_\epsilon(\tau) \phi_n^*(\tau)d\tau 
\]
to be the inverse image of $\langle \phi_n, w_t \rangle (f)$ under $X_\epsilon(t) \rightarrow w_t$. Although in most situations this inverse image cannot be found explicitly, we are usually interested in the joint statistics of $Y_n$ and $X(t)$, given by
\begin{equation} 
\label{eq:sample_stat}
\mathbb E \left[ Y_n X(t) \right] = \int_{-\infty}^\infty  \langle \phi_n, w_t \rangle(f)   e^{-2\pi i f} S_{X_\epsilon}(f) dt 
\end{equation}
In particular, condition \eqref{eq:func_condition} guarantee that the integral in \eqref{eq:sample_stat} exists.
\begin{example}[pointwise evaluation of bandlimited signals] \label{ex:sampling}
Assume that $\phi_n$ is the Dirac distribution at $t_n$ corresponding to pointwise evaluation at $t=t_n$, so that \eqref{eq:func_condition} holds and $\langle \phi_n, w_t\rangle = w_{t_n}$, whose inverse image is $X_\epsilon(t_n)$. If in addition $S_{X_\epsilon}(f)$ is supported on an open set $S\in \mathbb R$, then the element of $\mathcal W(S_{X_\epsilon}))$ can be identified with the Paley-Wiener space of complex valued functions whose Fourier transform is supported on $S$. In this case, for most applications it is enough to take $\mathcal X = \mathcal W(S_{X_\epsilon}))$ with its Hilbert space topology so that $\mathcal X^\star = \mathcal X$. For example, pointwise evaluation at $t=t_n$ for  $x \in \mathcal W(S_{X_\epsilon}))$ is obtained by the inner product of $x$ with $\mathcal F^{-1} (\mathbf 1_S(f) e^{2\pi i f t_n})$ which is a member of $\mathcal W(S_{X_\epsilon})$. 
\end{example}

In contrast to the scenario described in Example~\ref{ex:sampling}, we do not restrict ourselves to bandlimited signals at the sampler input. Thus, our setting supports any PSD $S_{X_\epsilon}(f)$ and corresponding set of functionals $\mathcal X^\star$ such that \eqref{eq:func_condition} holds. \\

\begin{figure}
\begin{center}
\begin{tikzpicture}[node distance=2cm,auto,>=latex]
\node at (0,0) (source) {$X_\epsilon(\cdot)$};

\node[int1, right of = source, node distance = 2cm ]  (pre_sampling2) {$K_H(t,\tau)$};  

\draw[-, line width=1pt] (source)--(pre_sampling2);   

\node [coordinate, right of = pre_sampling2,node distance = 2cm] (smp_in2) {};
  \node [coordinate, right of = smp_in2,node distance = 0.7cm] (smp_out2){};
	\node [coordinate,above of = smp_out2,node distance = 0.4cm] (tip2) {};
\fill  (smp_out2) circle [radius=2pt];
\fill  (smp_in2) circle [radius=2pt];
\fill  (tip2) circle [radius=2pt];`
\node[left,left of = tip2, node distance = 0.7 cm] (ltop2) {$t_n\in \Lambda$};

\node [right of = smp_out2, node distance=3cm]  (out) {$Y_T \in \mathbb R^{N_T}$};

\draw[->,densely dotted,line width = 1pt,thin] (ltop2) to [out=0,in=70] (smp_out2.north);
 \draw[line width=1pt]  (smp_in2) -- (tip2);
 \draw[-,line width=1pt]   (pre_sampling2)-- (smp_in2);
\draw[line width=1pt]  (smp_in2) -- (tip2);
\draw[->,line width = 1pt] (smp_out2) -- (out); 
\draw[line width=1pt, dashed] (0.9,0.8) rectangle (6,-0.7) ;
\end{tikzpicture}
\caption{\label{fig:sampler_nonuniform} Bounded linear sampler with a pre-sampling transformation with kernel $K_H$ and a sampling set $\Lambda$. }
\end{center}
\end{figure}
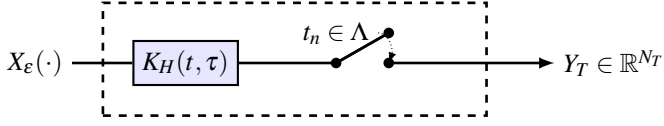

\begin{figure}
\begin{center}
\input{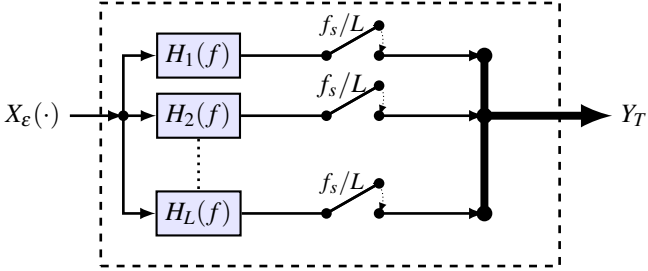}
\end{center}
\caption{ \label{fig:sampler_multi} Multi-branch linear time-invariant (MB-LTI) sampler.}
\end{figure}
Without loss of generality, it follows from the Schwartz kernel theorem \cite{gelfand1964generalized} applied to $\mathcal X \times \mathcal X^\star$ that the sequence of functionals defining the samples can be described in terms of a bilinear kernel $K_H(t,s)$ on $\mathbb R \times \mathbb R$ and a discrete \emph{sampling set} $\Lambda \subset \mathbb R$, as illustrated in Fig.~\ref{fig:sampler_nonuniform}. That is, the $n$th sample is given by
%The kernel $K_H(t,s)$ is chosen such that element $t_n \in \Lambda$ defines a linear bounded functional $K(t_n,s)$ on $\mathcal X$ given by
\[
y_n \triangleq \int_{-\infty}^\infty X_\epsilon(s)K_H(t_n,s) ds. 
\]
In order to control the number of samples taken every time horizon, we consider a set $\Lambda \subset \mathbb R$ that is uniformly discrete in the sense that there exists $\delta>0$ such that $|t-s|>\delta$ for every non identical $t,s\in \Lambda$. For a time horizon $T$, we denote
\[
\Lambda_T \triangleq \Lambda \cap [-T/2,T/2],
\] 
and define $y_T$ to be the finite dimensional vector obtained by sampling $x_\epsilon(\cdot)$ at times $t_1,\ldots,t_n \in \Lambda_T$. \par
The assumption that $\Lambda$ is uniformly discrete ensures that for any $T$, the \emph{density} of $\Lambda_T$,
\[
d_T(\Lambda) \triangleq \frac{\card \left(\Lambda_T \right)}{T},
\]
 if finite, and so is the limit
\[
d^+(\Lambda) \triangleq \limsup_{T\rightarrow \infty} d_T(\Lambda).
\]
We denote $d^+(\Lambda)$ as the  \emph{upper symmetric density} of $\Lambda$. Whenever it exists, we define the limit
\[
d(\Lambda) = \lim_{T\rightarrow \infty} d(\Lambda_T) = \lim_{T\rightarrow \infty} \frac{\card\left(\Lambda \cap [-T/2,T/2] \right)}{T},
\]
as the \emph{symmetric density} of $\Lambda$. \\

\subsection{Multi-Branch LTI Uniform Sampling}
An important special case of bounded linear sampling is described by the sampler in Fig.~\ref{fig:sampler_multi}. This sampler has $L$ sampling branches, where the $l$th branch consists of a linear time invariant (LTI) pre-sampling filter with transfer function $H_l(f)$ followed by a uniform sampler at rate $f_s/L$. Consequently, the $n$th sample produced by the $l$th branch is given by
\[
Y_{l,n} = \int_{-\infty}^\infty h_l(nL/f_s-\tau)X_\epsilon(\tau).
\] 
We define 
\[
\Yv_n = (Y_{1,n},\ldots,Y_{L,n})
\]
as the $n$th output of all branches. For a finite time horizon $T$, the output of the sampler is 
\[
Y_T \triangleq \left\{ \Yv_n,\, |n| < \lfloor T f_s/L \rfloor /2 \right\},
\]
so that $Y_T$ incorporates at most $N_T = \lfloor T f_s \rfloor$ samples from the process at the input to the sampler. \par
The class of samplers obtained in this manner is called multi-branch LTI uniform samplers (MB-LTI), where we denote a single sampler from this class by $S_{f_s}(H_1,\ldots,H_L)$. In order to see that a MB-LTI is a bounded linear sampler, note that its $n$th sample can be defined by the functional $\phi_n = \phi_{kL+l}$, $k=0,\ldots,N/L$, $l=1,\ldots,L$, 
\[
\int_{-\infty}^\infty X_\epsilon(\tau) \phi_n(\tau) d\tau = \int_{-\infty}^\infty X_\epsilon(\tau) h_l(kL/f_s-\tau) d\tau. 
\]
A MB-LTI sampler belongs to the class of shift-invariant samplers \cite{eldar2015sampling}, for which 
\begin{equation} \label{eq:asymp_samples}
Y_{\infty} \triangleq \cup_{T>0}  Y_{N_T},
\end{equation}
is invariant to time shifts by integer multiples of $L/f_s$ in the input to the sampler $X_\epsilon(\cdot)$.

\subsection{Properties of optimal Encoding and Connection to Classical Results
\label{subsec:properties}}
We now explore basic properties of the functions $D_T(S,R)$ and $D(S,R)$ of \eqref{eq:DRF_finit_R} and \eqref{eq:DRF_asymp} describing the minimal distortion in ADX. By doing so, we review previous results in sampling and source coding theory and explain their connection to our setting. \\

Denote by $\widetilde{X}_T(\cdot)$ the process that is obtained by MMSE estimation of $X(\cdot)$ from the vector of samples $Y_T$. Namely
\begin{equation}
    \label{eq:mmse_estimator_finite_T}
\widetilde{X}_T(t) \triangleq \mathbb E \left[X(t)|Y_T \right],\quad t \in \mathbb R.
\end{equation}
From properties of the conditional expectation, for any encoder $f$ we have 
\begin{align} \label{eq:decomp}
    & \frac{1}{T}\int_{-T/2}^{T/2} \mathbb E \left(X(t) - \mathbb E \left[X(t) | Y_T \right] \right)^2 dt \\
    & = \mmse_T(S) + \mmse \left(\widetilde{X}_T | f(Y_T) \right), \nonumber
\end{align}
where $\widehat{X}_T(\cdot) = g\left(f\left( Y_T \right) \right)$, 
\begin{equation}
    \label{eq:ADX_mmse}
\mmse_T(S) \triangleq \frac{1}{T} \int_{-T/2}^{T/2} \mathbb E \left(X(t) - \widetilde{X}_T(t) \right)^2 dt,
\end{equation}
is the distortion associated only with sampling and noise, and
\[
\mmse \left(\widetilde{X}_T | f(Y_T) \right) \triangleq \frac{1}{T} \int_{-T/2}^{T/2} \mathbb E \left( \widetilde{X}_T(t) - \mathbb E \left[\widetilde{X}_T(t) | f(Y_T) \right] \right)^2 dt
\]
is the distortion associated with the lossy compression procedure, and depends on the sampler only through $\widetilde{X}_T(\cdot)$. \par
%The decomposition \eqref{eq:decomp} already provides important clues on an optimal encoder and decoder pair that attains  $D_T(S,R)$. Specifically, it follows from \eqref{eq:decomp} that there is no loss in performance if the encoder tries to describe the process $\widetilde{X}_T(\cdot)$ subject to the bitrate constraint, rather than the process $X(\cdot)$. Consequently, the optimal decoder outputs the conditional expectation of $\widetilde{X}_T(\cdot)$ given $f(Y_T)$. 
%
It follows from \eqref{eq:decomp} that there is no loss in performance if the encoder tries to describe the process $\widetilde{X}_T(\cdot)$ subject to the bitrate constraint, rather than the process $X(\cdot)$. In addition, optimal decoding is obtained by outputting the conditional expectation of $\widetilde{X}_T(\cdot)$ given $f(Y_T)$. These observations, which hold in general in indirect source coding situations, were used in \cite{1057738} to derive the indirect DRF of a pair of stationary Gaussian processes, and later in
\cite{1054469} to derive indirect DRF expressions in other settings. An extension of the principle presented in this decomposition to arbitrary distortion measures is discussed in \cite{1056251}. \par
The decomposition \eqref{eq:decomp} is also related to the behavior of $D(S,R)$ under the two extreme cases illustrated in Fig.~\ref{fig:contribution}:
\subsubsection{Unconstrained Bitrate}
As the bitrate $R$ goes to infinity, the MSE as a result of lossy compression goes to zero. Consequently, \eqref{eq:decomp} implies that
\[
\lim_{R\rightarrow \infty} D(S,R) = \inf_{R>0} D(S,R) = \mmse(S),
\]
where $\mmse(S) = \liminf_{T\rightarrow \infty } \mmse_T(S)$. \par
Since the sampling operation is linear and the signals are Gaussian, we have that
\begin{equation} \label{eq:inst_mmse}
    \mathbb E \left( X(t) - \mathbb E [X(t)|Y_T] \right)^2 = \inf_{\mathbf a \in \mathbb R^{N_T}}  \mathbb E \left( X(t) - \sum_{t_n \in \Lambda_T} a_n Y_n \right)^2. 
\end{equation}
Under a MB-LTI sampler, an expression for \eqref{eq:inst_mmse} in the limit as $T$ goes to infinity can be derived in a closed form \cite{ShannonMeetsNyquist, 1090615}, leading to a closed form expression for $\mmse(S)$. 
Although it is unfeasible to obtain $\mmse_T(S)$ in a closed form for an arbitrary bounded linear sampler, it is sometimes possible to derive conditions on the density of $\lambda$ such that $\mmse_T(S)$ converges to zero. In particular, assuming zero noise, $K_H(t,s) = \delta(t-s)$ the identity operator, and $\supp~S_X$ is a finite union of bounded intervals, 
%the sampler $S$ is defined only in terms of $\Lambda$. In this setting, 
the condition on \eqref{eq:inst_mmse}, and hence on $\mmse_T(S)$, to converge to zero are related to a classical problem in sampling theory studied by Beurling \cite{beurling1989collected} and
Landau \cite{Landau1967}. 
In order to see this relation, use \eqref{eq:isometry} to translate the interpolation problem of \eqref{eq:inst_mmse} to the Hilbert space $\mathcal W(S_X)$. Since the support of $S_X(f)$ is a finite union of bounded intervals, interpolation in $\mathcal W(S_X)$ with vanishing MSE is equivalent to the same operation in the Paley-Wiener space of analytic functions whose Fourier transform vanishes outside $\supp S_X$. Specifically, this holds whenever the non-harmonic Fourier basis $\left\{ e^{2\pi i t_n},\, t_n \in \Lambda \right\}$ defines a frame in this Paley-Wiener space, i.e., there exists a universal constant $A>0$ such that the $\mathrm L_2$ norm of each function in this space is bounded by $A$ times the energy of the samples of this function. Landau \cite{Landau1967} showed that a necessary condition for this property is that the number of points in $\Lambda$ that fall within any interval of length $T$ is at least the spectral occupancy of $X(\cdot)$ times $T$, perhaps minus a constant that is logarithmic in $T$. For this reason, this spectral occupancy is now termed the \emph{Landau rate} of $X(\cdot)$, and we denote it here by $f_{\Lnd}$. In the special case where $\supp S_X$ is an interval (symmetric around the origin since $X(\cdot)$ is real), the Landau and Nyquist rates coincide.

\subsubsection{Unconstrained Sampling}
The other lower bound in Fig.~\ref{fig:contribution} describes the case when there is no loss in the sampling operation, so that the distortion is only due to lossy compression and noise. This situation occurs when the process $X|X_\epsilon (\cdot) \triangleq \left\{ \mathbb E \left[ X(t) | X_\epsilon(\cdot) \right],\, t\in \mathbb R \right\}$, whose spectral density is
\begin{equation}
    \label{eq:S_X_Xe}
S_{X|X_\epsilon}(f) = \frac{S_X^2(f)}{S_X(f) + S_\epsilon(f)},
\end{equation}
can be recovered from $Y_T$ with zero MSE as $T\rightarrow \infty$. Note that $X|X_\epsilon (\cdot)$ is a Gaussian stationary process obtained by estimating $X(t)$ using the non-causal Wiener filter. The resulting MSE in this estimation is 
\[
\mmse(X|X_\epsilon) = \sigma_X^2 - \int_{-\infty}^\infty S_{X|X_\epsilon}(f) df.
\]
Since no limitation is imposed on the encoder in Fig.~\ref{fig:system_model} except the bitrate, the  encoder can estimate $X|X_\epsilon (\cdot)$ from $Y_T$ and encode it at bitrate $R$ as in standard source coding. When $T\rightarrow \infty$, the distortion in this procedure is given by \cite{1057738}
\begin{subequations}
\label{eq:dobrushin}
\begin{align} 
    D_{X|X_\epsilon}(R_\theta) & = \mmse(X|X_\epsilon) + \int_{-\infty}^\infty \min \left\{S_{X|X_\epsilon}(f),\theta \right\}df,\\
    R_\theta & = \frac{1}{2} \int_{-\infty}^\infty \log^+ \left[S_{X|X_\epsilon}(f)/\theta \right]df. 
\end{align}
\end{subequations}
In the special case when $\epsilon(\cdot) \equiv 0$, \eqref{eq:dobrushin} reduces to Pinsker's formula \cite{1056823} for the DRF of $X(\cdot)$:
\begin{subequations}
\label{eq:pinsker}
\begin{align} 
    D_{X}(R_\theta) & =  \int_{-\infty}^\infty \min \left\{S_{X}(f),\theta \right\}df,\\
    R_\theta & = \frac{1}{2} \int_{-\infty}^\infty \log^+ \left[S_{X}(f)/\theta \right]df. 
\end{align}
\end{subequations}
Note that \eqref{eq:pinsker} is the continuous-time counterpart of \eqref{eq:D_finite}. \par
When the minimal ADX distortion $D_T(S,R)$ approaches $D_{X|X_\epsilon}(R)$, or $D_X(R)$ in the non-noisy case, as $T\rightarrow \infty$, we say that the conditions for optimal sampling in ADX are met. Namely, optimal sampling occurs whenever
\begin{equation}
    \label{eq:optimal_sampling_cond}
D(S,R) = D_{X|X_\epsilon}(R),
\end{equation}
For example, \eqref{eq:optimal_sampling_cond} holds under MB-LTI sampling with a single sampling branch, provided $f_s \geq f_{\Nyq}$ and the passband of the pre-sampling filter $H(f)$ contains $\supp\, S_X$ (which equals to $\supp\, S_{X|X_\epsilon}$). More generally, it is possible to chose the pre-sampling filters of a MB-LTI sampler such that optimal sampling occurs for any $f_s \geq f_{\Lnd}$ \cite[Sec. IV]{Kipnis2014}, \cite{eldar2015sampling}, where $f_\Lnd$ is the Landau rate of $X(\cdot)$ (or its spectral occupancy). In these cases, we also have that $\mmse_T(S)$ of \eqref{eq:ADX_mmse} converges to $\mmse(X|X_\epsilon)$, which is a sufficient condition for \eqref{eq:optimal_sampling_cond} to hold. As we shall see in the next section, this condition is not necessary, and optimal sampling can be attained by sampling below the Nyquist or Landau rates.

\section{The Fundamental Distortion Limit \label{sec:ADX}}
We now provide the general definition of $D^\star(f_s,R)$, explore its basic properties, and use it to fully characterize the ADX distortion. 

\subsection{Expression for ADX Distortion}

\begin{definition}
For a sampling rate $f_s$ and Gaussian signals $X(\cdot)$ and $\epsilon(\cdot)$, let $F^\star_{f_s}$ be a set of Lebesgue measure $\mu$ not exceeding $f_s$ that maximizes
\begin{equation}
    \label{eq:F_star_def}
\int_{F} S_{X|X_\epsilon}(f) df = \int_{F} \frac{S_X^2(f)}{ S_X(f) + S_\epsilon(f) }df,
\end{equation}
over all sets $F$ with $\mu(F)\leq f_s$. Define
\begin{equation}
    \label{eq:bound_def}
D^\star(f_s,R_\theta) \triangleq \sigma_X^2 - \int_{F^\star_{f_s}}  \left[ S_{X|X_\epsilon}(f)-\theta \right]^+(f) df,
\end{equation}
where $\theta$ is determined by
\[
R_\theta = \frac{1}{2} \int_{F^\star_{f_s}} \log^+ \left[S_{X|X_\epsilon}(f)/\theta \right] df. 
\]
\end{definition}
We also define 
\begin{equation}
    \label{eq:mmse_bound_def}
\mmse^\star(f_s) = \sigma_X^2 - \int_{F^\star_{f_s}} S_{X|X_\epsilon}(f) df,
\end{equation}
and note that 
\[
D^\star(f_s,R) = \mmse^\star(f_s) + \int_{F^\star_{f_s}} \min\left\{S_{X|X_\epsilon}(f),\theta \right\}df. 
\]
Graphical interpretations of $D^\star(f_s,R)$ and $\mmse^\star(f_s)$ are provided in Fig.~\ref{fig:water-filling_lower_bound}.
\begin{figure}
\begin{center}
\input{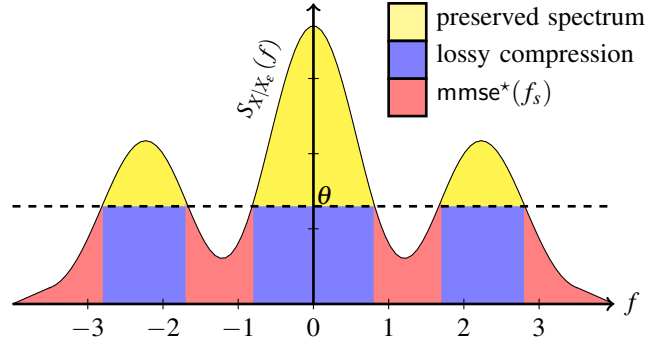}
\caption{
\label{fig:water-filling_lower_bound} 
water-filling interpretation of the fundamental distortion limit
$D^\star(f_s,R)$. The distortion is the sum of the sampling error ($\mmse^\star(f_s)$) and the lossy compression error. The set $F^\star_{f_s}$ defining $D^\star(f_s,R)$ is the support of the preserved spectrum.
} 
\end{center}
\end{figure}
The main results of this paper are summarized by the following two theorems:
\begin{thm}[achievability] \label{thm:achivability}
For any $f_s$ and $\epsilon>0$, there exists a MB-LTI sampler $S$ with sampling rate $f_s$, such that, for any $R$, the distortion in ADX attained by sampling $X_\epsilon(\cdot)$ using $S$ over a large enough time interval $T$, and encoding these samples using $\lfloor TR \rfloor$ bits, does not exceed
$D^\star(f_s,R) + \epsilon$.
\end{thm}
\begin{thm}[converse] \label{thm:converse}
Let $S = (K_H,\Lambda)$ be a bounded linear sampler such that $d^+(\Lambda) \leq  f_s$. Then for any bitrate $R$, \[
D(S,R) \geq D^\star(f_s,R).
\]
\end{thm}
From the definition of $D_T(S,R)$ and the upper symmetric density $d^+(\Lambda)$, Theorem~\ref{thm:converse} implies the following corollary:
\begin{cor} \label{cor:finite_T}
Let $S = (K_H,\Lambda)$ be a bounded linear sampler such that, for every $T>0$, $\card(\Lambda_T)\leq T f_s$. Then for any representation of the samples $Y_T$ using at most $ \lfloor TR \rfloor$ bits, the MSE \eqref{eq:MSE_def} in recovering $X(\cdot)$ is bounded from below by $D^\star(f_s,R)$.
\end{cor}

The proofs of Theorems~\ref{thm:achivability} and \ref{thm:converse} can be found in the appendix. A sketch of these proofs is as follows. To prove Theorem~\ref{thm:achivability}, we use the expression for the ADX distortion under an MB-LTI sampler with $L$ samplers derived in \cite{Kipnis2014}. We then show that for any $\delta>0$, there exists $L$ large enough such that $L$ filters $H_1,\ldots,H_L$ can be chosen to have disjoint supports whose union approximates $F^\star_{f_s}$ in the sense that the difference between $D(S,R)$ and $D^\star(f_s,R)$ is less than $\delta$. 
The converse in Theorem~\ref{thm:converse} is first established for a MB-LTI sampler using results from \cite{Kipnis2014} that characterize the optimal set of pre-sampling filter for a given $S_{X|X_\epsilon}(f)$ and number of sampling branches $L$. Next we consider the distortion attained by a general linear bounded sampler $S=(\Lambda, K_H)$ over a finite time horizon $T$. We bound this distortion from below by the distortion in recovering $X(\cdot)$ over $[-T/2,T/2]$ using an encoding of a periodic extension of the sampling set. We then show that this extension is equivalent to sampling using a specific MB-LTI sampler, so that the bound from the first part of the proof is valid for an arbitrary linear bounded sampler.

Before exploring additional properties of $D^\star(f_s,R)$, it is instructive to consider its behavior under various examples for the PSDs $S_X(f)$ and $S_\epsilon(f)$. 
\begin{example}[rectangular PSD] \label{ex:rect1}
Let $X_{\Pi}(\cdot)$ be the process with PSD 
\begin{equation}
    \label{eq:psd_rect}
S_{\Pi}(f) = \sigma_X^2 \frac{\mathbf 1_{|f|<W}(f)}{2W}. 
\end{equation}
Assume that $\epsilon(\cdot)$ is a flat spectrum noise within the band $[-W,W]$ such that $\gamma \triangleq S_{\Pi}(f)/S_\epsilon(f)$ is the SNR at the spectral component $f$. 
Under these conditions, 
\[
S_{X|X_\epsilon}(f) = \frac{\gamma}{1+\gamma} S_{\Pi}(f),
\]
and the set $F^\star_{f_s}$ that maximizes \eqref{eq:F_star_def} can be chosen as any subset of $[-W,W]$ with Lebesgue measure $f_s$. For simplicity we pick
\begin{equation}
    \label{eq:F_star_rect}
    F^\star_{f_s} = \left\{f\,:\,|f|<f_s/2 \right\},
\end{equation}
and conclude that
\[
D^\star(f_s,R) = \sigma_X^2 \begin{cases} 1 - f_s \left[\frac{\gamma }{2W(1+\gamma)} ,\theta \right]^+ & f_s < 2W,\\
1 - 2W \min \left[ \frac{1}{2W(1+\gamma)} \right]^+  & f_s \geq 2W,
\end{cases}
\]
where $\theta$ is determined by
\[
R = \frac{1}{2} \begin{cases} f_s \left(\log \frac{\sigma_X^2 \gamma}{2W(1+\gamma)} - \log \theta \right) & f_s < 2W, \\
 2W \left(\log \frac{\sigma_X^2 \gamma}{2W(1+\gamma)} - \log \theta \right) & f_s \geq 2W.
\end{cases}
\]
Since $\theta$ can be isolated from the last expression, we obtain
\begin{equation}
    \label{eq:DRF_rect}
D^\star(f_s,R) = \sigma_X^2 \begin{cases} 1 - \frac{f_s}{2W} \frac{\gamma}{1+\gamma}(1-2^{-2R/f_s})  & f_s < 2W,\\
\frac{1}{1+\gamma} + \frac{\gamma}{1+\gamma} 2^{-R/W}  & f_s \geq 2W.
\end{cases}
\end{equation}
We note that in the case $f_s \geq 2W$, \eqref{eq:DRF_rect} equals the DRF of $X_{\Pi}(\cdot)$ given $X_{\Pi}(\cdot) + X_\epsilon(\cdot)$ that is obtained from \eqref{eq:dobrushin}. Therefore only sampling at or above the Nyquist rate $f_\Nyq = 2W$ implies $D^\star(f_s,R) = D_{X|X_\epsilon}(R)$. 
\end{example}

\begin{example}[triangular PSD] \label{ex:triangle1}
Let $X_{\triangle}(\cdot)$ be the process with PSD
\begin{equation}
    \label{eq:psd_triangle}
S_{\triangle}(f) \triangleq \sigma_X^2\frac{\left[1-|f/W| \right]^+}{W},
\end{equation}
for some $W>0$, and assume that $\epsilon(\cdot) \equiv 0$. Then 
\[
F^\star_{f_s} = \left\{f\, : \,|f|<f_s/2  \right\},
\]
and
\begin{equation}
    \label{eq:DRF_triangle}
D^\star(f_s,R) = \sigma_X^2  \begin{cases}
 \left(1-\frac{f_s}{2W}\right)^2  + \theta f_s & f_s \leq f_R, \\
 \left(1-\frac{f_R}{2W}\right)^2 +  \theta f_R  & f_s > f_R,
\end{cases}
\end{equation}
where $f_R \triangleq 2W(1-\theta W)$, and $\theta$ is given by
\[
R = \frac{1}{2} \begin{cases}
 \int_{-\frac{f_s}{2}}^\frac{f_s}{2} \left[ \log S_{\triangle}(f) - \log \theta \right]  & f_s \leq f_R \\ 
  \int_{-\frac{f_R}{2}}^{\frac{f_R}{2}} \log S_\triangle(f) - 
f_R \log \theta  & f_s > f_R.
\end{cases}
\]
\end{example}
%The functions $D^\star(f_s,R)$ for the PSDs $S_\Pi(f)$ and $S_\triangle(f)$ of Examples \eqref{ex:rect1} and \ref{ex:triangle1}, respectively, as well as other PSDs, are illustrated in Fig.~\ref{fig:D_star_PSD}.\\ 

\begin{figure*}
\begin{center}
\input{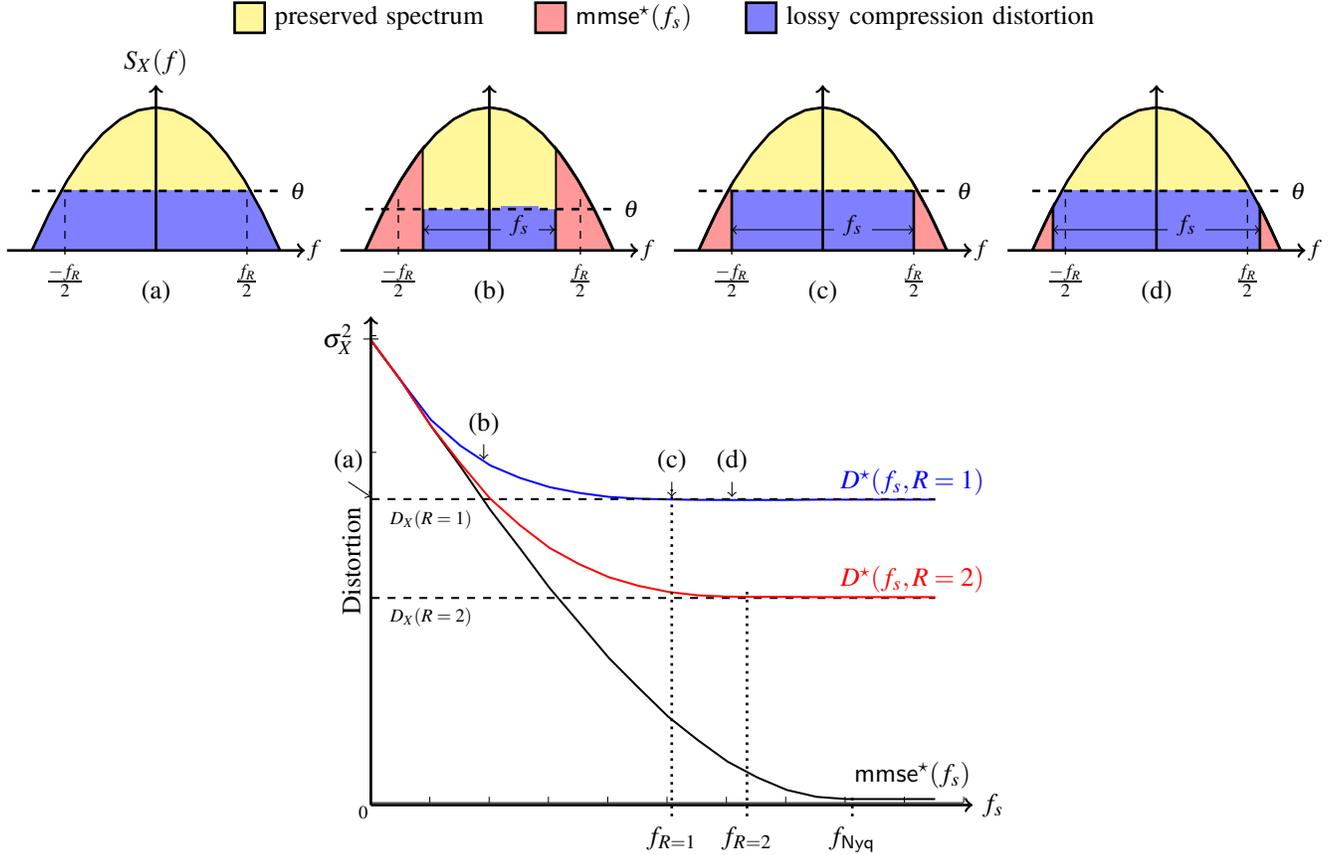}
\caption{\label{fig:proof_sketch} An illustration of Proposition~\ref{prop:opitmal_sampling_rate} for a unimodal PSD and zero noise: the distortion is the sum of the term $\mmse^\star(f_s)$ (red) and the lossy compression distortion (blue) in each water-filling scheme. The functions $D^\star(f_s,R)$ and $\mmse^\star(f_s)$ are illustrated versus $f_s$ at the bottom-right for two fixed values of the bitrate $R$. Also shown is the DRF of $X(\cdot)$ at these values that is attained at the sub-Nyquist sampling rates marked by $f_R$. The marked points on the curve of $D^\star(f_s,R)$ correspond to different water-filling scenarios.
%: (a) Fix $R$ and find $D_X(R)$ from \eqref{eq:pinsker}. (b) $D^\star\left(f_s,R \right) > D_X(R)$ for $f_s < f_R$. (c) $f_s=f_R$. (d) $D^\star(f_s,R)=D(R)$ for all $f_s \geq f_R$. %The function $D^\star(f_s,R)$ versus $f_s$ is illustrated at the bottom-right for two fixed values of the bitrate $R$. 
%\label{fig:DRF_optimal}
}
\end{center}
\end{figure*}

Note that in Example~\ref{ex:triangle1}, the function $D^\star(f_s,R)$ in \eqref{eq:DRF_triangle} is independent of $f_s$ for the case $f_s\geq f_R$, and equals to the DRF of $X_\triangle(\cdot)$ given by Pinsker's expression \eqref{eq:pinsker}. Consequently, for $X_\triangle(\cdot)$, the DRF is attained by sampling above $f_R$ that is smaller than $2W$, which is the Nyquist rate of $X_\triangle(\cdot)$. Since the DRF is the minimal distortion subject only to the bitrate constraint regardless of the sampling mechanism, we conclude that the optimal distortion performance is attained by sampling below the Nyquist rate in this case. In the following subsection, we extend this observation to arbitrary PSDs.

\subsection{Optimal Sampling Rate  \label{sec:main_result}}
We now consider the minimal sampling rate that lead to optimal sampling in ADX. We first note the following proposition, that follows from the definition of $D^\star(f_s,R)$.
\begin{prop}[optimal sampling rate] \label{prop:opitmal_sampling_rate}
For each point $(R,D)$ on the graph of $D_{X|X_\epsilon}(R)$ associated with a water-level $\theta$ via \eqref{eq:dobrushin}, define
\[
F_\theta \triangleq \left\{f: S_{X|X_\epsilon}(f) > \theta \right\},
\]
and set $f_R = \mu (F_\theta)$. Then for all $f_s \geq f_R$,
\[
D^\star(f_s,R) = D_{X|X_{\epsilon}}(R).
\]
\end{prop}

The proof of Proposition~\ref{prop:opitmal_sampling_rate} is given in the Appendix. To gain some intuition into the results, consider the special case of zero noise and a unimodal $S_X(f)$ as illustrated in Fig.~\ref{fig:proof_sketch}: fix a point $(R,D)$ on the distortion rate curve of $X(\cdot)$ obtained from \eqref{eq:pinsker}. The set $F_\theta = \left\{ f\in \mathbb R\,:\, S_X(f)>\theta \right\}$ is the support of the non-shaded area in Fig.~\ref{fig:proof_sketch}$(a)$. We define the sampling rate $f_R$ to be the Lebesgue measure of $F_\theta$. Fig.~\ref{fig:proof_sketch}$(b)$ shows the function $D^\star(f_s,R)$ for $f_s<f_{R}$, where the overall distortion is the sum of the term $\mmse^\star(f_s)$ given by the partially shaded area, and the water-filling term given by the blue area. 
Figs.~\ref{fig:proof_sketch} $(c)$ and $(d)$ show the function  $D^\star(f_s,R)$ for $f_s=f_{R}$ and $f_s>f_{R}$, respectively. The assertion of Proposition~\ref{prop:opitmal_sampling_rate} is that the sum of the red area and the blue area stays the same for any $f_s\geq f_R$. It can also be seen from Fig.~\ref{fig:proof_sketch} that $f_{R}$ increases with the source coding rate $R$ and coincides with $f_{\Nyq}$ as $R\rightarrow \infty$. The bottom-right of Fig.~\ref{fig:proof_sketch} shows
$D^\star(f_s,R)$ as a function of $f_s$ for two fixed values of $R$. \\ %The points marked by (b), (c) and (d) on the graph of $D^\star(f_s,R)$ correspond to the 

We emphasize that the critical frequency $f_R$ arising from Proposition~\ref{prop:opitmal_sampling_rate} depends only on the PSD and on the operating point on the DRF curve of $X(\cdot)$ given $X_\epsilon(\cdot)$, which can be parametrized by either $D$, $R$ or the water-level $\theta$ using \eqref{eq:dobrushin}. In fact, by inverting the function $D^\star(f_s,R)$ with respect to $R$, we obtain the following result.
\begin{thm}[rate-distortion lower bound] \label{thm:opitmal_sampling_rate_R} 
Given Gaussian stationary processes $X(\cdot)$ and $\epsilon(\cdot)$, sampling rate $f_s$ and a target distortion $D > \mmse^\star(f_s)$, define
\begin{equation} \label{eq:R_underline}
R^\star(f_s,D) \triangleq \begin{cases} \frac{1}{2} \int_{F^\star_{f_s}}  \log^+ \left(\frac{f_sS_{X|X_\epsilon}(f)}{D-\mmse^\star(f_s)} \right) df, &  f_s < f_{R}, \\
R_{X|X_\epsilon}(D), &  f_s \geq f_{R},
\end{cases}
\end{equation}
where 
\[
R_{X|X_\epsilon}(D_\theta) = \frac{1}{2}\int_{-\infty}^\infty \log^+\left[ S_{X|X_\epsilon}(f)/\theta\right]df,
\]
is the indirect rate-distortion function of $X(\cdot)$ given $X_\epsilon(\cdot)$, $f_R = \mu\left( \left\{ f\,:\, S_{X|X_\epsilon}(f) > \theta \right\} \right)$, and $\theta$ is determined by
\[
D = \mmse(X|X_\epsilon) + \int_{-\infty}^\infty \min\left\{ S_{X|X_\epsilon}(f),\theta\right\} df.
\]
Then: 
\begin{itemize}
    \item[(i)] The number of bits per unit time required to attain ADX distortion at most $D$ is at least $R^\star(f_s,D)$.
    \item[(ii)] For any $\epsilon>0$ and $\rho>0$, there exists $T$ large enough and a MB-LTI sampler $S$ at rate $f_s$ such that
    \[
    D_T(S,R^\star(f_s,D)+\rho) < D+\epsilon.
    \]
\end{itemize}
\end{thm}
\begin{proof}
Theorem~\ref{thm:opitmal_sampling_rate_R} is a restatement of Theorems \ref{thm:converse} and \ref{thm:achivability} that is obtained using  Proposition~\ref{prop:opitmal_sampling_rate} and by inverting the role of the distortion and the bitrate. 
\end{proof}

%It follows from Proposition~\ref{prop:properties}-(v) that $\underline{R}(f_s,D)$ is attained by multi-branch sampling with optimized pre-sampling filters.

\subsection{Discussion}
Theorem~\ref{thm:achivability} together with Proposition~\ref{prop:opitmal_sampling_rate} extend the conditions for the equality \eqref{eq:optimal_sampling_cond}, which, as argued in Subsection~\ref{subsec:properties}, holds for $f_s \geq f_{\Lnd}$, to all sampling rates above $f_R$. This $f_R$ depends on the bitrate $R$ and is smaller than $f_\Lnd$ provided the signal power is not uniformly distributed over its spectral occupancy (unlike $S_\Pi(f)$ of Example~\ref{ex:rect1}). 
\par
As $R$ goes to infinity, the water-level $\theta$ goes to zero, the set $F_\theta$ coincides with the support of $S_{X|X_\epsilon}(f)$ and $S_X(f)$, and $D^\star(f_s,R)$ converges to $\mmse^\star(f_s)$. In particular, $f_R = \mu(F_\theta)$ converges to the spectral occupancy $f_{\Lnd}$ of $X(\cdot)$. In this limit, Proposition~\ref{prop:opitmal_sampling_rate} then implies that $\mmse^\star(f_s) = \mmse(X|X_\epsilon)$ for all $f_s \geq f_{\Lnd}$. When the noise is zero, this last fact agrees with the Landau's necessary condition for stable sampling in the Paley-Wiener space \cite{Landau1967}. 
\\

The discussion in Section~\ref{sec:finite_dims} on the finite-dimensional counterpart of ADX suggests the following intuition for our result assuming $\epsilon(\cdot) \equiv 0$: Pinsker's waterfilling expression \eqref{eq:pinsker} implies that for a Gaussian stationary signal whose power is not uniformly distributed over its spectral occupancy, the optimal distortion-rate tradeoff $D_X(R)$ is achieved by communicating only those bands with the highest energy. This means that less degrees of freedom are used in the signal's representation. Proposition~\ref{prop:opitmal_sampling_rate} implies that this reduction in degrees of freedom can be translated to a lower required sampling rate in order to achieve $D_X(R)$. The counterpart of this phenomena in the finite dimensional case is the conditions for equality between \eqref{eq:D_finite} and \eqref{eq:D_finite_regular} as discussed in Section~\ref{sec:finite_dims}. \\

\subsection{Examples}
In the following examples the exact dependency of $f_R$ on $R$ and $D$ is determined for various PSDs, and illustrated in Fig.~\ref{fig:critical_f}.
\begin{example}[continuation of Example~\ref{ex:rect1}] \label{ex:rect2}
Consider the PSDs $S_\Pi(f)$ and $S_\epsilon(f) = S_\Pi(f)/\gamma$ as in Example~\ref{ex:rect1}. In this case we have that $F_\theta = [-W,W]$ for all $f_s$, so that $f_{R}  = 2W$. Therefore, in this example,  $D^\star(f_s,R) = D_{X|X_\epsilon}(R)$ only for $f_s$ larger than the Nyquist rate of $X_\Pi(f)$. This observation agrees with the expression for $D^\star(f_s,R)$ in \eqref{eq:DRF_rect}.
\end{example}
\begin{example}[continuation of Example~\ref{ex:triangle1}] \label{ex:triangle2}
Consider the situation of Example~\ref{ex:triangle1} with zero noise and PSD $S_\triangle(f)$. For a point $(R,D)\in [0,\infty) \times [0,\sigma_X^2]$ on the distortion-rate curve of $X_\triangle(\cdot)$, we have that $F_\theta = W\left[ -1+W \theta,1-W\theta\right]$ and hence $f_R = 2W(1-W\theta)$. Indeed, this value for $f_R$ agrees with \eqref{eq:DRF_triangle}, since for $f_s\geq f_R$ the function $D^\star(f_s,R)$ is independent of $f_s$ and equals the DRF of $X_\triangle(\cdot)$. \par
The exact relation between $R$ and $f_R$ is given by
\begin{align} 
R & = \frac{1}{2} \int_{-\frac{f_{R}}{2}}^\frac{f_{R}}{2} \log\left( \frac{1-|f/W|}{1-\frac{f_R}{2W}} \right)df \nonumber \\
& =   W  \log\frac{1}{1-\frac{f_R}{2 \ln 2}}-\frac{f_R}{2W} . \label{eq:R_fc_triangle}
\end{align}
Expressing $f_R$ as a function of the distortion $D$ leads to $f_R = 2W\sqrt{1-D/\sigma_X^2}$.  \\
\end{example}
\begin{example}[effect of noise on $f_R$]\label{ex:triangle_noise}
Consider again $X_\triangle(\cdot)$ from Examples~\ref{ex:triangle1} and \ref{ex:triangle2}, but with $\epsilon(\cdot)$ a flat spectrum Gaussian noise with intensity $\sigma^2_{\epsilon}$, i.e., 
$S_\epsilon(f) = \sigma_\epsilon^2 \mathbf 1_{[-W,W]}$. The relation between $R$ and $f_R$ is given by:
\begin{align*}
R & = \int_{-\frac{f_{R}}{2}}^\frac{f_{R}}{2} \log \left[ \frac{(1- \frac{f}{W})^2} {1- \frac{f}{W}+W\sigma_\epsilon^2}  \right] df- f_R \log \left[
 \frac{(1-\frac{f_{R}}{2W} )^2} {1-\frac{f_{R}}{2 W}+W \sigma^2_\epsilon}  
 \right] \\
 & =  2W \log \frac{1}{1-\frac{f_{R}}{2W}} -W(1+\sigma^2_\epsilon W) \log\frac{1} {1-\frac{f_{R}}{2W (1+\sigma^2_\epsilon W) } }  -\frac{f_{R}} {2 \ln 2}.
\end{align*}
The expression above decreases as the intensity of the noise $\sigma^2_\epsilon$ increases. Since $f_{R}$ increases with $R$, it follows that $f_{R}$ decreases in $\sigma^2_\epsilon$, as can be seen in Fig.~\ref{fig:critical_f_noise} where $f_{R}$ is plotted versus the SNR $1/\sigma^{2}_\epsilon$ for two fixed values of $R$.  
\end{example}
The dependency between the critical sampling rate $f_R$ and the SNR observed in Example~\ref{ex:triangle_noise} can be generalized to any signal PSD experiencing a uniform increase in the SNR: increase in SNR decreases $\mmse(X|X_\epsilon)$ and leads to the use of more spectral bands in the optimal indirect source coding scheme that attains $D_{X|X_\epsilon}(R)$. As a result, more spectral bands of $S_{X|X_\epsilon}(f)$ must be utilized in order for $D^\star(f_s,R)$ to approach $D_{X|X_\epsilon}(R)$. 
\begin{example}[Gauss-Markov process] \label{ex:GaussMarkov1}
In this example we assume $\epsilon(\cdot) \equiv 0$ and consider the Gauss-Markov process $X_\Omega(\cdot)$ whose PSD is
\begin{equation} \label{eq:psd_gaussmarkov}
S_{\Omega}(f) = \frac{1/f_0}{(\pi f/f_0)^2+1},
\end{equation}
for some $f_0>0$. Note that the support of $S_\Omega(f)$ is the entire real line, and therefore the Nyquist rate of $X_\Omega(\cdot)$ is infinite. That is, it is impossible to recover $X_\Omega(\cdot)$ from its samples over any uniformly discrete set. In particular, the MMSE in recovering $X_\Omega(\cdot)$ from its uniform samples at rate $f_s$ equals the area bounded by the tails of its PSD:
\begin{equation}
\label{eq:gaussmarkov_mmse}
\mmse^\star(f_s) = 2\int_{f_s/2}^\infty S_{\Omega}(f) df=  1-\frac{2}{\pi} \arctan\left(\frac {\pi f_s} {2f_0} \right) .
\end{equation}

For a point $(R,D)$ on the distortion-rate curve of $X_\Omega(\cdot)$ and its corresponding $\theta$, we  have $f_R = \frac{2f_0}{\pi}\sqrt{\frac{1}{\theta f_0}-1}$. This means that the distortion cannot be reduced more by sampling above this rate. The exact relation between $R$ and $f_R$ is given by
\begin{equation}
\label{eq:fDR_R_GaussMarkov}
R = \frac{1}{\ln 2} \left( f_{R} -2 f_0 \frac{\arctan\left(\frac{\pi f_{R}}{2f_0} \right)}{\pi} \right).
\end{equation}
This relation is illustrated in Fig.~\ref{fig:critical_f}, and reveals an interesting phenomena: although the Nyquist rate of $X_\Omega(\cdot)$ is infinite, for any finite $R$ there exists a critical sampling frequency $f_R$, satisfying \eqref{eq:fDR_R_GaussMarkov}, such that the DRF of $X_\Omega(t)$ can be attained by sampling at or above $f_R$. Namely, when the non-bandlimited signal $X_\Omega(\cdot)$ is considered under a bitrate constraint $R$, there exists a sampling scheme at finite rate $f_s$ such that the overall distortion in the system equals the minimal distortion subject only to the bitrate constraint. 
%Since the distortion in ADX is bounded from below by $D_X(R)$ and $\mmse^\star(f_s)$, both $f_s$ and $R$ must goes to infinity in order to describe a realization of $X_\Omega(\cdot)$ with vanishing distortion. In this asymptotic regime, it follows from \eqref{eq:fDR_R_GaussMarkov} that the relation between the critical sampling rate $f_R$ and $R$ is $R = f_s/\ln2 + o(1/f_R)$. Therefore, for $R$ sufficiently large, the digital representation of the samples of $X_\Omega(\cdot)$ must allocate $R/f_s = 1/\ln2\approx 1.45$ bits per sample in order to attain the DRF of $X_\Omega(\cdot)$. If the number of bits per sample goes below this value, then, as $R$ and the sampling rate increase, the distortion is dominated by the lossy compression distortion as there are not enough bits to represent the information acquired by the sampler. Consequently, as illustrated in Figure~\ref{fig:bit_per_sample}, the ratio between $D^\star(f_s,R)$ and $D_X(R)$ goes to one in this case. If the number of bits per sample is above $1/\ln2$, then the distortion is dominated by the sampling distortion, and the ratio between $D^\star(f_s,R)$ and $D_X(R)$ converges to a constant bigger than one as illustrated in Figure~\ref{fig:bit_per_sample}. 
\end{example}

\begin{figure}
\begin{center}
\begin{tikzpicture}[scale=1]
\node at (3.82,2.3) {\includegraphics[scale=0.5, trim = 0cm 0cm 0cm 3cm, clip = true]{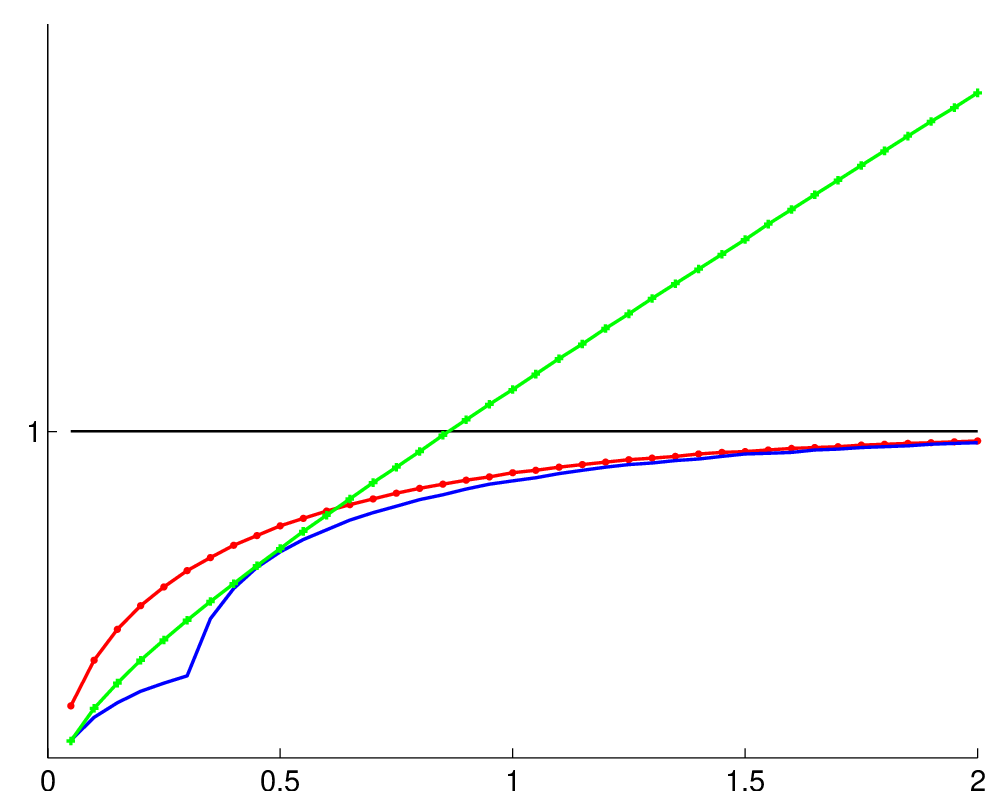} };

\node (rect) at (1,4.7) {\includegraphics[scale=0.1]{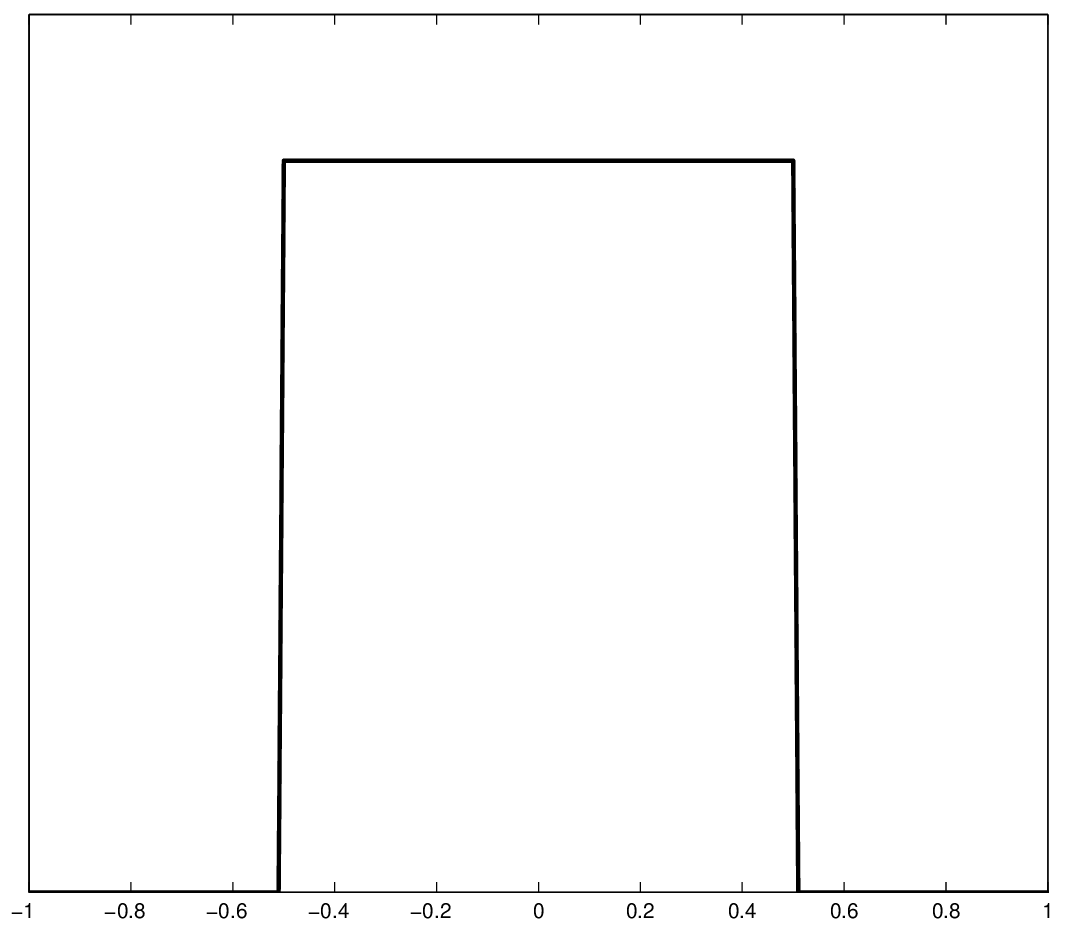} };
\node (triangle) at (3,4.7) {\includegraphics[scale=0.1]{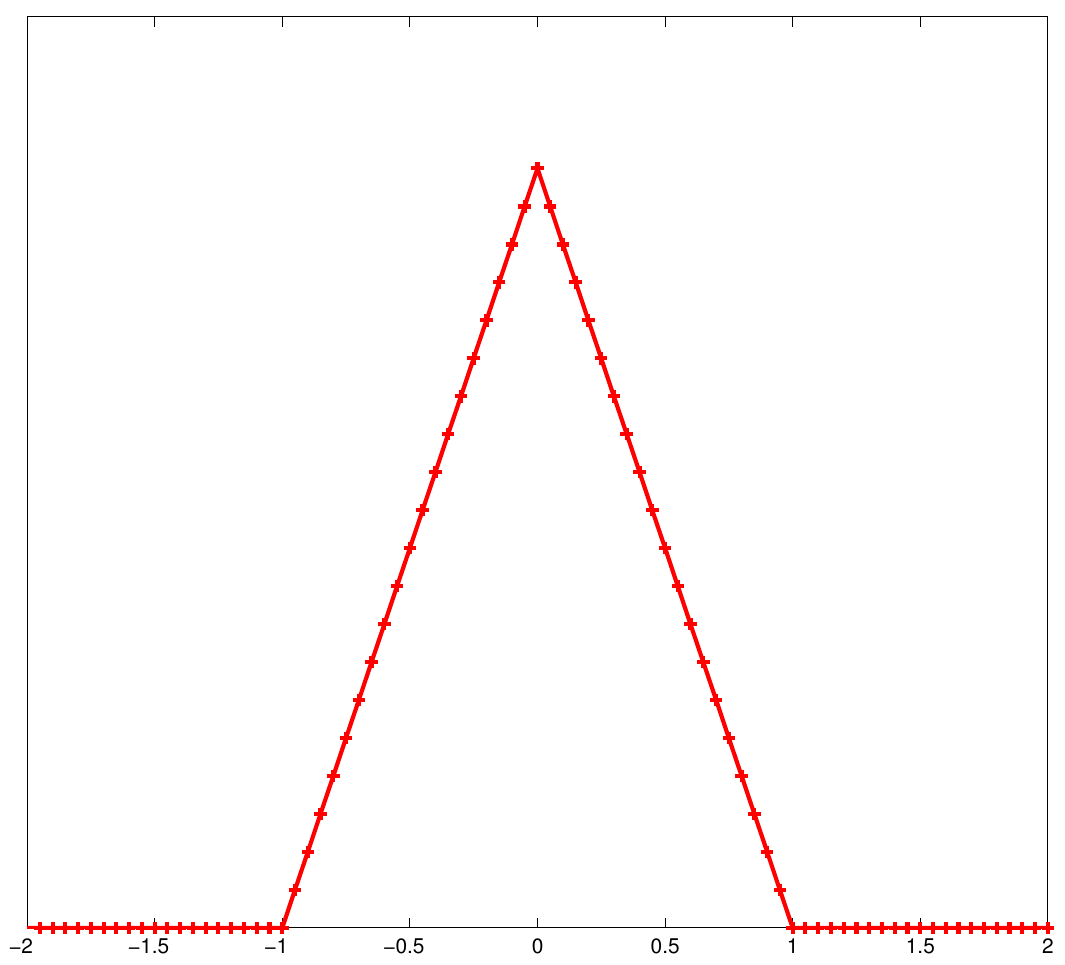} };
\node (gauss) at (5,4.7) {\includegraphics[scale=0.1]{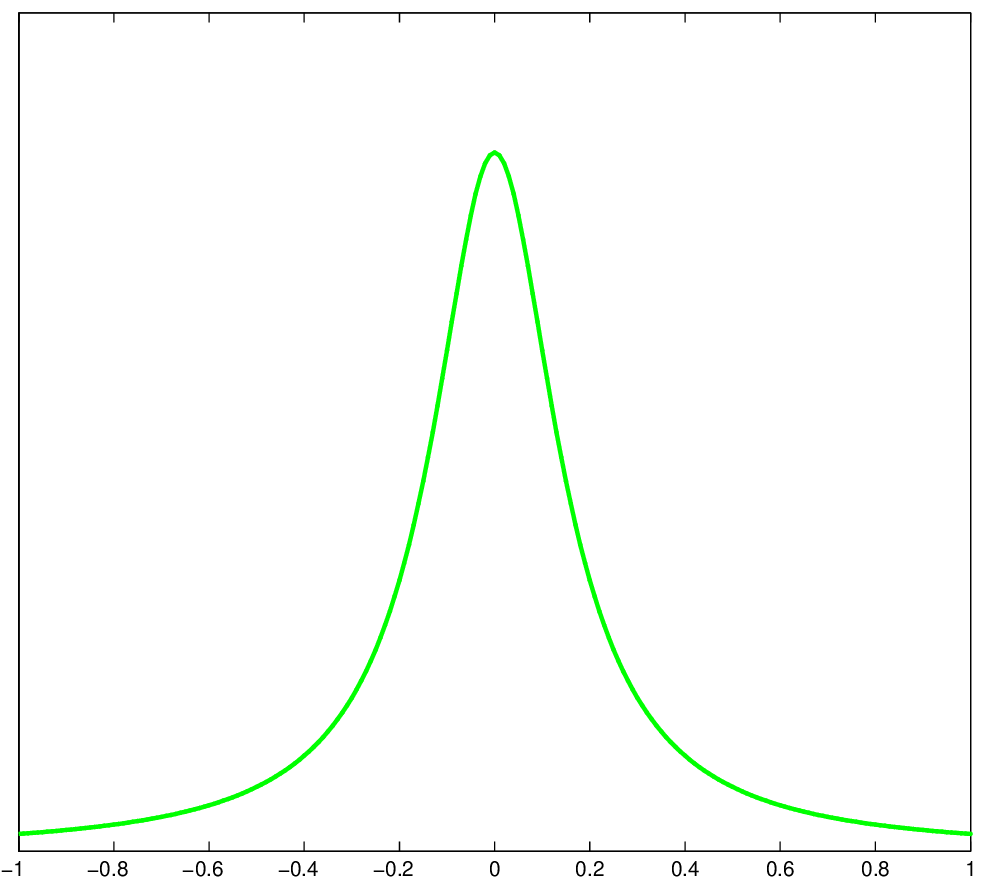} };
\node (wiggly) at (7,4.7) {\includegraphics[scale=0.1]{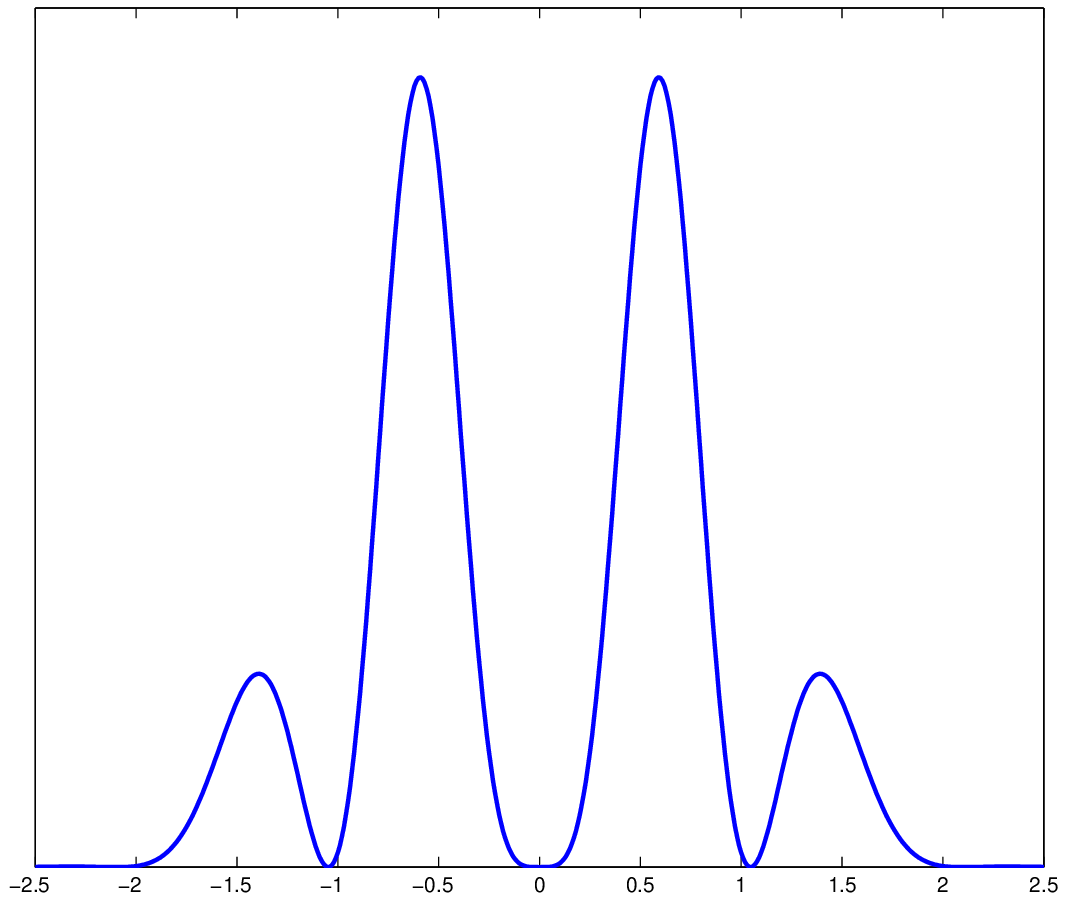} };

\node[above of = rect, node distance = 1cm] {\small $S_\Pi(f)$};
\node[above of = triangle, node distance = 1cm] { \color{red} \small $S_\triangle(f)$};
\node[above of = gauss, node distance = 1cm] {\color{green} \small $S_\Omega(f)$};
\node[above of = wiggly, node distance = 1cm] {\color{blue} \small $S_\omega(f)$};

\draw[->, line width=1pt] (0,0) -- node[below,yshift=-0.2cm] {\small $R$ \small [bit/sec]} (7.9,0);

\draw[->, line width=1pt] (0,0) -- node[above,rotate=90, yshift = 0.4cm] {\small $f_{R}$ \small [smp/sec]} (0,4.8);
\end{tikzpicture}
\vspace{-5pt}
\caption{\label{fig:critical_f} The critical rate $f_R$ as a function of the source coding rate $R$ for the PSDs given in the small frames and zero noise ($\epsilon(\cdot) \equiv 0$). $S_\Pi(f)$, $S_\triangle(f)$ and $S_\omega(f)$ have the same bandwidth while the support of $S_\Omega(f)$ is unbounded.}
\end{center}
\end{figure}
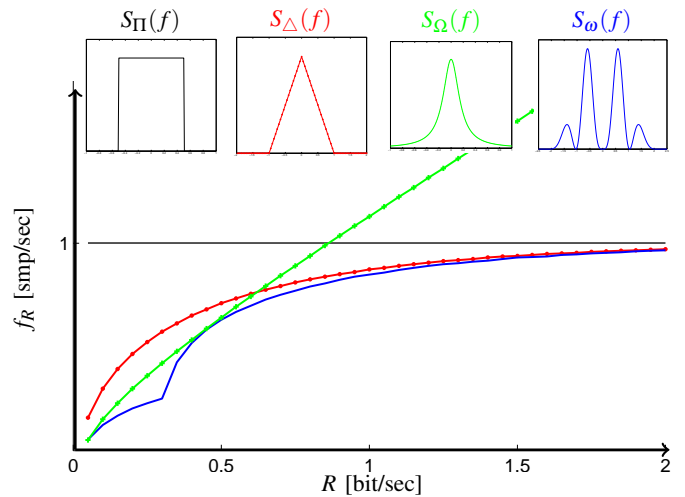

\begin{figure}
\begin{center}
\begin{tikzpicture}[scale=1]
\node at (3.05,3.05) {\includegraphics[scale=0.5, trim = 0cm 0cm 0cm 0cm, clip = true]{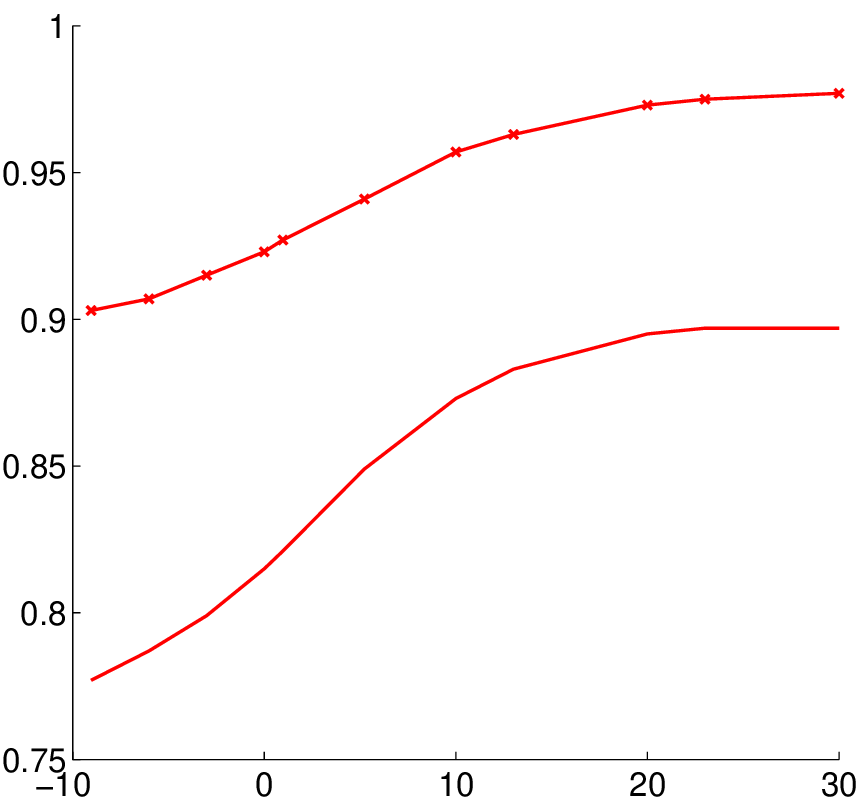}};
\draw[dashed] (0,3.7) -- (7,3.7);
\draw[dashed] (0,5.65) -- (7,5.65);
\node (triangle) at (5.5,1.5) {\includegraphics[scale=0.1]{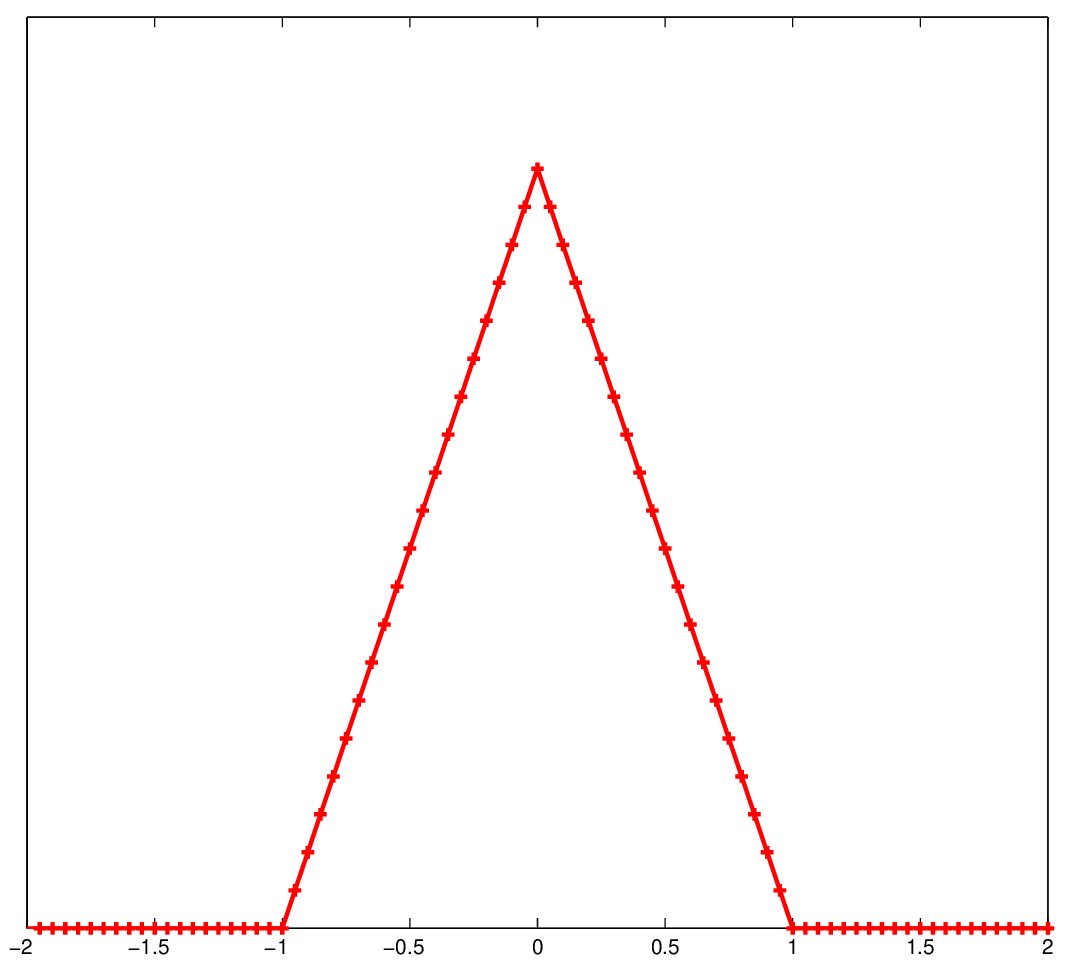} };
%\node[above of = rect, node distance = 1cm] {\small $S_\Pi(f)$};
\node[above of = triangle, node distance = 1cm] { \color{red} \small $S_\triangle(f)$};
%\node[above of = gauss, node distance = 1cm] {\color{green} \small $S_\Omega(f)$};
%\node[above of = wiggly, node distance = 1cm] {\color{blue} \small $S_X(f)$};

\node[rotate = 0] at (6,5.9) {\small $R=2$};
\node[rotate = 0] at (6,3.9) {\small $R=1$};

\draw[->, line width=1pt] (0,0) --  (0,6.5) node[above] {\small $\frac{f_{R}}{2W}$} ;
\draw[->, line width=1pt] (0,0) --  (7,0) node[right] {\small $\frac{\sigma_X^2}{\sigma_\epsilon^2}$ \small [dB] };

\end{tikzpicture}
\vspace{-5pt}
\caption{\label{fig:critical_f_noise} The critical rate $f_{R}$ as a function of the SNR $\sigma_X^2 / \sigma_{\epsilon}^2$ for an input signal with PSD $S_\triangle(f)$ corrupted by a flat spectrum Gaussian noise for two fixed values of $R$. The dashed line corresponds to the value of $f_{R}$ without noise. }
\end{center}
\end{figure}
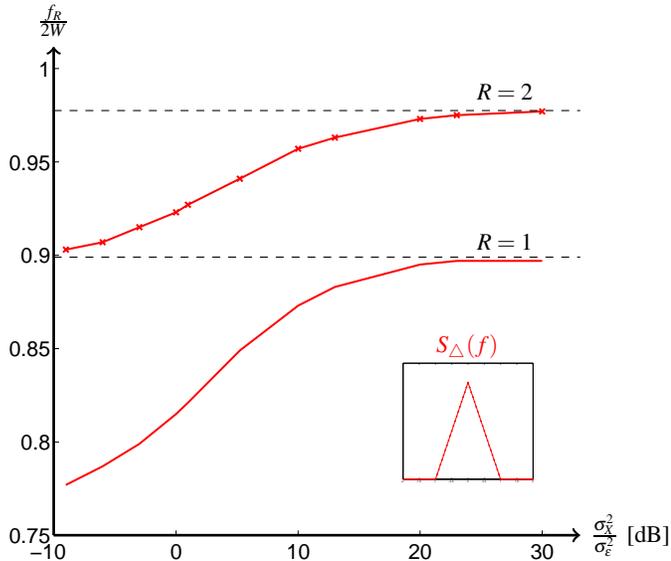

%\begin{example}[memoryless source corrupted by  colored noise]
%{ \color{red} Finish this example}
%In this example the source is the memoryless PSD $S_\Pi(f)$ of \eqref{eq:psd_rect} and the noise process $X_\epsilon(\cdot)$ has PSD 
%\[
%S_\epsilon(f) = 
%\]
%Consider the source with triangular PSD $S_\triangle(f)$ defined in \eqref{eq:psd_triangle} and a white Gaussain noise $\epsilon(\cdot)$ with intensity $S_\epsilon(f) = \sigma_\epsilon^2$. 
%\end{example}

\section{Pulse-Code Modulation \label{sec:PCM}}
So far we considered the conversion of analog signals to bits using bounded linear sampling and under optimal encoding of these samples to bits, subject only to the bitrate constraint $R$. In particular, we did not impose any limitations on the complexity or delay of the encoder and decoder aside from the bitrate at the encoder's output. Indeed, the achievablity of $D^\star(f_s,R)$ in Theorem~\ref{thm:achivability} is obtained as the time horizon $T$ grows to infinity, whereas the number of states assumed by the encoder and decoder grows exponentially in $T$. 
% Moreover, when the PSD of $X(\cdot)$ is not unimodal, the number of sampling branches may be arbitrarily large. 
\par  
In this section we are interested in imposing additional constraints on the restricted-bitrate representation of the samples and the recovery of $X(\cdot)$ beyond those associated with the achievable scheme of Theorem~\ref{thm:achivability}. Specifically, we now assume that the samples are obtained using a single sampling branch, the encoder maps each sample $Y_n$ to its finite-bit representation $\hat{Y}_n$ at time $n$ using a scalar quantizer with a fixed number of bits $q$, and the decoder recovers $X(\cdot)$ using a linear procedure. This form of encoding is known as \emph{pulse-code modulation} (PCM) \cite{black1947pulse,1697556}; we refer to 
\cite[Sec I.A]{gray1998quantization} for a historical overview. In order to focus on the effect of this sub-optimal encoding and decoding on the distortion-rate performance, we assume that no noise is added to $X(\cdot)$ prior to sampling. The extension of the distortion analysis below to the case in which such a noise is present is straightforward. 

%These restrictions on the ADX can be seen as using the sampler of Fig~\ref{fig:sampler_multi} with $L=1$, taking the encoder in Fig.~\ref{fig:system_model} to be a $q$-bit scalar quantizer, and an estimator that interpolate the quantized samples using a sinc kernel. In particular, the output of the encoder is a sequence of quantized amplitude values of source samples taken at uniform intervals. This form of encoding is known as \emph{pulse-code modulation} (PCM) \cite{1697556}, and is commonly used in many electronic devices to convert audio signals to bits. \\

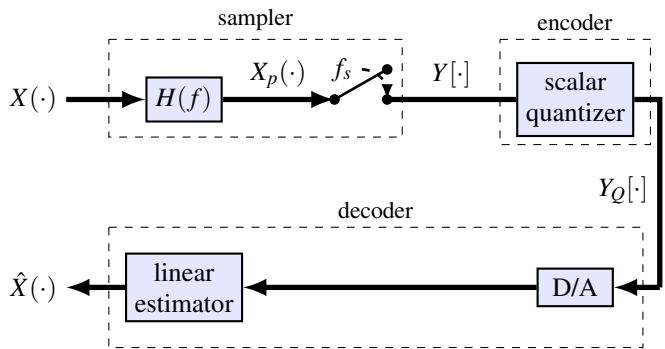
\begin{figure}
\begin{center}
\begin{tikzpicture}[node distance=2cm,auto,>=latex]
  \node at (0,0) (source) {$X(\cdot)$} ;
  \node[int1, right of = source, node distance =2cm] (anti) {$H(f)$};   
 \node [coordinate,right of = anti, node distance = 4cm] at (0,0) (smp_in) {};
  \node [coordinate, right of = smp_in,node distance = 0.7cm] (smp_out){};
	\node [coordinate,above of = smp_out,node distance = 0.4cm] (tip) {};
\fill  (smp_out) circle [radius=2pt];
\fill  (smp_in) circle [radius=2pt];
\fill  (tip) circle [radius=2pt];
\node[left,left of = tip, node distance = 0.6 cm] (ltop) {$f_s$};
\draw[->,dashed,line width = 1pt] (ltop) to [out=0,in=90] (smp_out.north);
\draw[line width=1pt]  (smp_in) -- (tip);
\draw[dashed] (smp_in)+(-3,0.8) -- node[above] {\small sampler} +(0.9,0.8) --  +(0.9,-0.5)-- +(-3,-0.5) -- +(-3,0.8);

\node[int1,right of=smp_out, node distance = 2.5cm, align = center] (enc) {scalar \\  quantizer};
%\node [above of=plus_out, node distance=1.3cm] (eta) {$\eta[n]$};

\draw[dashed] (enc)+(-1,0.8) -- node[above, align = center] {\small encoder} +(1,0.8) -- +(1,-0.6)-- +(-1,-0.6) -- +(-1,0.8);

\node [right of = enc, node distance = 1cm] (right_edge) {};
\node [below of = right_edge, node distance = 2.5cm] (right_b_edge) {};
\node [right] (dest) [below of=source, node distance = 2.5cm]{$\hat{X}(\cdot)$};

\node [int1] (LPF) [below of=anti, node distance = 2.5cm, align = center] {linear \\ estimator};

\node[int1,below of = enc, node distance = 2.5cm] (dec) {D/A};

\draw[-,line width=2pt] (smp_out) -- node[above] {$ Y[\cdot]$} (enc);
\draw[-,line width=2pt] (enc) -- 
%node[above, xshift = 0.9cm]{$R=\bar{R} f_s [\frac{\mathrm{bits}}{\mathrm{sec}}]$}
(right_edge);

\draw[dashed] (LPF)+(-1,0.8) -- node[above, align = center] {\small decoder} +(6.1,0.8)-- +(6.1,-0.8)-- +(-1,-0.8) -- +(-1,0.8);

%\draw[dotted, line width = 1] (0.9,0.3) -- (0.9,-3)  node[left] {analog~} ; 
%\draw[dotted, line width = 1] (7.2,0.3)  -- (7.2,-3)  node[right] {~digital} ;

\draw[-,line width = 2]  (right_edge.west) -|  (right_b_edge.east) node[left, yshift = 1.3cm] {$Y_Q[\cdot]$};
\draw[->,line width = 2]  (right_b_edge.east) -- (dec);
\draw[->,line width = 2]  (dec) -- (LPF);
\draw[->,line width=2pt] (LPF) -- (dest);
\draw[->,line width=2pt] (source) -- (anti);
\draw[->,line width=2pt] (anti) -- node[above] {$\scriptsize X_p(\cdot)$} (smp_in);
\end{tikzpicture}
\end{center}
\caption{\label{fig:PCM_system} Pulse-code modulation  and reconstruction system.}
\end{figure}

\subsection{PCM A/D conversion and Reconstruction Setup}
We consider the system described in Fig.~\ref{fig:PCM_system}, where the input $X(\cdot)$ is assumed to be a wide-sense stationary stochastic process with PSD $S_X(f)$, not necessarily Gaussian. This process is sampled using a pre-sampling filter $H(f)$ followed by a uniform sampler with sampling rate $f_s$. This is a special case of the multi-branch LTI uniform sampler of Fig.~\ref{fig:sampler_multi} with $L=1$ and $H_1(f) = H(f)$. The sample $Y[n]$ at time $n/f_s$ is mapped to a quantization level $\hat{Y}[n]$ using a procedure denoted as \emph{fixed-length scalar quantization} \cite{gray1998quantization}: we consider a set of $M$ real numbers $\hat{y}_1,\ldots,\hat{y}_M$ called \emph{reconstruction levels}. Each reconstruction level is assigned a digital number of length $q \triangleq \lceil \log M \rceil$, where $q$ is the \emph{bit resolution} of the quantizer. Upon receiving the input $Y[n]$, the quantizer outputs the nearest reconstruction level to $Y[n]$ among the set of reconstruction levels, what we denote $\hat{Y}[n]$. Using this notation, the \emph{bitrate} of the digital representation, namely, the number of bits per unit time required to represent the process $\hat{Y}[\cdot]$, is $R = q f_s$. \par
Denote by $\eta[n]$ the quantization error, i.e.,
\begin{equation} \label{eq:quant_model}
\hat{Y}[n] = Y[n] + \eta[n], \quad n\in \mathbb Z.
\end{equation}
The variance of $\eta[n]$ depends on the square of the size of the quantization regions induced by the quantizer, i.e., the Voronoi sets associated with the reconstruction levels. The number of these sets increases exponentially in the bit resolution $q$ and so does the radius of each set, provided all radii decrease uniformly \cite{930935}. As a result, the variance of $\eta[n]$ behaves as
\begin{equation} \label{eq:noise_var}
\sigma_{\eta}^2 = c_0 2^{-2q},
\end{equation}
for some $c_0>0$. The constant $c_0$ depends on other statistical assumptions on the input to the quantizer. For example, if the amplitude of the input signal is bounded within the interval $(-A_m/2,A_m/2)$, then we may choose uniformly spaced quantization levels resulting in $c_0=A_m / 12$. If the input to the quantizer is Gaussian with variance $\sigma_{in}^2$ and the quantization rule is chosen according to the ideal point density allocation of the Lloyd algorithm \cite{1056489}, then \cite[Eq. 10]{gray1998quantization}
\begin{equation} \label{eq:quantizer_var}
c_0 = \frac{\pi\sqrt{3}}{2} \sigma_{in}^2.
\end{equation}

The non-linear relation between $Y[n]$ and $\hat{Y}[n]$ complicates the analysis. To simplify the problem, we adopt a common assumption in the signal processing literature (e.g. \cite{5672380, goyal2008compressive}) :
\begin{enumerate}
\item[(A)] The process $\eta[\cdot]$ is zero mean, white (uncorrelated entries), and is uncorrelated with $Y[\cdot]$.
\end{enumerate}
There exists a vast literature on conditions under which assumption (A) provides a good approximation to the system behavior. For example, in \cite{930935} it was shown that two consecutive samples $\eta[n]$ and $\eta[n+1]$ are approximately uncorrelated if the distribution of $Y[\cdot]$ is smooth enough, where this holds even if the sizes of the quantization regions are on the order of the variance of $Y[\cdot]$ \cite{1086334}. This property justifies the assumption that the process $\eta[\cdot]$ is white. Bennett \cite{BLTJ1340} showed that $\eta[\cdot]$ and $Y[\cdot]$ are approximately uncorrelated provided the PSD of $Y[\cdot]$ is smooth, the quantization regions are uniform and the quantizer resolution $q$ is high. Since in our setting the quantizer resolution may also be relatively low when $f_s$ approaches $R$, our analysis under (A) does not lead to an exact description of the performance limit under scalar quantization. Nevertheless, under (A) the distortion due to quantization decreases exponentially as a function of the quantizer bit precision and is proportional to the variance of the input signal. These two properties, which hold also under an exact analysis of the error due to scalar quantization with entropy coding \cite{gray1998quantization}, are the dominant factors in the MMSE analysis below. \\

\subsection{Distortion Analysis}
Under (A), the relation between the input and the output of the quantizer can be represented in the $z$ domain by
\begin{equation} \label{eq:z_domain}
\hat{Y}(z) = Y(z)+ \eta(z).
\end{equation}
This leads to the following relation between the corresponding PSDs:
\begin{align} \label{eq:pcm_in_out}
S_{\hat{Y}} \expphi & =  S_{Y} \expphi + S_\eta \expphi \nonumber \\
& = f_s \sum_{k\in \mathbb Z} S_X(f-f_sk) \left| H(f-f_sk)\right|^2  +  \sigma_{\eta}^2.
\end{align}
The block diagram of a generic system that realizes the input-output relation \eqref{eq:z_domain} is given in Fig.~\ref{fig:pcm}, where, in accordance with (A), $\eta[\cdot]$ is a white noise independent of $X(\cdot)$. In what follows, we derive an expression for the linear MMSE in estimating $X(\cdot)$ from $\hat{Y}[\cdot]$ according to the relation \eqref{eq:pcm_in_out} and an optimal choice of the pre-sampling filter $H(f)$, that minimizes this MSE. \\

\begin{figure}
\begin{tikzpicture}[node distance=2cm,auto,>=latex]
  \node at (0,0) (source) {$X(\cdot)$};
  \node[int1, right of = source] (anti) {$H(f)$};   
   \node [coordinate, right of = anti,node distance = 1.5cm] (smp_in) {};
  \node [coordinate, right of = smp_in,node distance = 0.7cm] (smp_out){};
	\node [coordinate,above of = smp_out,node distance = 0.4cm] (tip) {};
\fill  (smp_out) circle [radius=2pt];

\fill  (smp_in) circle [radius=2pt];
\fill  (tip) circle [radius=2pt];
\node[left,left of = tip, node distance = 0.6 cm] (ltop) {$f_s$};

\draw[->,dashed,line width = 1pt] (ltop) to [out=0,in=90] (smp_out.north);

\node[sum,right of=smp_out, node distance = 1.5cm] (plus_out) {$+$};
  
 \draw[->,line width=1pt,align=left] (smp_out) -- node[below,xshift= - 5pt] {$Y[n]$} (plus_out);
 
\node [right of = plus_out, node distance =0.7 cm] (quant_out) {};
\node [right, right of = quant_out, node distance=1cm] (sys_out) {$\hat{Y}[n]$};

\node [above of=plus_out, node distance=1.3cm] (eta) {$\eta[n]$};

\draw[->,line width=1pt] (plus_out) -- (sys_out);
\draw[->,line width=1pt] (eta)--(plus_out);
 %\draw[-,dashed] (plus_out)+(0.5,0) -| +(0.5,1.7) -| node[above] {\small ~~~~~~~~~$q$-bit quanitzer} +(-0.7,1.7) -| +(-0.7,-0.5) -| +(0.5,-0.5) -| +(0.5,0);
\draw[->,line width=1pt] (source) -- (anti);
\draw[->,line width=1pt] (anti) -- (smp_in);
\draw[line width=1pt]  (smp_in) -- (tip);
\end{tikzpicture}
\caption{\label{fig:pcm} Sampling and quantization system model.}
\end{figure}
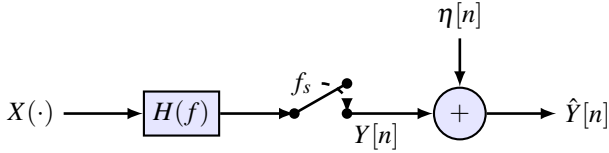 

The goal of the linear decoder is to provide a reconstruction signal $\hat{X}(\cdot)$ that minimizes 
\begin{equation} \label{eq:PAM_mmse_def}
\lim_{T\rightarrow \infty} \frac{1}{2T} \int_{-T}^T \mathbb E \left(X(t)-\hat{X}(t) \right)^2
\end{equation}
over all possible reconstruction signals of the form
\begin{equation}
    \label{eq:PAM_linear_constraint}
\hat{X}(t) = \sum_{n\in \mathbb Z} w(t,n)\hat{Y}[n],
\end{equation}
where $w(t,n)$ is square summable in $n$ for every $t \in \mathbb R$. Note that this decoder is non-causal in the sense that the estimate of the source sample $X(t)$ is obtained from the entire history of the quantized signal $\hat{Y}[\cdot]$. Since all signals in Fig.~\ref{fig:pcm} are assumed stationary, an expression for the minimal value of \eqref{eq:PAM_mmse_def} subject to the constraint \eqref{eq:PAM_linear_constraint} can be found using standard linear estimation techniques, leading to the following proposition:
\begin{prop}\label{prop:mmse_pre_filtering}
Consider the system in Fig.~\ref{fig:pcm}. The minimal time-averaged MSE \eqref{eq:PAM_mmse_def} in linear estimation of $X(\cdot)$ from $\hat{Y}[\cdot]$ is given by
\begin{align}
& D_{\PCM} \triangleq  \label{eq:mmse_pcm_general}   \\
& =  \sigma_{X}^{2}-\frac{1}{f_s} \int_{-\frac{f_s}{2}}^{\frac{f_s}{2}} \frac{
  \sum_{k\in\mathbb{Z}} S_X^2 \left(f-f_sk\right) \left|H\left(f-f_sk\right) \right|^2} 
  { \sum_{k\in\mathbb{Z}} S_X \left(f-f_sk\right) \left|H\left(f-f_sk\right) \right|^2 +\sigma_\eta^2/f_s  } df.   \nonumber
\end{align}
\end{prop}
\begin{proof}
See Appendix.
\end{proof}
The effect of the quantization noise is expressed in \eqref{eq:mmse_pcm_general} by an additive noise with a constant PSD over the digital domain. 
%An expression for the optimal estimator $ w^\star(t,n)$ can be shown to be of the form 
%\[
% w^\star (t,n) =  w(t-n/f_s).
%\]
%where the Fourier transform of $w(t)$ is 
%\begin{equation} \label{eq:estimator_frequency_response}
%W(f) = \frac{ H^*(f) S_X(f) }{\sum_{k\in \mathbb Z} \left| H(f) \right|^2 S_X(f-f_sk)+ \sigma_\eta^2/f_s }.
%\end{equation}
%The details are given in Appendix \ref{app:pcm_proof}. \\

Using H\"{o}lder's inequality and monotonicity of the function $x\rightarrow \frac{x}{x+1}$, the integrand in \eqref{eq:mmse_pcm_general} can be bounded for each $f$ in the integration interval $\left(-f_s/2, f_s/2 \right)$ by
\begin{align}
  \frac{ \left(S^\star(f)\right)^2 } 
  { S^\star(f)  +\sigma_\eta^2/f_s  }, \label{eq:optimal_j}  
\end{align}
where
\begin{equation} \label{eq:s_star}
S^\star(f) = \sup_{k\in\mathbb{Z}} S_X \left(f-f_sk\right) \left|H\left(f-f_sk\right) \right|^2.
\end{equation}
Since $S_X(f)$ is an $\mathrm L_1$ function, the supremum in \eqref{eq:s_star} is finite for all $f \in (-f_s/2,f_s/2)$ except for perhaps a set of Lebesgue measure zero. It follows that a lower bound on $D_{\PCM}$ is obtained by replacing the integrand in \eqref{eq:mmse_pcm_general}
with $S^\star(f)$. \par
Under the assumption that $S_X(f)$ is unimodal in the sense that it is symmetric and non-increasing for $f>0$, for each $f \in [-f_s/2,f_s/2]$ the supremum in \eqref{eq:s_star} is obtained for $k=0$. This implies that \eqref{eq:optimal_j} is achievable if the pre-sampling filter is a low-pass filter with cut-off frequency $f_s/2$, namely 
\begin{equation} \label{eq:pre_sampling_filter}
H(f) = \begin{cases}
 1, &  |f|\leq f_s/2, \\
 0, & \textrm{otherwise}. \end{cases}
\end{equation}
This choice of $H(f)$ in \eqref{eq:mmse_pcm_general} leads to
\begin{equation} \label{eq:mmse_optimal}
D_{\PCM} = \mmse^\star(f_s) + \int_{-\frac{f_s}{2}}^{\frac{f_s}{2}} \frac{
   S_X(f)  }
{1+\SNR(f) } df,
\end{equation}
where $\mmse^\star(f_s)$ is given by \eqref{eq:mmse_bound_def} and
\begin{equation} \label{eq:snr_super_Nyquist_analog}
\SNR(f) \triangleq
 f_s S_X(f)/ \sigma_\eta^2,\quad  -\frac{f_s}{2}\leq f \leq \frac{f_s}{2}.
\end{equation}
Henceforth, we will consider only processes with unimodal PSD, so that the MMSE under optimal pre-sampling filtering is given by \eqref{eq:mmse_optimal}. 

\subsection{PCM Distortion under a Fixed Bitrate}
From \eqref{eq:snr_super_Nyquist_analog} we see that when the variance of the quantization noise is independent of $f_s$, than the SNR in the system in Fig.~\ref{fig:pcm} increases linearly in $f_s$. The MMSE of $X(\cdot)$ given $\hat{Y}[\cdot]$ then decreases by a factor of $1/f_s$ when $f_s$ is large. However, when the bitrate $R = q f_s$ is fixed, the relation between $\sigma_\eta^2$ and $f_s$ is given by
\begin{equation} \label{eq:noise_var_freq}
\sigma_{\eta}^2 = c_0 2^{-2q}= c_0 2^{-2R/f_s}.
\end{equation}
Substituting \eqref{eq:noise_var_freq} into \eqref{eq:mmse_optimal} and \eqref{eq:noise_var_freq} we obtain:
\begin{prop} \label{prop:main}
The MMSE in estimating $X(\cdot)$ from $\hat{Y}[\cdot]$ assuming (A) and $R=qf_s$ satisfies
\begin{align}
D_{\PCM}(f_s,R) = \mmse^\star(f_s)+ \int_{-\frac{f_s}{2}}^{\frac{f_s}{2}} \frac{S_{X} (f) }{1+\SNR(f)}df \label{eq:main_result}
\end{align} 
where
\begin{equation} \label{eq:main_result_snr}
\SNR(f)=\SNR_{f_s,R}(f) = f_s \frac{2^{2R/f_s}}{c_0} S_{X}(f)
\end{equation} 
and $\mmse^\star(f_s)$ is given by \eqref{eq:mmse_bound_def}. \\
\end{prop}

We denote the two terms in the RHS of \eqref{eq:main_result} as the \emph{sampling distortion} and the \emph{quantization distortion}, respectively. Note that when $R\rightarrow \infty$ the quantization error vanishes and the distortion in PCM is only due to sampling. Since we assumed unimodal PSD, the sampling distortion vanishes only for $f_s \geq f_{\Nyq}$. \par
Figure~\ref{fig:various_PSD} shows $D_{\PCM}(f_s,R)$ as a function of $f_s$ for a given $R$ and various PSDs compared to their corresponding optimal ADX distortions $D^\star(f_s,R)$ of \eqref{eq:bound_def}. In Fig.~\ref{fig:various_PSD} and in other figures throughout, we take $c_0$ as in  \eqref{eq:quantizer_var} which corresponds to an optimal point density of the Gaussian distribution whose variance is proportional to the signal at the input to the quantizer. The variance of the latter is given by 
\[
\sigma_{\mathsf{in}}^2 = \int_{-\infty}^\infty S_X(f) \left|H(f)\right|^2df = \sigma_X^2-\mmse^\star(f_s).
\]
While $\sigma_{\mathsf{in}}^2$ depends on the sampling rate $f_s$, it can be shown to have a negligible effect on $D_{\PCM}(f_s,R)$ for sampling rates close to $f_{\Nyq}$, which is our main area of interest. We therefore ignore this dependency and continue our discussion assuming $\sigma_{\mathsf{in}}^2 = \sigma_X^2$.

\begin{figure}
\begin{center}

\begin{tikzpicture}
\draw [line width=1pt, fill=red!30] (-2.3,0) rectangle  (-2,0.3) node[right, yshift = -0.2cm] {\small $\mmse^\star(f_s)$};
\draw [fill=blue!30, line width=1pt] (2,0) rectangle  (2.3,0.3) node[right, yshift = -0.2cm, align = center] {\small quantization noise};
\end{tikzpicture}

\begin{tikzpicture}[scale=0.96]
\node at (0,0) (origin) {};

\fill[fill=blue!30] (-1.5,0) --(-1.5,0.3) --(1.5,0.3)--(1.5,0) -- cycle;

\draw[line width=0.5pt] (-1.5,0) node[below,xshift = 0.1cm] {\scriptsize $-\frac{1}{2}$}-- node[above,yshift=0.1cm,xshift=2.3cm] {\scriptsize $S_{\eta}(e^{2\pi i \phi})$}(-1.5,0.3)-- (1.5,0.3)--(1.5,0) node[below]{\scriptsize $\frac{1}{2}$}--cycle;

\draw[dashed] (-1.5,0) --(-1.5,2.7) node[above] {\scriptsize $H(\phi f_s)$} --(1.5,2.7)--(1.5,0) -- cycle;

\draw[line width=1pt] plot[domain=-1.8:1.8, samples=100]  (\x, {-(\x)*(\x)*2.5/1.8/1.8+2.5 }) ;

\draw [line width=1pt, fill=red!30]  plot[domain=-1.8:-1.5, samples=100]  (\x, {-(\x)*(\x)*2.5/1.8/1.8+2.5 }) -- (-1.5,0) ;

\draw [line width=1pt, fill=red!30] (1.5,0) -- plot[domain=1.5:1.8, samples=100]  (\x, {-(\x)*(\x)*2.5/1.8/1.8+2.5 }) ;

\draw[line width=1pt] (-1.8,-0.1) node[below,xshift=-0.1cm] {\scriptsize $-\frac{W}{f_s}$}-- (-1.8,0.05);
\draw[line width=1pt] (1.8,-0.1) node[below,xshift=0cm] {\scriptsize $\frac{W}{f_s}$}-- (1.8,0.05);

\node[above,yshift=0.45cm,xshift=0.2cm, rotate=50] at (-0.9,1.5) {\scriptsize $S_{Y}(e^{2\pi i \phi})$};
%\node[right,yshift=0.25cm] at (0,1) {\scriptsize $\sigma_\eta^2$};
%\draw[line width=1pt] (-0.1,0.9) -- (0.1,1.1);

\draw[->,line width=1pt] (origin)+(0,-0.2) -- +(0,2.9) node[above] {};
\draw[->,line width=1pt] (origin)+(-2,0) node[left, xshift=0cm] {\scriptsize $\phi$} -- +(2,0) ;
\node[below of = origin, node distance = 1cm] {(a) $f_s <  f_\Nyq$};
\end{tikzpicture}
\begin{tikzpicture}[scale=0.96]
\node at (0,0) (origin) {};
\fill[fill=blue!30] (-1.5,0) --(-1.5,1) --(1.5,1)--(1.5,0) -- cycle;
\draw[line width=0.5pt] (-2,0) node[below,xshift = 0cm] {\scriptsize $-\frac{1}{2}$}-- node[above,yshift=0.45cm,xshift=0.2cm] {\scriptsize $S_{\eta}(e^{2\pi i \phi})$}(-2,1)-- (2,1)--(2,0) node[below]{\scriptsize $\frac{1}{2}$}--cycle;

\draw[dashed] (-1.5,0) --(-1.5,2.7) node[above] {\scriptsize $H(\phi f_s)$} --(1.5,2.7)--(1.5,0) -- cycle;

\draw[line width=1pt] plot[domain=-1.5:1.5, samples=100]  (\x, {-(\x)*(\x)*2.5/2.25+2.5 }) ;

\draw[line width=1pt] (-1.5,-0.1) node[below,xshift=0.cm] {\scriptsize $-\frac{W}{f_s}$}-- (-1.5,0.05);
\draw[line width=1pt] (1.5,-0.1) node[below,xshift=0cm] {\scriptsize $\frac{W}{f_s}$}-- (1.5,0.05);

\node[above,yshift=0.45cm,xshift=0.2cm, rotate=50] at (-0.8,1.5) {\scriptsize $S_{Y}(e^{2\pi i \phi})$};
%\node[right,yshift=0.25cm] at (0,1) {\scriptsize $\sigma_\eta^2$};
%\draw[line width=1pt] (-0.1,0.9) -- (0.1,1.1);

\draw[->,line width=1pt] (origin)+(0,-0.2) -- +(0,2.9) node[above] {};
\draw[->,line width=1pt] (origin)+(-2,0) -- +(2,0) ;
\node[below of = origin, node distance = 1cm] {$f_s > f_\Nyq$};
\end{tikzpicture}

\caption{ \label{fig:PCM_spectral} Spectral representation of the distortion in PCM \eqref{eq:main_result}: (a) sampling below the Nyquist rate introduces sampling distortion $\mmse^\star(f_s)$. (b) As $f_s$ increases, $\mmse^\star(f_s)$ decreases and vanishes when $f_s \geq f_\Nyq$, but the contribution of the in-band quantization noise  increases due to lower bit-precision of each sample.
.}
\end{center}
\end{figure}
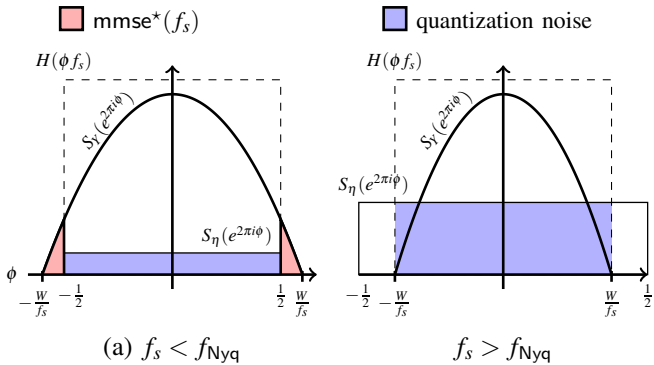

\subsection{The Optimal Sampling Rate \label{subsec:optimal_frequency_pcm}}

The quantization error in \eqref{eq:main_result} is an increasing function of $f_s$ (mainly due to the decrease in the exponent, but also due to the increase in $\sigma_{\mathsf{in}}^2$), whereas the sampling error $\mmse^\star_{X}(f_s)$ decreases in $f_s$. This situation is illustrated in Fig.~\ref{fig:PCM_spectral}. The sampling rate $f_s^\star$ that minimizes $D_{\PCM}(f_s,R)$ is obtained at an equilibrium point where the derivatives of both terms are of equal magnitudes. Figure~\ref{fig:various_PSD} shows that $f_s^\star$ depends on the particular shape of the input signal's PSD. If the signal is bandlimited, then we have the following result.
\begin{cor}\label{cor:optimal_sampling_rate_PCM}
For a bandlimited $X(\cdot)$, $f_s^\star$ that minimizes $D_{\PCM}(f_s,R)$ is at or below the Nyquist rate.
\end{cor}
\begin{proof}
Since $\SNR_{f_s,R}(f)$ is an increasing function of $f_s$ in the interval $0\leq f_s \leq R$, and since $\mmse^\star(f_s)=0$ for $f_s\geq f_\Nyq$, for all $f_s > f_\Nyq$ we have that $D_{\PCM}(f_\Nyq,R) \leq D_{\PCM}(f_s,R)$. Therefore, the minimizing sampling rate cannot be greater than $f_\Nyq$.
\end{proof}

How far $f_s^\star$ is below $f_\Nyq$ is determined by the derivative of $\mmse^\star(f_s)$, which equals $-2S_X(f_s/2)$. For example, in the case of $S_\Pi(f)$ of Examples~\ref{ex:rect1} and \ref{ex:rect2}, the derivative of $-2S_X(f_s/2)$ for $f_s<f_\Nyq = 2W$ is $-\sigma_X^2$. The derivative of the second term in \eqref{eq:main_result} is smaller than $\sigma_X^2$ for most choices of system parameters\footnote{This holds whenever $ \left(2^{0.5R/W}-1\right)^2>\frac{c_0}{\sigma_X^2}$.}. It follows that $0$ is in the sub-gradient of $D_{\PCM}(f_s,R)$ at $f_s=2W$, and thus $f_s^\star = 2W$, i.e., Nyquist rate sampling is optimal when the energy of the signal is uniformly distributed over its bandwidth. We now consider the other PSDs illustrated in Fig.~\ref{fig:various_PSD}.
\begin{example}[triangular PSD] \label{ex:triangle3}
Let $S_{\triangle}(f)$ be the PSD of Examples~\ref{ex:triangle1} and \ref{ex:triangle2}. For any $f_s\leq f_\Nyq = 2W$, we have
\[
\mmse^\star(f_s) = \sigma_X^2 \left(1 - \frac{f_s}{2W} \right)^2.
\]
Since the derivative of $\mmse^\star(f_s)$, which is $-2S_{\triangle}(f_s/2)$, changes continuously from $0$ to $-2\sigma_X^2/W$ as $f_s$ varies from $2W$ to $0$, we have $0<f_s^\star<2W$. The exact value of $f_s^\star$ depends on $R$ and the ratio $\sigma_X^2/c_0$. It converges to $2W$ as the value of any of these two increases. 
 \end{example}

\begin{figure}
\begin{center}
\begin{tikzpicture} 
\node at (0,0) {\includegraphics[trim=0cm 0cm 0cm 0cm,  ,clip=true,scale=0.45]{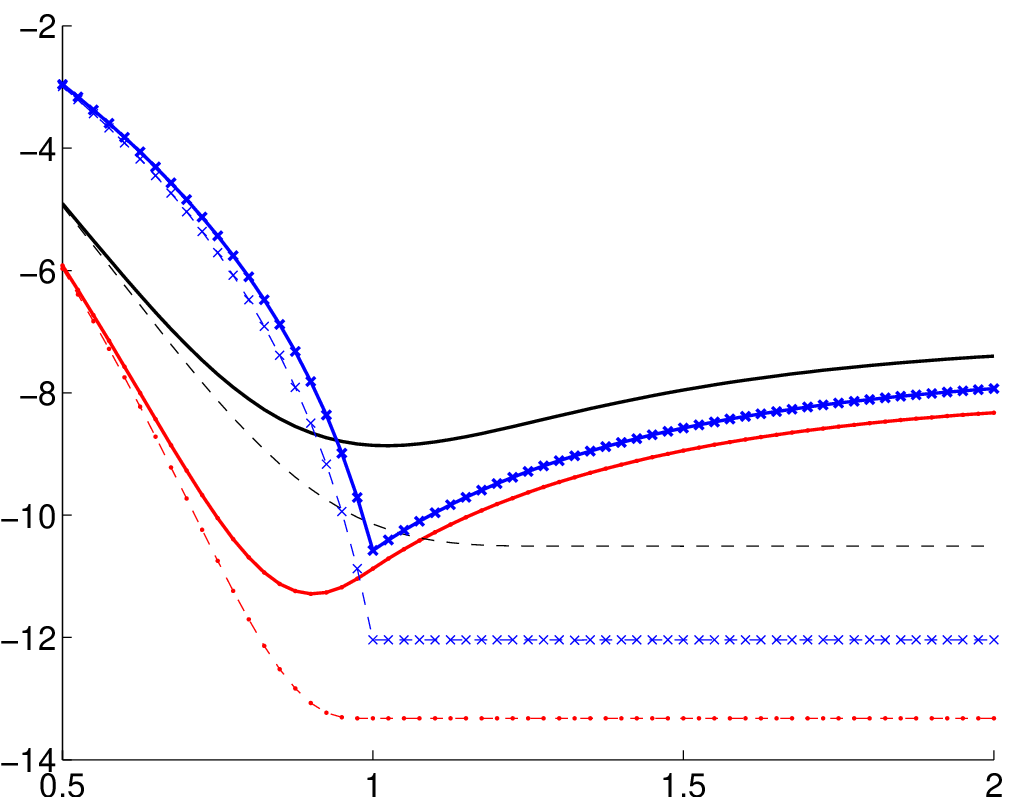}};

\node[rotate=90] at (-4.1,0) {\small {$D_{\PCM}(f_s,R)$~[dB]}} ;
\draw[->,line width = 1pt] (-3.4,-2.75) -- (4,-2.75) node[right] {\small $\frac{f_s}{2W}$}; 
%\draw[-,line width = 0.5pt] (-3,-2.75) -- node[above,yshift=-0.1cm,xshift=-0.25cm] {\scriptsize sub-Nyquist regime} (-1.02,-2.75) ; 
%\draw[line width = 0.5pt,dashed] (-1.04,-2.75) --  (-1.04,2) ; 

\node[rotate=-40] at (-3,2.4) {\small { \color{blue} $S_{\Pi}$ }} ;
\node[rotate=-65] at (-3,-0.4) {\small {\color{red} $S_{\triangle}$ }} ;
\node[rotate=-45] at (-3,1.4) {\small { $S_\Omega$ }};
\node at (-0.5,1.5) 
 {\includegraphics[width=1.8cm,height=1.6cm]{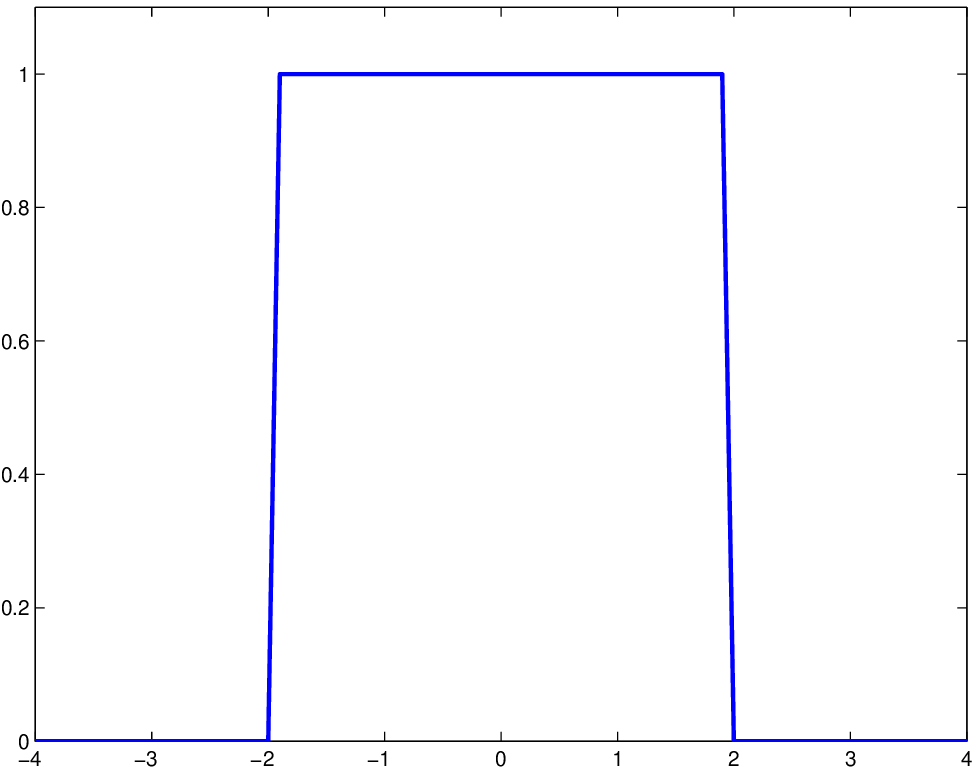} };
\node at (1.5,1.5) {\includegraphics[scale=0.1]{triangle_psd} };
\node at (3.5,1.5) {\includegraphics[width=1.8cm,height=1.6cm]{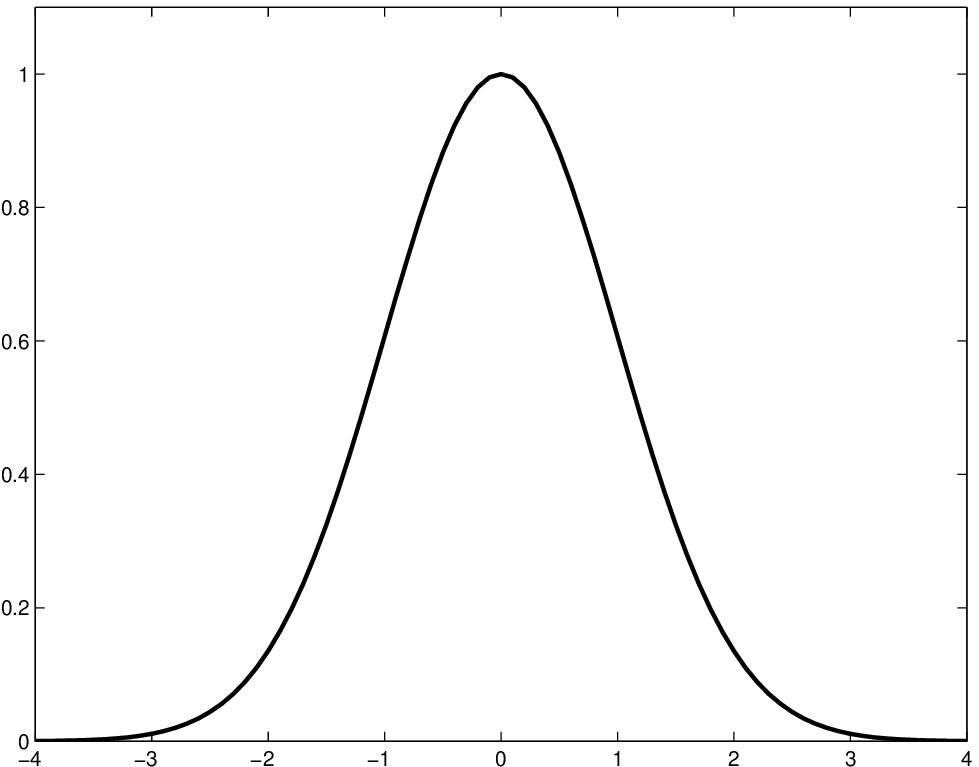} };
\node[rotate=0] at (3.5,1) {\scriptsize { $S_\Omega(f)$ }};
%\contourlength{1.5pt}
\node at (3.82,1.5) {\contour{white}{\scriptsize { $f_0=W$ }}};
\node[rotate=0] at (1.5,1) {\scriptsize { $S_{\triangle}(f)$ }};
\node[rotate=0] at (-.5,1) {\scriptsize { $S_{\Pi}(f)$ }};

\node at (-1.46,-1.5) {\color{red} \large $\star$};
\node at (-1.05,-1.2) {\color{blue} \large $\star$};
\node at (-0.9,-0.36) { \large $\star$};

\node at (-1.25,-2.45) {\color{red}  $\diamond$};
\node at (-1.05,-1.85) {\color{blue}  $\diamond$};
\node at (-0.23,-1.16) { \large $\diamond$};

\end{tikzpicture}
\caption{\label{fig:various_PSD} PCM Distortion $D_{\PCM}(f_s,R)$ as a function of $f_s$ for a fixed $R$ and various PSDs, which are given in the small frames. The dashed curves are the corresponding minimal ADX distortions $D(f_s,R)$. The symbols $\star$ and $\diamond$ indicate the distortion at rates $f_s^\star$ and $f_R$, respectively.}
\end{center}
\end{figure}

\begin{example}[PSD of unbounded support] \label{ex:gaussmarkov2}
Consider the PSD $S_\Omega(f)$ of the Gauss-Markov process  $X_\Omega(\cdot)$ in Example~\ref{ex:GaussMarkov1}. Since $X_\Omega(f)$ is not bandlimited, Corollary~\ref{cor:optimal_sampling_rate_PCM} does not hold. Nevertheless, as can be seen in Fig.~\ref{fig:various_PSD}, there exists an optimal sampling rate $f_s^\star$ that balances the two trends as explained in Subsection~\ref{subsec:optimal_frequency_pcm}.
\end{example}

\subsection{Discussion}
Under a fixed bitrate constraint, oversampling no longer reduces the MMSE since increasing the sampling rate forces a reduction in the quantizer resolution and increases the magnitude of the quantization noise. As illustrated in Fig.~\ref{fig:PCM_spectral}, for any $f_s$ below the Nyquist rate the bandwidths of both the signal and the quantization noise occupy the entire digital frequency domain, whereas the magnitude of the noise decreases as more bits are used in quantizing each sample. \par 
It follows that $f_s^\star$ cannot be larger than the Nyquist rate (Corollary~\ref{cor:optimal_sampling_rate_PCM}), and is strictly smaller than Nyquist when the energy of $X(\cdot)$ is not uniformly distributed over its bandwidth. In this case, some distortion due to sampling is preferred in order to increase the quantizer resolution. In other words, restricted to scalar quantization, the optimal rate $R$ code is achieved by sub-Nyquist sampling. This behavior of $D_{PCM}(f_s,R)$ is similar to the behavior of the minimal ADX distortion $D^\star(f_s,R)$, as both provide an optimal sampling rate which balances sampling distortion and lossy compression distortion. On the other hand, oversampling introduces redundancy into the PCM representation, and yields a worse distortion-rate code than with $f_s = f_s^\star$. In this aspect the behavior of $D_{PCM}(f_s,R)$ is different than $D^\star(f_s,R)$ that represents the information theoretic bound, since the latter does not penalize oversampling as the optimal ADX encoder has the freedom to discard redundant samples when needed.\par
The similarity between $f_s^\star$ and $f_{R}$ as a function of $R$ is due to the fact that the optimal representation is obtained by discarding the same part of the signal under both the optimal lossy compression scheme or PCM. The observation that $f_s^\star \leq f_{R}$ in Examples \ref{ex:triangle2} and \ref{ex:gaussmarkov2} is explained by the diminishing effect of reducing the sampling rate on the overall error. That is, since $D_{\PCM}(f_s^\star,R) \geq D_X(R)$, the optimal lossy compression scheme is more sensitive to changes in the sampling rate than the sub-optimal implementation of A/D conversion via PCM.

\section{Conclusions\label{sec:conclusion}}
We considered an analog-to-digital compression (ADX) setting in which an analog source is described by its rate-limited samples, obtained using any bounded linear sampling technique.
We have shown that for any given bitrate $R$, there exists a critical sampling rate denoted $f_{R}$, such that the minimal distortion subject only to the bitrate constraint can be achieved by sampling at or above $f_{R}$. This minimal distortion is the indirect DRF of the signal given its noisy version, or the standard DRF when this noise is zero.
The critical sampling rate $f_R$ is strictly smaller than the Nyquist or Landau rates for processes whose power is not uniformly distributed over their spectral band. As the bitrate $R$ increases, $f_R$ increases as well and converges to the Nyquist or Landau rates as $R$ goes to infinity. \par
The results in this paper imply that with an optimized multi-branch LTI uniform sampler, sampling below the Nyquist rate and above $f_R$ does not degrade performance in the case where lossy compression of the samples is introduced. Since lossy compression due to quantization is an inherent part of any analog to digital conversion scheme, our work suggests that sampling below the Nyquist rate is optimal in terms of minimizing distortion in practice for most systems. \par
We also considered the case of a more restricted encoder and decoder which corresponds to pulse-code modulation (PCM) sampling and quantization. That is, instead of a vector quantizer whose block-length goes to infinity, PCM uses a zero-memory zero-delay quantizer. Under a fixed bitrate at the output of this quantizer, there exists a trade-off between bit-precision and sampling rate. We examined the behavior of this trade-off under an approximation on the scalar quantizer using additive white noise. We have shown through various examples that the optimal sampling rate in PCM experiences a similar behavior as the critical rate $f_{R}$, which is the minimal sampling rate under optimal source encoding-decoding of the samples. \\

There are a few important future research directions that arise from this work. While we restricted ourselves to bounded linear samplers, it is important to understand whether the distortion at a given sampling rate can be improved by considering non-linear sampling functions. Indeed, such improvement is seen in the setting of \cite{wu2011optimal}, where a finite dimensional sampling system with a Gaussian input is considered. In addition, reduction of the optimal sampling rate under the bitrate constraint from the Nyquist rate to $f_{R}$ can be understood as the result of a reduction in degrees of freedom in the compressed signal representation compared to the original source. It is suggested that a similar principle may hold under non-Gaussian signal models (e.g., sparse signals), so that the sampling rate under a bitrate restriction can be reduced without incurring additional distortion. Finally, under sub-optimal encoding such as in PCM, it is important to characterize the conditions on the encoder under which oversampling has a detrimental effect on the distortion. 

%%%%%%%%%%%%%%%%%%%%%%%%%%%%%%%%%%%%%%%%%%%%%%%%%%%%%%%%%%%%%%%%%%%%%%%%%%%%%%%%

%%%%%%%%%%%%%%%%%%%%%%%%%%%%%%%%%%%%%%%%%%%%%%%%%%%%%%%%%%%%%%%%%%%%%%%%%%%%%%%%
\appendix

\section{Proofs \label{app:proofs}}

\subsection{Proof of Theorem~\ref{thm:achivability}}
%The proof of Theorem~\ref{thm:achivability} can be obtained using a similar procedure outlined in \cite{Kipnis2014} that leads to Theorem 21 there. However, since the combined sampling and source coding setting that of \cite{Kipnis2014} differs than ours, we explain this 

For the MB-LTI sampler, the set $Y_{\infty}$ of \eqref{eq:asymp_samples} is invariant under time shifts by an integer multiple of $1/f_s$ of the input $X_\epsilon(\cdot)$. Hence, for any $n \in \mathbb Z$, the distribution of $X(t)$ and $X(t+n/f_s)$ conditioned on the sigma algebra generated by $Y_{\infty}$ are identical. As a result, the process $\widetilde{X}_T(\cdot)$ of \eqref{eq:mmse_estimator_finite_T} has an asymptotic distribution as $T\rightarrow \infty$ that is cyclostationary with period $1/f_s$ \cite{bennett1958statistics} (also known as $1/f_s$-ergodic \cite{Nedoma}). Denote by 
\[
\widetilde{X}(t) = \mathbb E \left[ X(t) | Y_{\infty} \right],\quad t\in \mathbb R,
\]
the process obtained as the asymptotic distribution law of $\widetilde{X}_T(\cdot)$ when $T\rightarrow \infty$. It follows from \eqref{eq:decomp} that with $S$ a MB-LTI sampler, the asymptotic ADX distortion is given by
\[
D(S,R) = \mmse(S) + D_{\widetilde{X}}(R),
\]
where 
\begin{align*}
\mmse(S) & = \lim_{T\rightarrow \infty} \frac{1}{T} \int_{-T/2}^{T/2} \left(X(t) - \widetilde{X}(t) \right)^2 dt,
\end{align*}
and 
$D_{\widetilde{X}}(R)$ is the DRF of the process $\widetilde{X}(\cdot)$. \par
We note that $\widetilde{X}(\cdot)$ can be derived in closed form using a procedure that extends the Wiener filter \cite{815501, 4663942, ShannonMeetsNyquist}. 
%For example, when $L=1$, we have
%\[ \widetilde{X}(t) = \sum_{n\in\mathbb Z} w(n/f_s - t)X_\epsilon(n/f_s), \]
% where $n$
Since cyclostationary processes are in particular asymptotic mean stationary processes \cite{gray2009probability}, it follows from \cite{ gray2011entropy} that the DRF of $\widetilde{X}(\cdot)$ equals its information DRF, i.e., the infimum over conditional probability distributions with mutual information rate not exceeding $R$. A closed form expression for this information DRF was derived in \cite{Kipnis2014} in terms of the pre-sampling filters $H_1(f),\ldots, H_L(f)$ and the PSDs $S_X(f)$ and $S_\epsilon(f)$. Under the special case where the supports of $H_1(f),\ldots,H_L(f)$ are disjoint, this expression from \cite{Kipnis2014} implies that
\begin{subequations}
\label{eq:DRF_multi_disjoint}
\begin{align}
\label{eq:DRF_multi_D_disjoint}
 D(S,R) & = \mmse\left(S\right) + \sum_{l=1}^L \int_{-\frac{f_s}{2}}^\frac{f_s}{2} \min\left\{\widetilde{S}_l(f) ,\theta \right\} df \\
R_\theta & = \frac{1}{2} \sum_{l=1}^L \int_{-\frac{f_s}{2}}^\frac{f_s}{2} \log^+ \left[\widetilde{S}_l(f)/\theta \right]df, \label{eq:DRF_multi_R}
\end{align}
\end{subequations}
where 
\[
\widetilde{S}_l(f) \triangleq \frac{\sum_{n\in \mathbb Z} S_X^2(f-f_s n) \mathbf 1_{\supp H_l}(f-f_s n)  }{\sum_{n\in \mathbb Z} \left[S_{X_\epsilon}(f-f_sn) \right]},
\]
and 
\[
\mmse\left(S\right) = \sigma_X^2 - \sum_{l=1}^L \int_{-\frac{f_s}{2}}^\frac{f_s}{2} \widetilde{S}_l(f) df. 
\]

Let $F^\star_{f_s}$ be a set of Lebesgue measure at most $f_s$ that maximizes \eqref{eq:F_star_def}. We now show that $F^\star_{f_s}$ can be approximated by $L$ intervals of measure at most $f_s/L$. Let $\epsilon>0$. Consider the measure $\mu_{S_X}$ defined by 
\[
\mu_{S}(A) = \int_A \frac{S_X^2(f)}{S_{X_\epsilon}(f)}df
\]
for a Lebesgue measurable set $A$. \par
Since $S_X(f)$ is $\mathrm L_1(\mathbb R)$, we can choose a set $G \subset F^\star_{f_s}$ such that $S_X(f)$ is bounded on $G$ and such that $\mu_S(G) > \mu_S(F^\star)- \epsilon / 3$. The measure $\mu_{S}$ is absolutely continuous with respect to the Lebesgue measure and hence is a regular measure \cite{royden1988real}. Therefore, there exists $M$ intervals $I_1,\ldots,I_M$ such that $\cup_{i=1}^M I_i \subset G$ and $\mu_{S}(\cup_{i=1}^M I_i ) > \mu_{S}(G) - \epsilon/3 > \mu_S(F^\star_{f_s})- 2\epsilon/3$. We can assume that $I_1,\ldots,I_M$ are disjoint; otherwise we use $I'_1 = I_1$, $I'_2 = I_2 \setminus I'_1$, $I_3' = I_3 \setminus (I_1' \cup I_2')$, and so forth. Therefore, $\sum_{i=1}^M\mu(I_i) \leq f_s$. For $\delta>0$, let $L_i = \lfloor M \mu(I_i) / \delta \rfloor$ and $L = \sum_{i=1}^M L_i$. We now define $L$ pre-sampling filters as follows: for each $i=1,\ldots,M$, consider $L_i$ disjoint intervals $I_{i,1},\ldots,I_{i,j}$ of length $r=\delta/M$ that are sub-intervals of $I_i$. Since $\mu(I_i) \geq L_i r$, such $L_i$ intervals exist and we set $A_i = I_i \setminus \cup_{j=1}^{L_i}I_{i,j}$. That is, $A_i$ is the part of the interval $I_i$ that is not covered by these $L_i$ intervals. In particular, $\mu(A_i)\leq r$. Finally, set the support of each filter $H_{i,j}$ to be $I_{i,j}$. Note that 
\[
\mu(\sum_{i,j} \supp H_{i,j}) = \sum_{i=1}^M L_i r \leq r \sum_{i=1}^M M \mu(I_i) / \delta \leq f_s. 
\]
This way we have defined $L = L_1 + \ldots + L_M$ filters, each of passband of width $r \leq f_s / L$. It is left to show that
\[
\sum_{i,j} \mu_{S}(\supp H_{i,j}) =\sum_{i,j} \int_{\supp H_{i,j}} \frac{S_X^2(f)}{S_{X_\epsilon}(f)} df > \int_{F^\star_{f_s}} \frac{S_X^2(f)}{S_{X_\epsilon}(f)} df-\epsilon.
\]
Denote by $m_s$ the essential supremum of $S_X(f)$ on $G$ and note that $S_{X|X_\epsilon}(f)  \leq m_s$ on $G$ as well. We have 
\[
\mu_S(A_i) = \int_{A_i} S_{X|X_\epsilon}(f) df \leq m_s \mu(A_i) \leq m_s r.
\]
It follows that
\begin{align*}
\mu_S(\sum_{i,j} \supp H_{i,j}) & = \sum_{i=1}^M \sum_{j=1}^{L_i} \mu_S(I_{i,j}) = \sum_{i=1}^{M} \mu_S(I_i) -  \sum_{i=1}^{M} \mu_S(A_i) \\
& \geq \mu_S(G)-\epsilon/3 - M m_s r \geq \mu_S(F^\star_{f_s})-2\epsilon/3 - M m_s \delta. 
\end{align*}
Taking $\delta = \epsilon/ (3 M m_s)$ leads to the desired result. \\

To summarize, we constructed $L$ interval $F_1,\ldots,F_L$, each of measure at most $f_s/L$, such that
\[
\sum_{l=1}^L \int_{F_l} S_{X|X_\epsilon}(f) df +\delta > \int_{F^\star_{f_s}} S_{X|X_\epsilon}(f)  df. 
\]

We now use use the following Proposition, proof of which can be found in \cite[Prop. 3.4]{KipnisThesis}:
\begin{prop} \label{prop:min_max}
Fix $R>0$ and set $A\subset \mathbb R$. For an integrable function $f$ over $A$, define
\begin{align*}
D(f) & = - \int_A \left[f(x)-\theta\right]^+dx \\
R & = \frac{1}{2} \int_A \log^+ \left[f(x)/\theta \right]dx.
\end{align*}
Let $f$ and $g$ be two integrable functions such that
\[
\int_A f(x) dx \leq \int_A g(x) dx.
\]
Then $D(g)\leq D(f)$. 
\end{prop}

We use Proposition~\ref{prop:min_max} with $A = F^\star \cup F_1 \cup \ldots \cup F_L$,
\[
f(x) = \mathbf 1_{F^\star} (x)  S_{X|X_\epsilon}(x) ,
\]
and
\[
g(x) = \mathbf 1_{\cup_{l=1}^L F_l} (x) \left(S_{X|X_\epsilon}(f) + \delta \right).
\]
Note that $D^\star(f_s,R)$ is a water-filling expression of the form \eqref{eq:bound_def} over $f(x)$ and A. Denote by $D_\delta$ the function defined by a water-filling expression over $g(x)$. Since $g(x)\geq f(x)$, it follows from Proposition~\ref{prop:min_max} that 
\[
D_\delta \leq D^\star(f_s,R). 
\]
Since $D_\delta$ is continuous in $\delta$ and since $\lim_{\delta \rightarrow 0} D_\delta = D(S,R)$, for $\epsilon>0$ there exists $L$ and $\delta$ such that $D(S,R)+\epsilon > D_{\delta} \geq D^\star(f_s,R)$.

\subsection{ Proof of Theorem~\ref{thm:converse} }
%first consider the case of MB-LTI sampler. 
Consider the following cases of the sampler $S$ in ADX:
\begin{itemize}
    \item [(i)]  $S$ is a MB-LTI uniform sampler of sampling rate $f_s$.
    \item [(ii)] $S=(K_H,\Lambda)$ is a bounded linear sampler such that $\Lambda$ is periodic with uniform density $f_s$.
    \item [(iii)] $S=(K_H,\Lambda)$ is any bounded linear sampler such that $d^+(\Lambda) \leq f_s$.
\end{itemize}
We show that case (iii) follows from (ii) which follows from (i). \\

\subsubsection*{Case (i)} For $S$ a MB-LTI sampler, 
and given $S_X(f)$, $S_\epsilon(f)$, $f_s$ and $L$, the properties of the set of optimal pre-sampling filters $H_1,\ldots,H_L$ that minimizes $D(S,R)$ were given in \cite[Thm. 21]{Kipnis2014}. In particular, it follows from this characterization that the support of each $H_l$ is a bounded aliasing-free set for sampling rate $f_s/L$, in the sense that for $f_1,f_2 \in \supp~H_l$, $f_1 \neq f_s$ modulo the grid $\mathbb Z f_s/L$. Since we are interested in bounding $D(S,R)$ from below, we can assume without loss of generality that the support of $H_1,\ldots,H_L$ satisfies the aliasing free condition. With this assumption, a closed form expression for $D(S,R)$ follows from \cite[Thm. 21]{Kipnis2014}
\begin{subequations}
\label{eq:DRF_multi}
\begin{align}
\label{eq:DRF_multi_D}
 D(S,R) & = \sigma_X^2  - \sum_{l=1}^L \int_{\supp~H_l} \min\left[S_{X|X_\epsilon}(f)  - \theta \right] df \\
R_\theta & = \frac{1}{2} \sum_{l=1}^L \int_{\supp~H_l} \log^+ \left[S_{X|X_\epsilon}(f) /\theta \right]df. 
\end{align}
\end{subequations}
Furthermore, \cite[Prop. 2]{Kipnis2014} implies that the Lebesgue measure of $H_l^\star$ is at most $f_s/L$. Therefore, the Lebesgue measure of the union of $\supp~H_1, \ldots, \supp~H_L$ is at most $f_s$. Since Propositon~\ref{prop:min_max} implies that a water-filling expression of the form \eqref{eq:DRF_multi_D} is non-increasing in the function $S_{X|X_\epsilon}(f)$, it follows that \eqref{eq:DRF_multi_D} is bounded from below by
\begin{align*}
D^\star & = \sigma_X^2 - \int_{F^\star} \left[ S_{X|X_\epsilon}(f) - \theta\right]^+ df \\
R & = \frac{1}{2}  \int_{F^\star} \log^+ \left[S_{X|X_\epsilon}(f)/\theta \right] df,
\end{align*}
which, by definition, equals $D^\star(f_s,R)$. \\

\subsubsection*{Case (ii)} Assume that the sampling set $\Lambda$ is periodic with period $T_0$, i.e. it satisfies $\Lambda=\Lambda+T_0$. Assume moreover that 
 $K_H(t+T_0k,\tau)=K_H(t,\tau)$ for all $k\in \mathbb Z$, i.e. $K_H(t,\tau)$ is periodic in $t$ with period $T_0$.
\par
Denote by $L$ the number of points in $\Lambda$ in the interval $[-T_0/2,T_0/2]$. Therefore, $\lfloor L/T_0 \rfloor \leq d(\Lambda_T) \leq \lceil L/T_0 \rceil$, and hence the symmetric density of $\Lambda$ exists and equals $L/T_0$. Denote by $t_0,\ldots,t_{L-1}$ the $L$ members of $\Lambda_{T_0} = \Lambda \cap [-T_0/2,T_0/2]$, where without loss of generality we can assume that $T_0/2 \notin \Lambda$. Continue to enumerate the members of $\Lambda$ that are larger than $t_{L-1}$ in the positive direction $t_L,t_{L+1},\ldots$, and the elements of $\Lambda$ smaller than $t_0$ in the negative direction $t_{-1}, t_{-2},\ldots$. By the periodicity of $\Lambda$, $t_{l+Lk}=t_l+T_0k$ for all $l=0,\ldots,L-1$ and $k\in \mathbb Z$. For $n=l+kL$, and $t_n < T/2$, each sample $Y_n$ in the vector of samples $Y_T$ satisfies
\begin{align*}
Y_n & = \int_{-\infty}^{\infty} K_H(t_{l+Lk},s) X_\epsilon(s) ds = \int_{-\infty}^\infty K_H(t_l+T_0k,s)X_\epsilon(s) ds \\
& = \int_{-\infty}^\infty K_H (t_l,s)X_\epsilon(s) ds = h_l(s-t_l)X_\epsilon(s) df, 
\end{align*}
where, for $l=0,\ldots,L-1$, we denoted $h_l(s) \triangleq K_H(t_l,s+t_l)$. We define the vector valued process $\Yv[\cdot] = \{\Yv[k],\, k\in \mathbb Z\}$ by
\[
\Yv[k] = \left(Y_{Lk},Y_{Lk+1},\ldots, Y_{Lk+L-1} \right), \quad k \in \mathbb Z.
\]
That is, the $k$th sample of $\Yv[\cdot]$ is a vector in $\mathbb R^L$ consists of $L$ consecutive samples of $X_\epsilon(\cdot)$. Note $Y[\cdot]$ is independent of the time horizon $T$. Since each $h_l(s)$ defines an LTI system, it follows that sampling with the periodic set $\Lambda$ and the pre-processing system $K_H$ is equivalent to sampling using $L$ uniform sampling branches each of sampling rate $1/T_0$. From case (i) of the proof, it follows that $D(S,R) \geq D^\star(L/T_0,R) = D^\star(d(\Lambda),R)$. \\

\subsubsection*{Case (iii)}
We now consider the general case of $S=(\Lambda,K_H)$ an arbitrary bounded linear sampler. For a sequence $\left\{T_n,\, n=1,2,\ldots\right\}$ such that $\lim_{n\rightarrow \infty} T_n = \infty$, denote
\[
d_n \triangleq d_{T_n}(\Lambda) = \frac{\Lambda \cap [-T_n/2,T_n/2]}{T_n},
\]
and let $Y_{T_n}$ be the vector of $d_n$ samples obtained by sampling $X_\epsilon(\cdot)$ using $S$ over the interval $[-T_n/s,T_n/s]$. In addition, define the set $\tilde{\Lambda}_n$ to be the periodic extension of $\Lambda_{T_n}$, i.e, 
\[
\tilde{\Lambda}_n \triangleq \Lambda_{T_n} + T_n\mathbb Z. 
\]
Therefore, $\tilde{\Lambda}_n$ is a periodic sampling set with period $T_n$ and, consequently, symmetric density $d_n$. We also extend $K_H(t,s)$ periodically as 
\[
\tilde{K}_n\left(t,s\right) \triangleq 
K_H([t], \tau)
\]
where here and henceforth $[t]$ denotes \emph{$t$ modulo the grid $T_n\mathbb Z$} (i.e. $t = [t]+kT_n$ where $k \in \mathbb Z$ and $0\leq [t] < T_n$). Let $S_n \triangleq (\tilde{\Lambda}_n, \tilde{K}_n)$. We have
\begin{equation} \label{eq:proof_converse}
D_{T_n}(S,R) \overset{(a)}{=} D_{T_n}(\tilde{S}_n,R) \overset{(b)}{\geq} D(\tilde{S}_n,R) \overset{(c)}{\geq} D^\star(d_n,R),
\end{equation}
where: (a) follows from the definition of $S_n$, (b) follows since the distribution of the estimator of $X(\cdot)$ from the samples obtained by a MB-LTI sampler is cyclostationary, hence enlarging the time horizon $T$ can only reduce distortion \cite{gray2011entropy}, and (c) is obtained from part (ii) of the proof. \par

Since any unbounded sequence of time horizons $\{T_n\}$ satisfies 
\eqref{eq:proof_converse}, we conclude that
\[
 \liminf_{T\rightarrow \infty} D_T(S,R) = D(S,R). 
\]
Finally, since $D^\star(f_s,R)$ is continuous and non-increasing in $f_s$, we have
\[
\lim_{n\rightarrow \infty} D^\star(d_n,R) \geq D^\star(d^+(\Lambda),R). 
\]

\subsection{Proof of Proposition~\ref{prop:opitmal_sampling_rate} }
Let $(R,D)$ be a point on the curve $\left(R,D_{X|X_\epsilon}(R)\right)$. For $\theta$ such that 
\begin{align*}
R & = \frac{1}{2} \int_{-\infty}^\infty \log^+ \left[S_{X|X_\epsilon}(f)/\theta \right] df,
\end{align*}
denote $F_\theta \triangleq \left\{ f \in \mathbb R~:~  S_{X|X_\epsilon}(f) > \theta \right\}$, so that $f_R = \mu(F_\theta)$,
\begin{align*}
R & = \frac{1}{2} \int_{F_\theta} \log \frac{S_{X|X_\epsilon}(f)}{\theta} df, 
\end{align*}
and
\begin{align}
D & = \sigma_X^2 -  \int_{F_\theta} \left(S_{X|X_\epsilon}(f) - \theta \right) df.
\label{eq:cor_proof_D}
\end{align}
Let $F^\star \subset \mathbb R$ be such that 
\begin{equation}
    \label{eq:cor_proof_Dstar}
D^\star(f_R,R) =\sigma_X^2 -  \int_{F^\star} \left[ S_{X|X_\epsilon}(f) df - \theta\right]^+,
\end{equation}
and 
\[
R =\frac{1}{2} \int_{F^\star} \log^+[S_{X|X_\epsilon}(f)/\theta] df.
\]
From the definition of $D^\star(f_s,R)$, it follows that
\begin{equation}
\int_{F^\star} S_{X|X_\epsilon}(f) df \geq  \int_{F_\theta} S_{X|X_\epsilon}(f)df. \label{eq:cor_proof_ineq}
\end{equation}
Since the distortion expressions \eqref{eq:cor_proof_D} and \eqref{eq:cor_proof_Dstar} are non-increasing in $\int_{F^\star} S_{X|X_\epsilon}(f) df$ and $\int_{F_\theta} S_{X|X_\epsilon}(f) df$, respectively, it follows from \eqref{eq:cor_proof_ineq} that 
$D^\star(f_R,R) \leq D_{X|X_\epsilon}(R)$. In order to prove the reverse inequality, note that for any bounded linear sampler $S$ and $R$ we have
\[
D(S,R) \geq D_{X|X_\epsilon}(R). 
\]
However, it follows from Theorem~\ref{thm:achivability} that $D^\star(f_s,R)$ is achievable, and hence $D^\star(f_s,R) \geq D_{X|X_\epsilon}(R)$. Evidently, this same inequality can be derived directly from the definition of $D^\star(f_s,R)$ in \eqref{eq:bound_def} and the expression for $D_{X|X_\epsilon}(R)$ in \eqref{eq:dobrushin} (that is, without using Theorem~\ref{thm:achivability}). 

\subsection{Proof of Proposition~\ref{prop:mmse_pre_filtering}\label{app:pcm_proof}}
For $0\leq \Delta \leq 1$ define
\[
X_\Delta[n] \triangleq X\left( (n+\Delta)T_s \right),\quad n\in \mathbb Z,
\]
where $T_s \triangleq f_s^{-1}$. Also define $\hat{X}_\Delta[n]$ to be the optimal MSE estimator of $X_\Delta[n]$ from $\hat{Y}[\cdot]$, that is
\[
\hat{X}_\Delta[n] = \mathbb E \left[X_\Delta[n]|\hat{Y}[\cdot] \right],\quad n\in \mathbb Z.
\]

The MSE in \eqref{eq:PAM_mmse_def} can be written as
\begin{align}
&\mmse_{X|\hat{Y}} = \lim_{N \rightarrow \infty} \frac{1}{2N+1} \int_{-N}^{N+1} \mathbb E \left(X(t)-\hat{X}(t) \right)^2dt \nonumber \\
& = \lim_{N\rightarrow \infty} \frac{1}{2N+1} \sum_{n=-N}^N \int_0^1 \mathbb E\left( X\left((n+\Delta)T_s \right) - \hat{X}\left((n+\Delta)T_s \right)  \right)^2 d\Delta \nonumber \\
& = \lim_{N\rightarrow \infty} \frac{1}{2N+1} \sum_{n=-N}^N \int_0^1 \mathbb E\left(X_\Delta[n] - \hat{X}_\Delta[n]  \right)^2 d\Delta \nonumber \\
& = \int_0^1 \mathbb E\left(X_\Delta[n] - \hat{X}_\Delta[n]  \right)^2 d \Delta \label{eq:proof_main1}.
\end{align}
Note that $S_{X_\Delta} \expphi = S_{Y} \expphi$ and $X_\Delta[\cdot]$ and $\hat{Y}[\cdot]$ are jointly stationary with cross-PSD
\begin{align*}
S_{X_\Delta \hat{Y}} \expphi &= S_{X_\Delta} \expphi 
= f_s \sum_{k\in \mathbb Z} S_{X} \left( f_s (k-\phi) \right)e^{2\pi i \Delta (k-\phi)}.
\end{align*}
Denote by $S_{X_\Delta | \hat{Y}} \expphi$ the PSD of the estimator obtained by the discrete Wiener filter for estimating $X_\Delta[\cdot]$ from $\hat{Y}[\cdot]$. We have
\begin{align}
& S_{X_\Delta | \hat{Y} } \expphi = \frac{S_{X_\Delta \hat{Y}} \expphi S^*_{X_\Delta \hat{Y}} \expphi }{S_{\hat{Y}} \expphi}  \nonumber \\
& = \sum_{n,k} \frac{f_s^2 S_{X_a} \left(f_s(k-\phi) \right)  S_{X_a}^* \left(f_s(n-\phi) \right) e^{2\pi i \Delta (k-n)}  } 
{S_{Y} \expphi + S_\eta \expphi }, \label{eq:proof_main2}
\end{align}
where $S_{X_a}(f) = S_X(f) H^*(f)$ is the cross-PSD of $X(\cdot)$ and the signal at the output of the filter $H(f)$. The estimation error in Wiener filtering is given by 
\begin{align}
\mathbb E\left(X_\Delta[n] - \hat{X}_\Delta[n] \right)^2 & \nonumber \\
= \inthalftohalf S_{X_\Delta} & \expphi d\phi - \inthalftohalf S_{X_\Delta| \hat{Y}} \expphi d\phi \nonumber  \\
 = \sigma_X^2 - & \inthalftohalf S_{X_\Delta| \hat{Y}} \expphi d\phi \label{eq:proof_main3}. 
\end{align}
Equations \eqref{eq:proof_main1}, \eqref{eq:proof_main2} and \eqref{eq:proof_main3} lead to
\begin{align}
 D_{PCM}(f_s,q,H) & = \int_0^1 \mathbb E\left(X_\Delta[n] - \hat{X}_\Delta[n]  \right)^2 d\Delta \nonumber \\
& = \sigma_X^2- \inthalftohalf \int_0^1 S_{X_\Delta| \hat{Y}} \expphi d\phi \nonumber \\ 
& \overset{a}{=} \sigma_X^2 - \inthalftohalf \frac{f_s \sum_{k\in \mathbb Z} \left| S_{X_a} \right|^2 \left(f_s (k-\phi)\right)}
{S_{Y} \expphi + S_\eta \expphi } d\phi, \label{eq:proof_main4}
\end{align}
where $(a)$ follows from \eqref{eq:proof_main2} and the orthogonality of the functions $\left\{ e^{2\pi x k}, k\in \mathbb Z \right\}$ over $0\leq x \leq 1$. Equation \eqref{eq:mmse_optimal} is obtained from \eqref{eq:proof_main4} by changing the integration variable from $\phi$ to $f=\phi f_s$. \\

The optimal MMSE linear estimator of $X(t)$ from $\hat{Y}$ has the property that the estimation error is uncorrelated with any sample from $\hat{Y}[\cdot]$, namely,
\[
\mathbb E \left[ \left( X(t) - \sum_{n} w[n] \hat{Y}[n] \right) \hat{Y}[k]  \right] =0
\]
for all $k\in \mathbb Z$. This implies that
\begin{equation} \label{eq:proof_estimator}
\int_{-\infty}^\infty R_X(t-u-k/f_s) h(u) du = \sum_{n} w[n] R_{\hat{Y}}[n-k].
\end{equation}
Taking the discrete time Fourier transform of both sides with respect to $k$ in \eqref{eq:proof_estimator} leads to
\begin{align}
f_s \sum_{m} S_X\left( f_s(\phi-k) \right) e^{-2\pi i t f_s(\phi-k)}  &  H^*\left(f_s(\phi-k) \right) \nonumber \\
&  =  W \expphi S_{\hat{Y}} \expphi, \nonumber
\end{align}
or 
\[
W \expphi = \frac{ f_s \sum_{m} S_X\left(f_s(\phi-k) \right) e^{-2\pi i t f_s(\phi-k)} H^*\left(f_s(\phi-k) \right) }{S_{\hat{Y}} \expphi}.
\]
Note that the last expression equals the discrete-time Fourier transform with respect to $n$ of the function $\tilde{w}(t-n/f_s)$, where the impulse response of $\tilde{w}(t)$ is given by
\begin{equation} 
W(f) = \frac{ H^*(f) S_X(f) }{\sum_{k\in \mathbb Z} \left| H(f) \right|^2 S_X(f-f_sk)+ \sigma_\eta^2/f_s }.
\end{equation}
 \QEDA \\

\section*{ACKNOWLEDGMENT}
This work was supported in part by the NSF Center for Science of Information (CSoI) under grants CCF-0939370 and CCF-1320628, and by NSF-BSF grant 2015711. The work of Y. Eldar was funded by the European Union's Horizon 2020 research and innovation programme under grant agreement ERC-BNYQ and by the Israel Science Foundation under Grant no. 335/14. We thank Robert Gray for helpful discussions regarding the white noise approximation. 

%%%%%%%%%%%%%%%%%%%%%%%%%%%%%%%%%%%%%%%%%%%%%%%%%%%%%%%%%%%%%%%%%%%%%%%%%%%%%%%%
\bibliographystyle{IEEEtran}
\bibliography{IEEEfull,sampling}

% Generated by IEEEtran.bst, version: 1.14 (2015/08/26)
\begin{thebibliography}{10}
\providecommand{\url}[1]{#1}
\csname url@samestyle\endcsname
\providecommand{\newblock}{\relax}
\providecommand{\bibinfo}[2]{#2}
\providecommand{\BIBentrySTDinterwordspacing}{\spaceskip=0pt\relax}
\providecommand{\BIBentryALTinterwordstretchfactor}{4}
\providecommand{\BIBentryALTinterwordspacing}{\spaceskip=\fontdimen2\font plus
\BIBentryALTinterwordstretchfactor\fontdimen3\font minus
  \fontdimen4\font\relax}
\providecommand{\BIBforeignlanguage}[2]{{%
\expandafter\ifx\csname l@#1\endcsname\relax
\typeout{** WARNING: IEEEtran.bst: No hyphenation pattern has been}%
\typeout{** loaded for the language `#1'. Using the pattern for}%
\typeout{** the default language instead.}%
\else
\language=\csname l@#1\endcsname
\fi
#2}}
\providecommand{\BIBdecl}{\relax}
\BIBdecl

\bibitem{KipnisAllerton2014}
A.~Kipnis, A.~J. Goldsmith, and Y.~C. Eldar, ``{G}aussian distortion-rate
  function under sub-{N}yquist nonuniform sampling,'' in \emph{Communication,
  Control, and Computing (Allerton), 2014 52nd Annual Allerton Conference
  on}.\hskip 1em plus 0.5em minus 0.4em\relax IEEE, 2014, pp. 874--880.

\bibitem{KipnisAllerton2015}
A.~Kipnis, Y.~C. Eldar, and A.~J. Goldsmith, ``Optimal trade-off between
  sampling rate and quantization precision in {A/D} conversion,'' in \emph{53th
  Annual Allerton Conference on Communication, Control, and Computing
  (Allerton)}.\hskip 1em plus 0.5em minus 0.4em\relax IEEE, 2015.

\bibitem{KipnisITW2015}
A.~Kipnis, A.~J. Goldsmith, and Y.~C. Eldar, ``Sub-{N}yquist sampling achieves
  optimal rate-distortion,'' in \emph{Information Theory Workshop (ITW), 2015
  IEEE}, April 2015, pp. 1--5.

\bibitem{eldar2015sampling}
Y.~C. Eldar, \emph{Sampling Theory: Beyond Bandlimited Systems}.\hskip 1em plus
  0.5em minus 0.4em\relax Cambridge University Press, 2015.

\bibitem{Shannon1948}
C.~E. Shannon, ``A mathematical theory of communication,'' \emph{Bell System
  Tech. J.}, vol.~27, pp. 379--423, 623--656, 1948.

\bibitem{neuhoff2013information}
D.~L. Neuhoff and S.~S. Pradhan, ``Information rates of densely sampled data:
  Distributed vector quantization and scalar quantization with transforms for
  {G}aussian sources,'' \emph{{IEEE} Transactions on Information Theory},
  vol.~59, no.~9, pp. 5641--5664, 2013.

\bibitem{761034}
R.~Walden, ``Analog-to-digital converter survey and analysis,'' \emph{Selected
  Areas in Communications, IEEE Journal on}, vol.~17, no.~4, pp. 539--550, Apr
  1999.

\bibitem{EldarMichaeliBeyond}
Y.~C. Eldar and T.~Michaeli, ``Beyond bandlimited sampling,'' \emph{{IEEE}
  Signal Processing Magazine}, vol.~26, no.~3, pp. 48--68, 2009.

\bibitem{Kipnis2014}
A.~Kipnis, A.~J. Goldsmith, Y.~C. Eldar, and T.~Weissman, ``Distortion rate
  function of sub-{N}yquist sampled {G}aussian sources,'' \emph{{IEEE}
  Transactions on Information Theory}, vol.~62, no.~1, pp. 401--429, Jan 2016.

\bibitem{950786}
R.~Venkataramani and Y.~Bresler, ``Optimal sub-nyquist nonuniform sampling and
  reconstruction for multiband signals,'' \emph{IEEE Transactions on Signal
  Processing}, vol.~49, no.~10, pp. 2301--2313, Oct 2001.

\bibitem{daubechies1998factoring}
I.~Daubechies and W.~Sweldens, ``Factoring wavelet transforms into lifting
  steps,'' \emph{Journal of Fourier analysis and applications}, vol.~4, no.~3,
  pp. 247--269, 1998.

\bibitem{1447892}
H.~Landau, ``Sampling, data transmission, and the {N}yquist rate,''
  \emph{Proceedings of the {IEEE}}, vol.~55, no.~10, pp. 1701--1706, Oct 1967.

\bibitem{1057738}
R.~Dobrushin and B.~Tsybakov, ``Information transmission with additional
  noise,'' \emph{{IRE} Transactions on Information Theory}, vol.~8, no.~5, pp.
  293--304, 1962.

\bibitem{berger1971rate}
T.~Berger, \emph{Rate-distortion theory: A mathematical basis for data
  compression}.\hskip 1em plus 0.5em minus 0.4em\relax Englewood Cliffs, NJ:
  Prentice-Hall, 1971.

\bibitem{donoho1998data}
D.~L. Donoho, M.~Vetterli, R.~A. DeVore, and I.~Daubechies, ``Data compression
  and harmonic analysis,'' \emph{IEEE Transactions on Information Theory},
  vol.~44, no.~6, pp. 2435--2476, 1998.

\bibitem{1056823}
A.~Kolmogorov, ``On the shannon theory of information transmission in the case
  of continuous signals,'' \emph{{IRE} Transactions on Information Theory},
  vol.~2, no.~4, pp. 102--108, December 1956.

\bibitem{dodson1985fourier}
M.~Dodson and A.~Silva, ``Fourier analysis and the sampling theorem,'' in
  \emph{Proceedings of the Royal Irish Academy. Section A: Mathematical and
  Physical Sciences}.\hskip 1em plus 0.5em minus 0.4em\relax JSTOR, 1985, pp.
  81--108.

\bibitem{Lloyd1959}
S.~P. Lloyd, ``\BIBforeignlanguage{English}{A sampling theorem for stationary
  (wide sense) stochastic processes},''
  \emph{\BIBforeignlanguage{English}{Transactions of the American Mathematical
  Society}}, vol.~92, no.~1, pp. pp. 1--12, 1959.

\bibitem{1057404}
A.~Balakrishnan, ``A note on the sampling principle for continuous signals,''
  \emph{{IRE} Transactions on Information Theory}, vol.~3, no.~2, pp. 143--146,
  June 1957.

\bibitem{1090615}
D.~Chan and R.~Donaldson, ``Optimum pre-and postfiltering of sampled signals
  with application to pulse modulation and data compression systems,''
  \emph{{IEEE} Transactions on Communication Technology}, vol.~19, no.~2, pp.
  141--157, April 1971.

\bibitem{815501}
M.~Matthews, ``On the linear minimum-mean-squared-error estimation of an
  undersampled wide-sense stationary random process,'' \emph{{IEEE}
  Transactions on Signal Processing}, vol.~48, no.~1, pp. 272--275, 2000.

\bibitem{beutler1961sampling}
F.~J. Beutler, ``Sampling theorems and bases in a hilbert space,''
  \emph{Information and Control}, vol.~4, no. 2-3, pp. 97--117, 1961.

\bibitem{dym1978gaussian}
H.~Dym and H.~McKean, ``{G}aussian processes, function theory, and the inverse
  spectral problem,'' \emph{Bull. Amer. Math. Soc. 84 (1978), 260-262}, pp.
  0002--9904, 1978.

\bibitem{FEICHTINGER1992530}
H.~G. Feichtinger and K.~Gr{\"{o}}chenig, ``Irregular sampling theorems and
  series expansions of band-limited functions,'' \emph{Journal of Mathematical
  Analysis and Applications}, vol. 167, no.~2, pp. 530 -- 556, 1992.

\bibitem{beurling1989collected}
A.~Beurling and L.~Carleson, \emph{The collected works of Arne Beurling:
  Complex analysis}.\hskip 1em plus 0.5em minus 0.4em\relax Birkhauser, 1989,
  vol.~1.

\bibitem{Landau1967}
H.~Landau, ``\BIBforeignlanguage{English}{Necessary density conditions for
  sampling and interpolation of certain entire functions},''
  \emph{\BIBforeignlanguage{English}{Acta Mathematica}}, vol. 117, no.~1, pp.
  37--52, 1967.

\bibitem{john1996sampling}
J.~R. Higgins, \emph{Sampling theory in Fourier and signal analysis:
  foundations}.\hskip 1em plus 0.5em minus 0.4em\relax Oxford University Press
  on Demand, 1996.

\bibitem{marvasti2012nonuniform}
F.~Marvasti, \emph{Nonuniform sampling: theory and practice}.\hskip 1em plus
  0.5em minus 0.4em\relax Springer Science \& Business Media, 2012.

\bibitem{unser1994general}
M.~Unser and A.~Aldroubi, ``A general sampling theory for nonideal acquisition
  devices,'' \emph{IEEE Transactions on Signal Processing}, vol.~42, no.~11,
  pp. 2915--2925, 1994.

\bibitem{Pinsker1954}
M.~S. Pinsker, ``Computation of the message rate of a stationary random process
  and the capacity of a stationary channel,'' \emph{Dokl. Akad. Nauk. USSR},
  vol. 111, pp. 753--766, 1956.

\bibitem{1056251}
H.~Witsenhausen, ``Indirect rate distortion problems,'' \emph{{IEEE}
  Transactions on Information Theory}, vol.~26, no.~5, pp. 518--521, 1980.

\bibitem{1054469}
J.~Wolf and J.~Ziv, ``Transmission of noisy information to a noisy receiver
  with minimum distortion,'' \emph{{IEEE} Transactions on Information Theory},
  vol.~16, no.~4, pp. 406--411, 1970.

\bibitem{KipnisRini2015}
A.~Kipnis, S.~Rini, and A.~J. Goldsmith, ``The indirect rate-distortion
  function of a binary i.i.d source,'' in \emph{Information Theory Workshop -
  Fall (ITW), 2015 IEEE}, Oct 2015, pp. 352--356.

\bibitem{6573236}
X.~Liu, O.~Simeone, and E.~Erkip, ``Lossy computing of correlated sources with
  fractional sampling,'' \emph{IEEE Transactions on Communications}, vol.~61,
  no.~9, pp. 3685--3696, September 2013.

\bibitem{370112}
R.~Zamir and M.~Feder, ``Rate-distortion performance in coding bandlimited
  sources by sampling and dithered quantization,'' \emph{{IEEE} Transactions on
  Information Theory}, vol.~41, no.~1, pp. 141--154, Jan 1995.

\bibitem{marco2010entropy}
D.~Marco and D.~L. Neuhoff, ``Entropy of highly correlated quantized data,''
  \emph{IEEE Transactions on Information Theory}, vol.~56, no.~5, pp.
  2455--2478, 2010.

\bibitem{kumar2011high}
A.~Kumar, P.~Ishwar, and K.~Ramchandran, ``High-resolution distributed sampling
  of bandlimited fields with low-precision sensors,'' \emph{{IEEE} Transactions
  on Information Theory}, vol.~57, no.~1, pp. 476--492, 2011.

\bibitem{eldar2012compressed}
Y.~C. Eldar and G.~Kutyniok, \emph{Compressed sensing: theory and
  applications}.\hskip 1em plus 0.5em minus 0.4em\relax Cambridge University
  Press, 2012.

\bibitem{kipnis2017fundamental}
A.~Kipnis, G.~Reeves, Y.~C. Eldar, and A.~J. Goldsmith, ``Fundamental limits of
  compressed sensing under optimal quantization,'' in \emph{Information Theory
  (ISIT), 2017 IEEE International Symposium on}, 2017.

\bibitem{wu2011optimal}
Y.~Wu and S.~Verdu, ``Optimal phase transitions in compressed sensing,''
  \emph{{IEEE} Transactions on Information Theory}, vol.~58, no.~10, pp.
  6241--6263, Oct 2012.

\bibitem{KipnisCS}
A.~Kipnis, G.~Reeves, and Y.~C. Eldar, ``Fundamental limits of compresses
  sensing under quantization,'' 2018, in preparation.

\bibitem{boda2017sampling}
V.~P. Boda and P.~Narayan, ``Sampling rate distortion,'' \emph{{IEEE}
  Transactions on Information Theory}, vol.~63, no.~1, pp. 563--574, 2017.

\bibitem{gray2011entropy}
R.~M. Gray, \emph{Entropy and information theory}.\hskip 1em plus 0.5em minus
  0.4em\relax Springer, 2011, vol.~1.

\bibitem{zemanian1965distribution}
A.~H. Zemanian, \emph{Distribution theory and transform analysis: an
  introduction to generalized functions, with applications}.\hskip 1em plus
  0.5em minus 0.4em\relax Courier Corporation, 1965.

\bibitem{GS2_english}
I.~Gelfand and G.~Shilov, \emph{{Generalized functions. Volume 2}}.\hskip 1em
  plus 0.5em minus 0.4em\relax Academic Press, 1968.

\bibitem{gelfand1964generalized}
I.~M. Gelfand, N.~{\^A}. Vilenkin, and A.~Feinstein, \emph{Generalized
  functions. Vol. 4.}\hskip 1em plus 0.5em minus 0.4em\relax Academic press,
  1964.

\bibitem{ShannonMeetsNyquist}
Y.~Chen, Y.~C. Eldar, and A.~J. Goldsmith, ``Shannon meets {N}yquist: Capacity
  of sampled {G}aussian channels,'' \emph{{IEEE} Transactions on Information
  Theory}, vol.~59, no.~8, pp. 4889--4914, 2013.

\bibitem{black1947pulse}
H.~S. Black and J.~Edson, ``Pulse code modulation,'' \emph{Transactions of the
  American Institute of Electrical Engineers}, vol.~66, no.~1, pp. 895--899,
  1947.

\bibitem{1697556}
B.~Oliver, J.~Pierce, and C.~Shannon, ``The philosophy of {PCM},'' \emph{{IRE}
  Transactions on Information Theory}, vol.~36, no.~11, pp. 1324--1331, Nov
  1948.

\bibitem{gray1998quantization}
R.~Gray and D.~Neuhoff, ``Quantization,'' \emph{{IEEE} Transactions on
  Information Theory}, vol.~44, no.~6, pp. 2325--2383, Oct 1998.

\bibitem{930935}
H.~Viswanathan and R.~Zamir, ``On the whiteness of high-resolution quantization
  errors,'' \emph{{IEEE} Transactions on Information Theory}, vol.~47, no.~5,
  pp. 2029--2038, Jul 2001.

\bibitem{1056489}
S.~Lloyd, ``Least squares quantization in {PCM},'' \emph{{IEEE} Transactions on
  Information Theory}, vol.~28, no.~2, pp. 129--137, Mar 1982.

\bibitem{5672380}
J.~de~la Rosa, ``Sigma-delta modulators: Tutorial overview, design guide, and
  state-of-the-art survey,'' \emph{{IEEE} Transactions on Circuits and
  Systems}, vol.~58, no.~1, pp. 1--21, Jan 2011.

\bibitem{goyal2008compressive}
V.~K. Goyal, A.~K. Fletcher, and S.~Rangan, ``Compressive sampling and lossy
  compression,'' \emph{{IEEE} Signal Processing Magazine}, vol.~25, no.~2, pp.
  48--56, 2008.

\bibitem{1086334}
B.~Widrow, ``A study of rough amplitude quantization by means of {N}yquist
  sampling theory,'' \emph{Circuit Theory, IRE Transactions on}, vol.~3, no.~4,
  pp. 266--276, Dec 1956.

\bibitem{BLTJ1340}
W.~R. Bennett, ``Spectra of quantized signals,'' \emph{Bell System Technical
  Journal}, vol.~27, no.~3, pp. 446--472, 1948.

\bibitem{bennett1958statistics}
W.~Bennett, ``Statistics of regenerative digital transmission,'' \emph{Bell
  Labs Technical Journal}, vol.~37, no.~6, pp. 1501--1542, 1958.

\bibitem{Nedoma}
N.~J., ``On the ergodicity and r-ergodicity of stationary probability
  measures,'' \emph{Z. Wahrsch. Verw. Gebiete}, no.~4, pp. 2:90--97, 1963.

\bibitem{4663942}
T.~Michaeli and Y.~C. Eldar, ``High-rate interpolation of random signals from
  nonideal samples,'' \emph{{IEEE} Transactions on Signal Processing}, vol.~57,
  no.~3, pp. 977--992, 2009.

\bibitem{gray2009probability}
R.~Gray, \emph{Probability, Random Processes, and Ergodic Properties}.\hskip
  1em plus 0.5em minus 0.4em\relax Springer, 2009.

\bibitem{royden1988real}
H.~L. Royden and P.~Fitzpatrick, \emph{Real analysis}.\hskip 1em plus 0.5em
  minus 0.4em\relax Macmillan New York, 1988, vol. 198, no.~8.

\bibitem{KipnisThesis}
A.~Kipnis, ``Fundamental distortion limits of analog-to-digital compression,''
  Ph.D. dissertation, Stanford University, 2018.

\end{thebibliography}

\end{document}